\documentclass[preview]{elsarticle}

\usepackage{ifthen}
\usepackage{microtype}
\usepackage{amssymb,amsmath}
\usepackage{amsthm,thmtools,thm-restate}
\usepackage{mathtools}

\declaretheorem{theorem}
\declaretheorem[sibling=theorem]{lemma}
\declaretheorem[sibling=theorem]{corollary}
\declaretheorem[sibling=theorem]{example}
\declaretheorem[sibling=theorem]{remark}

\declaretheorem[sibling=theorem]{proposition}

\newcommand{\cupdot}{\mathbin{\mathaccent\cdot\cup}}

\usepackage{xspace}
\usepackage{tikz}
\usepackage[draft]{fixme}
\usepackage{marvosym}
\usepackage{todonotes}
\usepackage{wrapfig}
\usepackage{enumerate}

\usetikzlibrary{shapes}
\usetikzlibrary{backgrounds}
\usetikzlibrary{decorations.pathmorphing}
\usetikzlibrary{decorations.pathreplacing}
\usetikzlibrary{calc}

\newcommand{\ie}{i.e.,\xspace}
\newcommand{\0}{\emptyset}
\newcommand{\sq}{\subseteq}

\newcommand{\DAGw}{DAG\mbox{}-width\xspace}
\newcommand{\Dagw}{DAG\mbox{}-width\xspace}
\newcommand{\dagw}{DAG\mbox{}-width\xspace}
\newcommand{\nmdagw}{non-{}monotone DAG\mbox{}-width\xspace}
\newcommand{\treew}{tree\mbox{}-width\xspace}
\newcommand{\pathw}{path\mbox{}-width\xspace}
\newcommand{\kellyw}{Kelly\mbox{}-width\xspace}
\newcommand{\dpathw}{directed path\mbox{}-width\xspace}
\newcommand{\Dpathw}{Directed path\mbox{}-width\xspace}

\newcommand{\DAGgame}{\dagw game\xspace}

\newcommand{\dpwgame}{\dpathw game\xspace}

\newcommand{\ol}[1]{{\overline{#1}}}
\newcommand{\prefix}{\mathbin{\sqsubseteq}}
\newcommand{\propprefix}{\mathbin{\sqsubset}}

\newcommand{\topordereq}{\mathrel{\tikz \path node[inner sep=0cm] %
at (0,0){$\rightsquigarrow$}%
node[inner sep=0cm] at(-0.03,-0.1){\hspace{0.2em}\underline{\hspace{0.6em}}};}}

\renewcommand{\epsilon}{\varepsilon}
\renewcommand{\phi}{\varphi}
\renewcommand{\theta}{\vartheta}

\renewcommand{\epsilon}{\varepsilon}
\renewcommand{\phi}{\varphi}
\renewcommand{\theta}{\vartheta}

\DeclareMathOperator{\last}{last}
\DeclareMathOperator{\Post}{succ}
\DeclareMathOperator{\Reach}{Reach}
\DeclareMathOperator{\Next}{Next}
\DeclareMathOperator{\col}{\Omega}
\DeclareMathOperator{\act}{act}

\DeclareMathOperator{\tw}{tw}

\DeclareMathOperator{\dpw}{dpw}
\DeclareMathOperator{\dw}{dagw}
\DeclareMathOperator{\nmdw}{nm\text{-}dagw}
\DeclareMathOperator{\ent}{ent}

\makeatletter
\newcommand{\@abbrev}[3]{
  \def\c@a@def##1{
      \if ##1.
        \relax
      \else
        \@ifdefinable{\@nameuse{#1##1}}{\@namedef{#1##1}{#2##1}}
        \expandafter\c@a@def
      \fi
    }
  \c@a@def #3.
}
\@abbrev{bb}{\mathbb}{ABCDEFGHIJKLMNOPQRSTUVWXYZ}
\@abbrev{bf}{\mathbf}{ABCDEFGHIJKLMNOPQRSTUVWXYZabcdefghijklmnopqrstuvwxyz}
\@abbrev{bit}{\boldsymbol}{ABCDEFGHIJKLMNOPQRSTUVWXYZabcdefghijklmnopqrstuvwxyz}
\@abbrev{cal}{\mathcal}{ABCDEFGHIJKLMNOPQRSTUVWXYZ}
\@abbrev{frak}{\mathfrak}{ABCDEFGHIJKLMNOPQRSTUVWXYZabcdefghijklmnopqrstuvwxyz}
\@abbrev{rm}{\mathrm}{ABCDEFGHIJKLMNOPQRSTUVWXYZabcdefghijklmnopqrstuvwxyz}
\@abbrev{scr}{\mathscr}{ABCDEFGHIJKLMNOPQRSTUVWXYZ}
\@abbrev{sf}{\mathsf}{ABCDEFGHIJKLMNOPQRSTUVWXYZabcdefghijklmnopqrstuvwxyz}
\@abbrev{Alg}{\Alg}{ABCDEFGHIJKLMNOPQRSTUVWXYZ}
\@abbrev{Str}{\Str}{ABCDEFGHIJKLMNOPQRSTUVWXYZ}
\@abbrev{set}{\mathbb}{ABCDEFGHIJKLMNOPQRSTUVWXYZ}
\@abbrev{tup}{\tup}{ABCDEFGHIJKLMNOPQRSTUVWXYZabcdefghijklmnopqrstuvwxyz}
\@abbrev{ol}{\overline}{ABCDEFGHIJKLMNOPQRSTUVWXYZabcdefghijklmnopqrstuvwxyz}
\@abbrev{class}{\mathsf}{ABCDEFGHIJKLMNOPQRSTUVWXYZ}

\makeatother

\def\mathbi#1{\textbf{\em #1}}

\DeclareMathOperator{\init}{init}
\DeclareMathOperator{\upd}{upd}
\DeclareMathAlphabet{\mathsc}{OT1}{cmr}{m}{sc}
\newcommand{\pspace}{\ensuremath{\mathsc{Pspace}}\xspace}
\newcommand{\exptime}{\ensuremath{\mathsc{Exptime}}\xspace}
\newcommand{\ptime}{\ensuremath{\mathsc{Ptime}}\xspace}

\newcommand{\aptime}{\ensuremath{\mathsc{APtime}}\xspace}
\newcommand{\apspace}{\ensuremath{\mathsc{APspace}}\xspace}
\newcommand{\aspace}{\ensuremath{\mathsc{ASpace}}\xspace}
\newcommand{\atime}{\ensuremath{\mathsc{ATime}}\xspace}
\newcommand{\Time}{\ensuremath{\mathsc{Time}}\xspace}
\newcommand{\Space}{\ensuremath{\mathsc{Space}}\xspace}

\DeclareMathOperator{\front}{front}

\newdimen\arrowsize
\newlength{\arrowlength}
\newlength{\arrowangle}
\newlength{\arrowthickness}

\pgfarrowsdeclare{slim}{slim}
{
\arrowsize=0.7pt
\advance\arrowsize by .5\pgflinewidth
\pgfarrowsleftextend{-4\arrowsize-.5\pgflinewidth}
\pgfarrowsrightextend{.5\pgflinewidth}
}
{
\advance\arrowsize by .5\pgflinewidth
\pgfsetdash{}{0pt} 
\pgfsetroundjoin 
\pgfsetroundcap 
\pgfpathmoveto{\pgfpointpolar{\arrowangle}{\arrowlength}}
\pgfpathlineto{\pgfpointorigin}
\pgfusepathqstroke
\pgfpathmoveto{\pgfpointpolar{\arrowangle}{\arrowlength}}
\pgfpathcurveto{\pgfpointpolar{\arrowangle*1.04}{\arrowthickness+0.5*(\arrowlength-\arrowthickness)}}%
{\pgfpointpolar{\arrowangle*1.1}{\arrowthickness+0.3*(\arrowlength-\arrowthickness)}}
{\pgfpointpolar{180}{\arrowthickness}}
\pgfpathlineto{\pgfpointorigin}
\pgfusepathqfill
\pgfpathmoveto{\pgfpointpolar{-\arrowangle}{\arrowlength}}
\pgfpathlineto{\pgfpointorigin}
\pgfusepathqstroke
\pgfpathmoveto{\pgfpointpolar{-\arrowangle}{\arrowlength}}
\pgfpathcurveto{\pgfpointpolar{-\arrowangle*1.04}{\arrowthickness+0.5*(\arrowlength-\arrowthickness)}}%
{\pgfpointpolar{-\arrowangle*1.1}{\arrowthickness+0.3*(\arrowlength-\arrowthickness)}}
{\pgfpointpolar{180}{\arrowthickness}}
\pgfpathlineto{\pgfpointorigin}
\pgfusepathqfill
}

\tikzset{>=latex}

\tikzstyle{vertex}=[circle,inner sep=2.5,minimum size =2mm,semithick,fill=black!20, draw=black]
\tikzstyle{smallcircle}=[circle,inner sep=1.5,fill=white, draw=black]
\tikzstyle{point}=[circle,inner sep=1,fill=black, draw=black]
\tikzstyle{path}=[-slim,thin,rounded corners] 
\tikzstyle{path1}=[-slim,thin,decorate,%
decoration={snake,amplitude=6pt,segment length=#1,%
pre length=#1/20,post length=#1/20}]
\tikzstyle{brace}=[thin,decorate,decoration=brace]
\tikzstyle{ie}=[thin,dashed,gray]

\tikzstyle{cop}=[star, star points=6, star point height=3.2,inner sep=3,minimum size =2mm,semithick,fill=gray!80, draw=black]
\tikzstyle{ocop}=[star, star points=6, star point height=3.2,inner sep=3,minimum size =2mm,semithick,fill=gray!40,draw=black,densely dotted]
\tikzstyle{robber}=[circle, inner sep=2,minimum size =2mm,semithick, draw=black]
\tikzstyle{sep}=[circle,inner sep=0.1 pt,minimum size = 0.1mm]
\tikzstyle{player1}=[rectangle,draw,thick,inner sep=0pt,minimum size = 5mm]
\tikzstyle{post}=[->, semithick]
\tikzstyle{ind}=[dotted, semithick]

\newcommand{\nmon}{non-\mbox{}monotone\xspace}

\newcommand{\mb}{\mbox{}}

\newcommand{\olpi}{\overline{\pi}}

\newcommand{\olcalG}{\overline{\mathcal{G}}}
\newcommand{\olgame}{\overline{\mathcal{G}}}
\newcommand{\olpisr}{\ensuremath{\overline{\pi}_{fg}}}
\newcommand{\pisr}{\ensuremath{\pi_{fg}}}

\newcommand{\rf}{\ensuremath{\otimes_r f}}

\newcommand{\game}{\calG}
\newcommand{\graphG}{G}
\newcommand{\graphT}{T}
\newcommand{\graphK}{K}
\newcommand{\olgraphG}{\overline{G}}

\newcommand{\perpcdot}{{\perp}{\kern0.1em{\cdot}\kern0.1em}}

\newcommand{\wLOG}{without loss of generality\xspace}
\newcommand{\WLOG}{Without loss of generality\xspace}

\newcommand{\nix}{\ensuremath{\text{--}}\xspace}
\newcommand{\Paths}{\ensuremath{\mathrm{Paths}}\xspace}
\newcommand{\accept}{\texttt{accept}\xspace}
\newcommand{\reject}{\texttt{reject}\xspace}

\newcommand{\arrow}[1]{\stackrel{#1}{\longrightarrow}}
\newcommand{\dt}{\mathrm{det}}
\newcommand{\acc}{\mathrm{acc}}
\newcommand{\rej}{\mathrm{rej}}
\newcommand{\darrow}[1]{\stackrel{#1}{\longleftrightarrow}}

\newcommand{\olzeta}{\overline{\zeta}}

\DeclareMathOperator{\mincol}{mincol}
\DeclareMathOperator{\opt}{opt}
\DeclareMathOperator{\best}{best}
\DeclareMathOperator{\Winner}{Winner}
\DeclareMathOperator{\Hist}{Hist}
\DeclareMathOperator{\Simulate}{Simulate}
\DeclareMathOperator{\Profile}{Profile}

\DeclareMathOperator{\Update}{Update}
\DeclareMathOperator{\topomin}{topomin}
\DeclareMathOperator{\flap}{C}

\begin{document}
\begin{frontmatter}
\title{Parity Games, Imperfect Information and Structural Complexity\tnoteref{g}}
\tnotetext[g]{This work was supported by the projects \textit{Games for Analysis and
Synthesis of Interactive
Computational Systems (GASICS)} and \textit{Logic for Interaction (LINT)} of the
\textit{European Science
Foundation}.}

\author[bp]{Bernd~Puchala}\ead{puchala@logic.rwth-aachen.de}
\author[rr]{Roman~Rabinovich}\ead{roman.rabinovich@tu-berlin.de}
\address[bp]{Mathematical Foundations of Computer Science, RWTH Aachen University}
\address[rr]{Logic and Semantics, Technical University Berlin}




\begin{abstract}

  We address the problem of solving parity games with imperfect
  information on finite graphs of bounded structural complexity.  It
  is a major open problem whether parity games with \emph{perfect}
  information can be solved in \ptime. Restricting the structural
  complexity of the game arenas, however, often leads to efficient
  algorithms for parity games. Such results are known for graph
  classes of bounded \treew, DAG-width, \dpathw, and entanglement,
  which we describe in terms of cops and robber games. Conversely, the
  introduction of imperfect information makes the problem more
  difficult, it becomes \exptime-hard. We analyse the interaction of
 both approaches.

  We use a simple method to measure the amount of ``unawareness'' of
  a player, the amount of imperfect information. It turns out that if it
  is unbounded, low structural complexity does not make the problem
  simpler. It remains \exptime-{}hard or \pspace-{}hard even on very
  simple graphs. 

For games with bounded imperfect information we analyse 
the powerset construction, which is commonly used to convert a game of
imperfect information into an equivalent game with perfect
information. This construction 
  preserves boundedness of directed path-{}width and DAG-width, but not of entanglement or of tree-{}width.  Hence, if directed
  path-width or DAG-width are bounded, parity games with bounded imperfect
  information can be solved in \ptime. 
For DAG-width we follow two approaches. One leads to a generalization of
the known fact that perfect information parity games are in \ptime if
DAG-width is bounded. We prove this theorem for \emph{non-monotone}
DAG-width. The other approach introduces a cops and \emph{robbers} game (with
multiple robbers) on directed graphs, considered in~\cite{RicherbyThi09} for
undirected graphs. We show a tight linear bound for the number of
additional cops needed to capture an additional robber. 
\end{abstract}

  
\end{frontmatter}

\section{Introduction}
Parity games play a key role in the theory of verification and
synthesis of state-\mbox{}based systems. They are the
model-\mbox{}checking games for the modal $\mu$-calculus, a powerful
specification formalism for verification problems. Moreover, parity
objectives can express all $\omega$-regular objectives and therefore
capture fundamental properties of non-terminating reactive systems,
cf. \cite{Thomas95}. Such a system can be modeled as a
two-\mbox{}player game (the players are called~$0$ and~$1$) where
changes of the system state correspond to changes of the game
position. Situations where the change of the system can be controlled
correspond to positions of Player~$0$, uncontrollable situations
correspond to positions of Player~$1$. A winning strategy for
Player~$0$ yields a controller that guarantees satisfaction of some
$\omega$-\mbox{}regular specification.

In a parity game, the players move a token along the edges of a
labeled graph by choosing appropriate edge labels, called actions. The
vertices of the graph, called positions, are labeled with natural
numbers and the winner of an infinite play of the game is determined
by the parity of the least color which occurs infinitely often.

The problem to determine, for a given parity game~$\game$ and a
position~$v$, whether Player~$0$ has a winning strategy for~$\game$
from~$v$, is called the strategy problem.  The algorithmic theory of
parity games with perfect information has received much attention
during the past years, cf.~\cite{Jurdzinski00}. 

However, assuming that both players have perfect information about the
history of events in a parity game is not always realistic. For
example, if the information about the system state is acquired by
imprecise sensors or the system encapsulates private states which
cannot be read from outside, then a controller for this system must
rely on the information about the state and the change of the system
to which it has access. A technique to solve the strategy problem in
presence of imperfect information is to track the knowledge of the
game of Player~$0$, thus reducing the problem to a strategy problem
for a game with perfect information on another
graph~\cite{Reif84}. This procedure is often referred to as \emph{powerset
construction} and we call the constructed graph the \emph{powerset graph.}

Such a knowledge tracking is inherently unavoidable and leads to an
exponential lower bound for the time complexity of the strategy
problem for reachability games with imperfect
information~\cite{Reif84} and a super-polynomial lower bound for the
memory needed to implement winning strategies in reachability
games~\cite{BCDHR08,Puchala08}.  

Our goal is to find interesting special cases of the problem that can
be solved in \ptime. A simple, yet effective, approach is to bound the
amount of uncertainty of Player~$0$. This is appropriate in situations where, e.g.,
the imprecision of the sensors or the amount of private information of
the system does not grow when the size of the system grows. Then the
game which results from the powerset construction has polynomial size,
so solving imperfect information parity games reduces to the strategy
problem for parity games with perfect information. However, it is not
known whether the latter problem can be solved efficiently, \ie in
\ptime, and the question whether this is possible remains one of the
most intriguing in game theory.

To obtain a class of parity games with imperfect information that we
can solve in \ptime, we thus have to bound certain other parameters. A
natural approach is to restrict the structural complexity of the game
graphs with respect to an appropriate measure. Several such measures
have proven to be very useful in algorithmic graph theory. Many
problems, including the strategy problem for perfect information
parity games which are intractable in general can be solved
efficiently on classes of graphs where such measures are bounded. It
has been shown that parity games played on graphs of bounded \treew,
\pathw, \dpathw, \dagw or entanglement can be solved in polynomial
time \cite{BerwangerDawHunKreObd12,BG04}.  A natural question is whether
these results can also be obtained for games with imperfect
information. For each of those complexity measures we answer two
questions about parity games on graphs of bounded complexity:
\begin{enumerate}
\item Are the games with (in general \emph{unbounded}) imperfect
  information solvable in \ptime?

\item Are the games with \emph{bounded} imperfect
  information solvable in \ptime?
\end{enumerate}
For two other important measures: directed \treew~\cite{JRST01} and
\kellyw~\cite{HK07} the problem remains open, for directed \treew even in the
case of perfect information.

\paragraph*{Organization and results.}
In Section~\ref{sec_prelim} we introduce the basic notions we use
throughout the paper. In Section~\ref{sec_upi} we consider unbounded
imperfect information. For all complexity measures we work with we
prove that there are classes of graphs $\graphG$ with complexity at
most two such that the size of the powerset graph and its complexity
are both exponential in the size of~$\graphG$. We further show that
the strategy problem even for simpler reachability games with
imperfect information is $\exptime$-\mbox{}hard on graphs with
entanglement and \dpathw at most two. On acyclic graphs, solving
reachability games turns out to be $\pspace$-{}complete. This shows
that bounding the structural complexity of graphs does not
substantially decrease the computational complexity of the strategy
problem, as long as the amount of imperfect information is unbounded.

In Section \ref{sec_bpi} we consider parity games with bounded
imperfect information.  In this case, the graphs which result from the
powerset construction have polynomial size. Thus if the construction
additionally preserves boundedness of appropriate graph complexity
measures, then the corresponding strategy problem is in~$\ptime$. We
obtain that the powerset construction, while preserving neither
boundedness of entanglement nor of \treew, does preserve boundedness
of \dpathw. The case of \dagw is much more involved. However, it is also more
interesting: \dagw is bounded in \dpathw, but not the other way
around. \Dagw (as well as the other measures) can be defined as a
graph searching game where a team of cops tries to capture a robber in
the given graph. The player move alternately. The cops occupy some
vertices and can change their placement arbitrarily in their move.
The robber runs between vertices along cop free paths. The \dagw of a
graph is the minimal number of cops needed to capture the robber in a
\emph{monotone} way, \ie such that the robber can never occupy a
vertex that has already been unavailable for him. The problem with
\dagw is now that while capturing the robber is preserved after the
applying the powerset construction, the monotonicity is not. For
entanglement the monotonicity is not needed, for \dpathw we obtain it
for free: if~$k$ cops capture the robber, then~$k$ cops can also do it
in a monotone way~\cite{Hunter06}.

We discuss three approaches to this problem. One of them fails giving
us an example (see Theorem~\ref{thm_OffhandedCops}) which at least
partially explains the difficulty with monotonicity for \dagw.  Two
other approaches lead to solutions of the problem. The first one is
presented in Section~\ref{sec_simulated_pg}. We prove that parity
games (with perfect information) can be solved efficiently not only if
\dagw is bounded, but also even if \emph{non-{}monotone} \dagw is
bounded. The idea of the proof is from the solution of the strategy
problem for parity games on graphs of bounded \dagw via
\emph{simulated games} by Fearnley and Schewe~\cite{FearnleySch12}. It
turns out that their construction can be used also for the case of
bounded non-{}monotone \dagw.  This relativizes the importance
of monotonicity for \dagw, as the strategy problem for parity games is
the only known to become easier if \dagw is bounded and 
not known to become easier when the more general directed \treew is
bounded.

The other approach, that we pursue in Section~\ref{sec-parity} is a
generalization of the graph searching game for \dagw to a game where
the cops have to capture \emph{multiple} robbers. The robbers
correspond to multiple plays of the parity game with imperfect
information that Player~$0$ considers to be possible in a
position. Thus if the amount of imperfect information is at most~$r$,
we consider the game with~$r$ robbers. The new game also generalizes a
similar game on undirected graphs from~\cite{RicherbyThi09} by
Richerby and Thilikos. Our setting is, however, different, which makes
our main result about the game in a sense more general (we discuss the
connection to the game from~\cite{RicherbyThi09} in
Section~\ref{sec-parity}). We prove that if~$k$ cops can capture~$r$
robbers, then $kr$ cops can capture~$r$ robbers. This is the
technically most involved proof, the main problem is again to preserve
monotonicity. However, this result allows us to preserve monotonicity
also for \dagw while translating a cop strategy from the game with
imperfect information to a game with perfect information. Thus we
establish a connection between imperfect information in parity games
and a multiagent graph searching game.  Interestingly, if the cops
have to capture infinitely many robbers, the game turns out to be
equivalent to the game that characterizes directed \pathw and is also
defined by means of imperfect information. This is the same situation
as in~\cite{RicherbyThi09} for the undirected case.

\section{Preliminaries}\label{sec_prelim}

We assume that the reader is familiar with basic notions from the
graph theory. All graphs in this work are directed, finite and without
multi\mb-edges. (An undirected graph is a graph with a symmetric edge
relation.) For sets $X \subseteq V$, $G-X$ denotes the subgraph of~$G$
induced by the vertices of~$G$ that are not in~$X$. By $\Reach_G(X)$
we denote the set of vertices reachable from~$X$ in~$G$. A strongly
connected component, or simply a \emph{component}, is a maximal subset
of the graph such that, from each vertex to each vertex, there is a
path in that subgraph. If~$U$ is a set of vertices in the
graph~$\graphG$ and $v\notin U$ is a vertex, then $\flap^\graphG_U(v)$
is the component of $\graphG - U$ containing~$v$\label{def_flap}. A
(directed) rooted tree is an orientation of an undirected tree where
all edges are oriented away from a designated vertex, the root. The
depth of a rooted tree is the number of vertices on its longest path.
For a finite sequence~$\pi$ of elements, $\last(\pi)$ denotes the last
element of~$\pi$. If~$v$ is a vertex and~$E$ the set of edges, then
$vE$ is $\{w \mid (v,w)\in E\}$. If $\sim$ is an equivalence relation,
we write $[v]_\sim$ or just $[v]$ for the equivalence class
of~$v$. The set of natural numbers is denoted by~$\omega$.

\subsection{Games}

We consider finite two-player zero-sum games with imperfect
information and perfect recall, \ie any play is won by either of the
players and both players never forget any information that has already
been available for them. The players are called Player~$0$ and
Player~$1$. Formally, a game \emph{arena} is a tuple $\calA =
(V,V_0,E)$ where $(V,E)$ is the \emph{game graph}, and $V_0\subseteq
V$ is the set of positions on which Player~$0$ has to move. Let $A$ be
a finite set of \emph{actions.} A \emph{game} is a tuple $\game =
(V,V_0,(E_a)_{a\in A},v_0,\sim,\Omega)$ where $(V,V_0,\bigcup_{a\in
  A}E_a)$ is an arena with $|vE_a| \le 1$ for all $v\in V$ and $a\in
A$. Thus all edges leaving the same vertex are uniquely labeled and
the player who moves at~$v$ determines the next position by choosing
one of those labels. Furthermore, $v_0\in V$ is the initial position,
and $\col\subseteq V^\omega$ is the \emph{winning condition} for
Player~$0$. For convenience, we define $V_1 = V\setminus V_0$ and $E =
\bigcup_{a\in A}E_a$. The game graph of~$\game$ is $G = (V,E)$. We
write $v \arrow{a} w$ if $(v,w)\in E_a$ and $v \darrow{a} w$ if $(v,
w)\in E_a$ and $(w,v)\in E_a$. For~$v\in V$, $\act(v) = \{a\in A \mid
vE_a\neq \0\}$.  A \emph{play} is a maximal finite or infinite
sequence $v_0a_0v_1a_1v_2a_2\ldots \in (VA)^*V \cup (VA)^\omega$ such
that $(v_i,v_{i+1})\in E_{a_i}$ for all $i\ge 0$. A finite play
$\pi=v_0a_0\ldots v_n$ is won by Player~$i\in\{0,1\}$ if and only if
$v_n\in V_{i-1}$ and $v_nE = \0$. An infinite play~$\pi$ is won by
Player~$0$ if and only if $\pi\in\col$, otherwise it is won by
Player~$1$.

Common winning conditions are \emph{reachability} (Player~$0$ wins a play if it 
reaches a vertex from a given set), \emph{safety} (Player~$0$ wins if the play never reaches
a given set of vertices), or \emph{parity} (the vertices are colored
by linearly ordered colors; 
Player~$0$ wins if the minimal infinitely often seen color is even).

A \emph{history} is a finite prefix $\pi$ of a play with
$\last(\pi)\in V$.  The set of all histories of a
game~$\game$ is~$\calH(\game)$. Now we can define the last component
of a game: $\sim$ is an equivalence relation on~$\calH(\game)$. For
$\pi,\pi'\in\calH(\game)$ we say that Player~$0$ cannot distinguish between
them if $\pi \sim \pi'$.

A \emph{strategy} for Player~$i$ is a partial function $g\colon
(V A)^* V_i \to A$ and if $i=0$, then~$g$ must be based only
on the information available for Player~$0$: if $\pi \sim \pi'$, then
$g(\pi) = g(\pi')$.  Let $\pi = v_0 a_0 v_1 a_1 v_2 \ldots $ be a
history or a play. We say that it is \emph{consistent} with~$g$ if for
all~$j$ with $v_j \in V_i$ we have $a_j = g(v_0 a_0 \ldots a_{j-1}
v_j)$. We call a strategy~$g$ for Player~$i$ \emph{winning} from $v_0$
if Player~$i$ wins every play~$\pi$ in~$\game$ from $v_0$ that is
consistent with~$g$. We are interested only in winning strategies for
Player~$0$, so we consider only games where Player~$1$ has perfect
information. If we introduced imperfect information for both players,
a non-{}winning strategy for Player~$0$ could exist even if there
were no winning counter-{}strategy for Player~$1$.


In order to speak about decision problems for games of imperfect
information we have to represent~$\sim$ in a finite way. For that we
consider equivalence relations $\sim^V \subseteq V^2$ and
$\sim^A\subseteq A^2$ on positions and on actions of the game,
respectively, and extend them to~$\sim$.  In this case we also write
$(V,V_0,(E_a)_{a\in A},\sim^v,\sim^A,\Omega)$ instead of
$(V,V_0,(E_a)_{a\in A},\sim,\Omega)$.  Relations $\sim^V$ and $\sim^A$
must satisfy the following conditions. For winning conditions defined
by a coloring of the arena vertices we abuse the notation and denote
by~$\col(v)$ the color of vertex~$v$.
\begin{enumerate}[1.]
\item If $u \sim^Vv$, then $u, v \in V_0$ or $u, v \notin V_0$ (Player~$0$ knows when it is his turn).
\item if for some $v\in V$, $a, b \in \act(v)$ and $a \neq b$, then $a \not \sim^A b$ (Player~$0$ distinguishes available actions).
\item if $u, v \in V_0$ with $u \sim^Vv$, then $\act(u) = \act(v)$ (Player~$0$ knows which actions are available).
\item if $u \sim^Vv$, then $\col(u) = \col(v)$ (game colors are observable for Player~$0$).
\end{enumerate}
The equivalence relation $\sim$ on histories is induced by $\sim^V$
and $\sim^A$ as follows.  For $\pi = v_0 a_0 \ldots a_{n-1} v_n$ and
$\pi' = w_0 b_0 \ldots b_{m-1} w_m \in V(AV)^*$, we have $\pi \sim
\pi'$ if and only if $n = m \text{ and } v_j \sim^V w_j \text{ and }
a_j \sim^A b_j \text{ for all } j\,.$

The \emph{winning region} of Player~$i$ in~$\game$ is the set of 
all positions $v \in V$ such that Player~$i$ has a winning strategy for~$\game$ from~$v$. 

We say that a class $\classC$ of games has \emph{bounded imperfect information}, if there is some 
$r \in \omega$ such that for every game~$\game = (V,V_0,(E_a)_{a\in A},v_0,\sim,\Omega)$ from $\classC$ and for
any position~$v \in V$, the equivalence class $[v]_{\sim^V} :=\{ w \in V \mid  v \sim^V w\}$ 
of~$v$ has size at most~$r$. Notice that the equivalence classes 
$[a]_{\sim^A} := \{ b \in A \mid  a \sim^A b\}$ of actions $a \in A$ may, however, be arbitrarily large.
 If $r=1$, we have a game of perfect information, in which case we omit the component~$\sim$ and the
actions, so a game with perfect information can be formalized as a
tuple $(V,V_0,E,v_0,\Omega)$.  Being in a positon $v\in V$, a player
choses an edge $(v,w)\in E$ and thus determines the next position~$w$.
In this case a play is defined in an
obvious way analogously to a play in the general case as a sequence of positions.


\subsection{Powerset Construction}\label{subsec_powerset}
A usual method to solve games with imperfect information is a powerset construction originally suggested
by John H. Reif in \cite{Reif84}. 
The construction turns a game with imperfect information into a \emph{non-\mbox{}deterministic} 
game with perfect information such that the existence of winning strategies for Player~$0$ is preserved.

A \emph{non-\mbox{}deterministic parity game} is defined as a deterministic game, but the condition $|vE_a| \le 1$ is dropped.
Plays, strategies and winning strategies are defined as before. In
particular, a strategy is winning for Player~$0$ if all plays
consistent with it are won by Player~$0$, regardsless which
non\mb-deterministic choices are made. In general, even finite non-\mbox{}deterministic
games are not determined (\ie neither of the players may have a
winning strategy) and hence not equivalent to deterministic games. 
However, for each non-\mbox{}deterministic game $\game$ and each player $i \in \{0,
1\}$, we can construct a deterministic game $\game^i$ such that the existence of
winning strategies for Player~$i$ is preserved. The non-\mbox{}determinism  can be resolved 
by giving player $1-i$ control of non-\mbox{}deterministic choices.
For any $v \in V$ and any $a \in \act(v)$ we add a unique
$a$-successor of $v$ to the game graph which belongs to player $1-i$ and from
which he can choose any $a$-successor of $v$ in the original game graph. The
color of such a new position is the color of its unique predecessor.

Formally, for a parity game $\game = (V, V_0, (E_a)_{a \in A}, v_0,
\sim,\col)$ where ${}\sim{}$ is defined by some $\sim^V$ and $\sim^A$,
we construct the \emph{powerset} game $\olgame = (\olV, \olV_0, (\olE_a)_{a \in A}, \olv_0, \overline{\col})$
with perfect information. Without loss of generality we always assume that
$[v_0] = \{v_0\}$.  For $S \sq V$ and $B \sq A$, let
$\Post_B(S) := \{ v \in V \mid   \text{ there are } s \in S \text{ and
} b \in B \text{ such that } b \in \act(s) \text{ and } v \in sE_b \}$.
The components of $\olgame$ are defined as follows:
\begin{itemize}

\item $\olV = \{ \olv \in 2^V \mid \olv \sq [u] \text{ for some } u\in
  V\}$ and $\olV_0 = \olV \cap 2^{V_0}$;

\item for all $a\in A$, $\olE_a = \{ (\olv,\olw) \mid  \olw =
  \Post_{[a]}(\olv) \cap [u] \text{ for some } u\in \Post_{[a]}(\olv)\}$;

\item $\olv_0 = \{v_0\}$;

\item $\col(\olv) = \col(v)$ for some $v \in \olv$ (note that colors
  are observable).

\end{itemize}

One can see that this construction preserves winning strategies for
Player~$0$.  We will always assume that the graph
game~$\olG$ of~$\olcalG$, the \emph{powerset graph}, is only the
part of the graph reachable from~$\{v_0\}$.  The following lemma,
whose proof is straightforward, states the key property for the
correctness of the construction.

\begin{lemma}\label{lemma_powerset}
For each history $\overline{\pi} = \overline{v}_0 a_1 \overline{v}_1 \ldots a_n \overline{v}_n$ in $\olgame$
and all $u_n \in \olv_n$, there is a history $\pi = u_0 a_1' u_1 \ldots a_n' u_n$ in $\game$
such that $u_i \in \olv_i$ and $a_i'\sim^A a_i$ for all $i$. 
\end{lemma}

\subsection{Graph searching games}\label{subsec_graph_searching_games_def}

In this section we introduce several measures for structural
complexity of graphs, which we define by means of graph searching
games. The actions play no role here, so we may assume that the edges
are not labeled and the players choose an outgoing edge to determine
their move.  Hereby Player~$0$ does not see which edge was chosen by
Player~$1$, he can only distinguish between positions. The games are
played by a robber and a team of~$k$ cops where~$k$ is a parameter of
the game. In a position, the robber occupies a vertex and each of the
cops either also occupies a vertex or is outside of the graph.  In a
move, the cops announce their next placement. Then the robber chooses
a new vertex that is reachable from his current vertex via paths that
do not contain any vertices occupied by cops. In the next position,
the robber is on his new vertex and the cops are placed as they have
announced. The cops try to capture the robber, \ie to reach a position
where he has no legal move. If they never capture him, the robber
wins.  Modifications of this basic game define a complexity measure of
a graph by the \emph{cop number}: the least number of cops needed to capture the
robber.\footnote{\DAGw, \treew and \dpathw are usually defined in
  terms of graph decompositions.}

\paragraph*{\DAGw} A \emph{\DAGgame} (or the cops and robber game)
$\game_k(G)$ is a game with
perfect information~\cite{BerwangerDawHunKreObd12}.  The game is
played on a directed graph $G=(V,E)$, which is different from the game
graph, by two players. 
Cop positions are of the form $(U,v)$ where
$U\subseteq V$ is the set of at most $k$ vertices occupied by cops (if $|U|<k$,
we say that the rest of the cops is outside of the graph) and $v\in V\setminus
U$ is the vertex occupied by the robber. Robber positions are of the form
$(U,U',v)$ where $U$ and $v$ are as before and $U'\subseteq V$ is the set of at
most $k$ vertices announced by the cops that will be occupied by them in the
next position. From a position $(U,v)$, the cops can move to a robber position
$(U,U',v)$. From a position $(U,U',v)$, the robber can move to a cop position
$(U',v') $ where $v'\in\Reach_{G-(U\cap U')}(v) \setminus U'$. In the first
move, the robber is placed on any vertex, \ie the first move is $\perp \to
(\emptyset,v)$ for any $v\in V$. Hereby $\perp$ is an additional dummy first
position of any play.

A play of a \DAGgame is \emph{(robber-)monotone} if
the robber cannot occupy any vertex that has been already unavailable
for him. 
Formally, the play contains no position $(U, U', v)$ such that some
$u \in U \setminus U'$ is reachable from~$v$ in $G - (U\cap U')$. A
finite play is won by cops if it is monotone. Non{}-monotone plays and infinite plays are won by
the robber.

For a graph~$G$, the least~$k$ such that the cops have a winning strategy for the
game~$\game_k(G)$ is the \emph{\DAGw} $\dw(G)$ of~$G$, defined
in~\cite{BDHK06,Obdrzalek06b}, see
also~\cite{BerwangerDawHunKreObd12}.  The \emph{non-{}monotone \DAGw} $\nmdw(G)$
is the same as \DAGw, but the requirement for the cops to guarantee
monotonicity is dropped. 
We define the \emph{\treew} $\tw(G)$ as $\dw(G^\leftrightarrow)-1$, where
the game is played on the graph
$G^\leftrightarrow=(V,E^\leftrightarrow)$ with $E^\leftrightarrow = \{(v,w) ~|~
(v,w)\in E\text{ or }(w,v)\in
E\}$.

\paragraph*{\Dpathw}
\Dpathw of a graph $G$ is the minimal number of cops minus one that
have a monotone winning strategy against an \emph{invisible} robber
on~$G$. This is a game with imperfect information for the cop player
where cop strategies are functions~$f$ that map sequences of cop
placements to a next placement: $f:(2^V)^*\to 2^V$. In other words,
the \emph{\dpwgame} or the \emph{cops and invisible robber game} is
defined as the cops and robber game, but now the equivalence relation
contains all pairs of positions. We can also define this game a
one{}-player perfect information game if we assume that the robber
occupies every vertex which is considered by the cops to be possibly
occupied. Let $G=(V,E)$ be a graph. Positions of the game have the
form $(U,U',R)$ where $|U|,|U'|\le k$ and $R\subseteq V$. The initial
position is $\perp$ and the next one is $(\0,\0,V)$. From a position
$(U,U',R)$ the cops can move to any position $(U',U'',R')$ where $R' =
\Reach_{G - (U\cap U')}(R) \setminus U''$. A play
$(U_0,U_0,R_0)(U_0,U_1,R_1)(U_1,U_2,R_2)\dots$ is monotone if $R_i$
are monotonically non-{}increasing. The cops win monotone finite plays, the
robber wins (\ie the cops lose) non-{}monotone plays and infinite
plays.  The \emph{\dpathw} of $G$ is the least number~$k$ such that
$k+1$ cops have a winning strategy on~$G$.

Obviously, $\dw(G) \leq \dpw(G)+1$ for any graph~$\graphG$. Moreover
the directed path-width of a graph is not bounded by
its DAG-width, that means, there is a class of directed graphs such
that the DAG-width is bounded and the directed path-width is unbounded
on this class.

\paragraph*{Entanglement} 
In the \emph{entanglement} game~\cite{BG04}, in each position, the
robber is on a vertex $r$ of the graph. In each round, the cop player
may do nothing or place a cop on $r$, either from outside the graph if
there are any cops left or from a vertex $v$ which was previously
occupied by a cop and is then freed.  No matter what the cops do, the
robber must go from his recent vertex $r$ to a new vertex~$r'$, which
is not occupied by a cop along an edge $(r,r') \in E$. If the robber
cannot move, he loses. So formally, the entanglement game on~$\graphG$
is a game with perfect information and a position of the entanglement
game on~$\graphG$ is a tuple $(U, r)$ if it is the cops' turn or a
tuple $(U,U',r)$ if it is
the robber's turn, with $U' = (U\setminus \{v\})\cup\{r\}$ for some
$v\in U$ (the cop is coming from $v$ to $r$) or $U' = U \cup \{r\}$ (a
new cop from outside is coming to $r$).  From $(U, r)$ the cops can
move to a position of the form $(U,U', r')$. On his turn, the robber
can move from $(U,U',r)$ to a position $(U',r')$ where $(r,r')\in E$
and $r'\not\in U'$.  The entanglement of a graph $G$, denoted
$\ent(\graphG)$ is the minimal number $k$ such that $k$ cops win the
entanglement game on~$\graphG$.

It is known that bounded entanglement implies bounded non-{}monotone
\dagw, but not vice versa~\cite{BerwangerDawHunKre06}. It is easy to see that bounded \dpathw implies
bounded \dagw and bounded non-{}monotone \dagw, but not vice versa.

\paragraph*{Using decompositions to solve parity games}
We will measure the complexity of a game by the complexity of its
underlying graph, so, e.g., if $\game = (V,V_0,(E_a)_{a\in A},v_0, \sim,
\Omega)$, then $\dw(\game) = \dw(V,\bigcup_{a\in A}E_a) = \dw(G)$. 

We defined DAG-width, tree-width and directed path-width in terms of
monotone winning strategies.  A monotone winning strategy for $k$ cops
on~$\graphG$ yields a decomposition of~$\graphG$ into parts of
size at most~$k$ which are only sparsely related among each
other. (The particular measure determines what ``sparsely'' precisely
means.) Such decompositions often allow for efficient dynamic
solutions of hard graph problems.

Entanglement is defined in terms of strategies which are not
necessarily monotone and a decomposition in the above sense is known
only for $k = 2$, see~\cite{GKR09}. Nevertheless, parity games can be
solved efficiently on graph classes of bounded entanglement.

\begin{theorem}[\cite{Obdrzalek03,BerwangerDawHunKreObd12,BerwangerGra05}]
Parity games can be solved in \ptime on classes of graphs of bounded
\treew, \dagw, (and hence \dpathw), or entanglement.
\end{theorem}

\paragraph*{Monotonicity costs}
In the following, let $\calM = \{\tw, \dw, \dpw, \ent\}$. We say that
a measure $X \in \calM$ has \emph{monotonicity costs} at most~$f$ for
a function $f : \omega \rightarrow \omega$ if, for any graph~$\graphG$
on which~$k$ cops have a winning strategy for the $X$-game on
$\graphG$, $k + f(k)$ cops have a monotone winning strategy for the
$X$-game on~$\graphG$. We say that~$X$ has \emph{bounded monotonicity
  costs} if there is a function $f : \omega \rightarrow \omega$ such
that~$X$ has monotonicity costs at most~$f$.  Tree-width has
monotonicity costs~$0$, see~\cite{ST93}, and the same holds for directed
path-width,~\cite{Barat06,Hunter06}.  On the contrary, DAG-width does
not have monotonicity costs~$0$: there is a class of graphs $\graphG_n$,
such that $3n-1$ cops have a winning strategy on $\graphG_n$, but $\dw(G_n)
= 4n-2$, see \cite{KO08}. Whether \dagw has bounded monotonicity costs, is
an open problem~\cite{BerwangerDawHunKreObd12,KreutzerOrd08}.



\section{Unbounded imperfect information}\label{sec_upi}

If imperfect information is unbounded, then the powerset
construction can produce a graph which is super-polynomially larger
than the original graph. Moreover, we show that the values of all measures we
consider become unbounded and super-{}polynomial in the size of the given graph.

Let $\calM = \{\tw, \dw, \dpw, \ent\}$ and let~$G_n$ be the undirected 
$n \times n$-grid $G_n = (V_n, E_n)$ with $V_n = \{ (i, j) \mid  1
\leq i, j \leq n\}$ and $\big((i_1, j_1),(i_2,j_2)\big)\in E$ if and
only if $|i_1-i_2| + |j_1-j_2| =1$.
We will need the well-{}known fact that, for any $n >1$, 
we have $X(G_n) \geq n$ for all $X \in \calM$.

\begin{proposition}\label{proposition_exponentialgrowth}
  There is a family of games $\game_n$ with imperfect information such
  that for all $X\in \calM$, $X(\game_n) \leq 2$, but $X(\olgame_n)$
  is super-{}polynomial in the size of $\game_n$, where $\olgame_n$ is
  the powerset graph of~$\game_n$.
\end{proposition}
\begin{proof}
From a very simple graph, we generate a graph containing an undirected square grid 
of super-{}polynomial size as a subgraph.
This is possible because we can consider large equivalence classes of positions and actions. 

Consider a disjoint union of~$n$ directed cycles of length~$2$ with self-loops on each vertex
where any two positions are equivalent. Additionally we have an
initial position such that, by applying the powerset construction from this position, we obtain a set which contains 
exactly one element from each cycle. Continuing the construction, we
obtain sets that represent binary numbers with~$n$ digits and for each
digit we have an action which causes exactly this digit to flip. So,
using the Gray-code, we can create all binary numbers with~$n$
digits by successively flipping each digit. If we do this
independently for the first $n/2$ digits and for the last $n/2$
digits, it is easy to see that the resulting positions are connected
in such a way, that they form an undirected grid 
$\olG_n$ of size $2^{n/2} \times 2^{n/2}$, for which we have
$X(\olG_n) \geq 2^{n/2}$ for any measure $X \in \mathcal{M}$.

To be more precise, for even $n \ge 2$, let $\game_n = (V_n, V_0 = \0,
(E_a^n)_{a\in A_n},\sim_n,\Omega)$ where $\Omega = \0$, $\sim_n$ is induced by $\sim^V_n$ and
$\sim^A_n$ (which we define below) and $G_n = (V_n, E_n= \bigcup_{a \in A_n} E_a^n)$ is the
following game graph. 
The set of vertices is
$\{v_0\} \cup \{ (i, j) \mid 1 \leq i \leq n, j \in \{0,1\}\}$
where~$i$ denotes the number of the cycle and~$j$ is the number of
a vertex in the cycle.
The actions are $A_n = \{ a_i \mid  1 \leq i \leq n\} \cup \{\neg_i \mid  1 \leq i \leq n\}$.
Here the actions $a_i$ lead from $v_0$ to the cycles: $v_0 \arrow{a_i} (0, i)$ for $1 \leq i \leq n$.
Further actions build the cycles:

\begin{itemize}
\item $(i, j) \arrow{\neg_i} (i, 1-j)$ for $1 \leq i \leq n$ and $j \in \{0, 1\}$.

\item $(i, j) \arrow{\neg_k} (i, j)$ for $1 \leq i \leq n$ with $k \neq i$ and $j \in \{0,1\}$.

\end{itemize}

Imperfect information is defined by $(i, j) \sim^V_n (k, l)$ and $a_i \sim^A_n a_k$ 
for any $1 \leq i, j, k, l \leq n$. So each two positions from any two cycles are indistinguishable and 
each two of the actions $a_\sigma^i$ are indistinguishable. 

In Figure \ref{figureprp3}, the game $\game_2$ and the powerset game
$\ol{\game}_n$ are depicted.  The position $\{v_0\}$ of the powerset
game is omitted and a position $\{(0,j_1),(1,j_2)\}$ is represented
as $j_1j_2$.

It is clear that $X(\game_n) \leq 2$ for any measure $X \in
\calM$. Indeed, \dagw is \treew plus one and \treew is one here, because the
underlying graphs are undirected trees. The entanglement game is won
by two cops: the cops force the robber to~$v_0$ and then one of them occupies
$v_0$. The robber goes into some cycle~$i$ and the other cop occupies
$(i,0)$. Then the first cop occupies $(i,1)$. In the cops and
invisible robber game, one cop is placed on $v_0$ and then the two
other cops visit successively every cycle, so $\dpw(\game_n) = 2$.

Performing the powerset construction on $\game_n$ from $v_0$ we obtain the graph $\olG_n$.
Obviously,~$\olG_n$ contains the position $\{(1, 0), \ldots, (n, 0)\}$.
From this position, an undirected square grid of super-{}polynomial
size is constructed as follows. The positions of~$\olG_n$ (except for
$\{v_0\}$) are precisely the sets of vertices of~$G_n$ that contain
exactly one vertex from every cycle of~$G_n$, \ie $\olV_n =
\bigl\{\{v_0\}\bigr\} \cup \bigl\{  \{ (1,j_1),\ldots, (n,j_n)\} \mid
j_i = 0,1  \bigr\}$. Action~$\neg_i$ switches the vertex in the $i$th
cycle and lets the other cycles unchanged. 

Now we observe how the powerset construction orders the positions
of~$\olG_n$ in a square grid.  
We successively apply actions $\neg_i$ for $i \in \{1, \ldots, n/2\}$ to create each vertex 
$\{(1, j_1), \ldots, (n/2, j_{n/2})$, $(1, 0), \ldots, (n, 0)\}$ 
with $j_1, \ldots, j_{n/2} \in \{0, 1\}$. In each step we can change exactly one $j_r$ to $1 - j_r$, so the creation of
all these vertices from $\{(1, 0), \ldots, (n/2, 0), (n/2+1, 0), \ldots, (n, 0)\}$ can, for instance, be done using the usual 
Gray-code for binary numbers: we get the next vertex
by applying $\neg_i$ to the previous vertex $\{(1, j_1), \ldots, (n/2, j_{n/2})$, $(n/2+1, 0), \ldots, (n, 0)\}$,
which changes exactly one position $(i, j_i)$.
This undirected path forms the upper horizontal side of the grid. 
Analogously, by successively applying the actions $\neg_i$ for $i \in \{n/2+1, \ldots, n\}$ we can create each vertex
$\{(1, 0), \ldots, (n/2, 0)$, $(n/2+1, j_{n/2+1}), \ldots, (n, j_n)\}$ with $j_{n/2+1}, \ldots, j_n \in \{0, 1\}$ using
the Gray-code. This undirected path forms the left vertical side of the grid. 

Likewise, given any vertex $\{(1, j_1)$, $\ldots$, $(n/2, j_{n/2})$, $(n/2+1, 0)$, $\ldots$, $(n, 0)\}$ we can
create any vertex $\{(1, j_1)$, $\ldots$, $(n/2, j_{n/2})$, $(n/2+1, j_{n/2+1})$, $\ldots$, $(n, j_n)\}$ 
by successively applying the actions $\neg_i$ for $i \in \{n/2+1, \ldots, n\}$ in the same order as before and 
given any vertex $\{(1, 0)$, $\ldots$, $(n/2, 0)$, $(n/2+1, j_{n/2+1})$, $\ldots$, $(n, j_n)\}$, 
by successively applying the actions $\neg_i$ for $i \in \{1, \ldots, n/2\}$, we can create any vertex
$\{(1, j_1)$, $\ldots$, $(n/2, j_{n/2})$, $(n/2+1, j_{n/2+1})$, $\ldots$, $(n, j_n)\}$.
All these paths form a $2^{n/2} \times 2^{n/2}$-grid and therefore, the \treew of
$\olG_n$ is super-{}polynomial in the size of $\game_n$. Furthermore, using that
$\olG_n$ is undirected one easily checks that for all $X\in \calM\setminus\{\ent\}$,
$X(\olG_n) \ge \tw(\olG_n)$. For entanglement, Berwanger et al. showed
in~\cite{BerwangerDawHunKre06} that non-{}monotone \dagw of a graph
(which is at most its \treew plus one) is
at most its entanglement plus one, so $\ent(\olG_n)\ge n+2$, for $n\ge
3$. 
\end{proof}

\begin{remark}
The super-{}polynomial size of the resulting graph is not needed for
unbounded growth of graph complexity. By the same technique,
replacing~$n$ cycles by two undirected $n$-paths with similar actions and self{}-loops on all
positions leads to an $n\times n$-grid.
\end{remark}

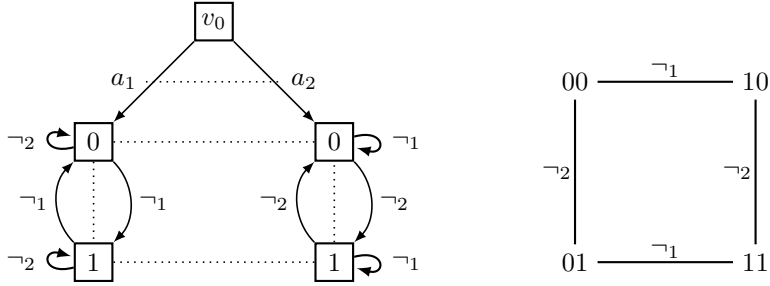
\begin{figure}
\begin{tikzpicture}[bend angle=45,auto, scale=0.8]

\node [player1] (v0) at (0, 0) {$v_0$};

\node [player1] (01) at (-2, -2) {$0$}
	edge [loop left, thick] (01);
\node [player1] (11) at (-2, -4) {$1$}
	edge [loop left, thick] (11);
\node [player1] (02) at ( 2, -2) {$0$}
	edge [loop right, thick] (02);
\node [player1] (12) at ( 2, -4) {$1$}
	edge [loop right, thick] (12);

\draw [post] (v0) to (01);
\draw [post] (v0) to (02);
\draw [post] (01) to [bend left = 45] (11);
\draw [post] (11) to [bend left = 45] (01);
\draw [post] (02) to [bend left = 45] (12);
\draw [post] (12) to [bend left = 45] (02);

\draw [ind] (01) to (11);
\draw [ind] (11) to (12);
\draw [ind] (12) to (02);
\draw [ind] (02) to (01);

\node (a1) at (-1.5,-1) {$a_1$};
\node (a2) at (1.5, -1) {$a_2$};
\node (neg1) at (-3,-3) {$\neg_1$};
\node (neg1) at (-1,-3) {$\neg_1$};
\node (neg1) at (3.2,-2) {$\neg_1$};
\node (neg1) at (3.2,-4) {$\neg_1$};
\node (neg2) at (3,-3) {$\neg_2$};
\node (neg2) at (1,-3) {$\neg_2$};
\node (neg2) at (-3.2,-2) {$\neg_2$};
\node (neg2) at (-3.2, -4) {$\neg_2$};

\draw [ind] (a1) to (a2);

\node (00) at ( 6, -1) {$00$};
\node (10) at ( 9, -1) {$10$};
\node (01) at ( 6, -4) {$01$};
\node (11) at ( 9, -4) {$11$};

\draw [-,thick] (00) to (10);
\draw [-,thick] (00) to (01);
\draw [-,thick] (01) to (11);
\draw [-,thick] (11) to (10);

\node (neg1) at (7.5,-0.8) {$\neg_1$};
\node (neg2) at (5.7, -2.5) {$\neg_2$};
\node (neg1) at (7.5,-3.8) {$\neg_1$};
\node (neg2) at (8.7,-2.5) {$\neg_2$};

\end{tikzpicture}
\caption{The game $\game_2$ and the powerset graph $\olG^2$}
\label{figureprp3}
\end{figure}

Proposition~\ref{proposition_exponentialgrowth} shows that Reif's
construction does not help to solve parity games efficiently even if
the game graphs are simple. Before we show that the problem is, in
fact, very hard, let us note that on trees, imperfect information does
not provide additional computational complexity.  The powerset graph
of a tree is again a tree (recall that we delete non-{}reachable
positions) where the set of positions on each level partitions the set
of positions on the same level of the original tree. Thus the new tree
can be computed in polynomial time and is at most as big as the
original tree.

For the following proofs we need the notion of an alternating Turing
machine. An alternating Turing machine~$M = (Q,\Gamma,\Sigma,q_0,\Delta,Q_\acc,Q_\rej)$ is defined as a deterministic
Turing machine, but now the set of non-{}final states is partitioned
in $Q_\dt$, $Q_\exists$ and $Q_\forall$. Whether a  word is accepted by~$M$ is
defined by  game semantics. There are two players, both having perfect
information: the existential
Player~$\exists$ and the universal Player~$\forall$.  If~$M$ is in a
state from~$Q_\dt$, then there is exactly one next configuration as for
deterministic Turing machines. If~$M$ is in a state~$q$ from~$Q_\exists$, the
existential player resolves the non-{}determinism choosing a transition
$(q,a)\to(q',a',s)\in \Delta$ and if~$M$ is in a state
from~$Q_\forall$, the universal player moves. The existential player
tries to accept the input word, the universal player aims to reject it
or to drive~$M$ into an infinite computation. A word~$w$ is accepted
by~$M$ if the existential player has a winning strategy from the
initial configuration of~$M$ on~$w$. The complexity classes
$\apspace$, $\aspace(S(n))$, $\aptime$, and $\atime(S(n))$ (for a
function $S\colon\omega\to\omega$) are defined using alternating Turing machines as the classes
$\pspace$, $\Space(S(n))$, $\ptime$, and $\Time(S(n))$ with deterministic Turing machines.
Our proofs are based on the following facts, see for example~\cite{WW86}.
\begin{lemma}\label{lemma_complexity}\
\vspace{-2mm}
\begin{enumerate}[(1)]
\item $\text{\sc{APspace}} = \exptime$.
\item $\text{\sc{APtime}} = \pspace$.
\end{enumerate}
\end{lemma}

\begin{theorem}\label{theorem_exptimehard}
The following problem is $\exptime$-hard. 
Given an imperfect information reachability game $\game $ with $\ent(\game) \leq 2$
and $\dpw(\game) \leq 3$ and a position~$v_0$, does Player~0 have a
winning strategy from~$v_0$ in $\game$?
\end{theorem}
\begin{proof}
By Lemma~\ref{lemma_complexity}, for any $L \in \exptime$, 
there is an alternating Turing machine $M = (Q, \Gamma, \Sigma, q_0, \Delta)$ 
with only one tape and space bound $n^k$ for some $k \in \omega$, where~$n$ is the size of the input, 
that recognizes~$L$. As usual,~$Q$ is the set of states, $\Gamma$ and
$\Sigma$ are the input and the tape alphabets with $\Gamma \subseteq \Sigma$,
$q_0$ is the initial state, and $\Delta$ is the transition relation.
First assume that~$M$ is deterministic. We describe the necessary changes to prove the general
case later.

Let $A = \Sigma \cupdot (Q \times \Sigma) \cupdot \{ \#\}$. Then each configuration~$C$ of~$M$ is described by a word 
$C = \# w_1 \ldots w_{i-1} (qw_i) w_{i+1} \ldots w_t  \in A^*$
over $A$ where~$w_j$ is the $j$th symbol on the tape and the
reading head is at symbol number~$i$ (counting from~$0$).  Since~$M$ has space bound $n^k$ 
and we have $k \geq 1$, without loss of generality we can assume that
$|C| = n^k+1$ for all configurations~$C$ of~$M$ on inputs of length~$n$. 
Moreover, for a configuration~$C$ of~$M$ and $1\le i \leq n^k$ the symbol
number~$i$ of the successor configuration~$C'$ only depends on the
symbols number $i-1$, $i$ and $i+1$ of~$C$. So there is a function $f : A^3 \to A$ such that for
any configuration~$C$ of~$M$ and any $i \leq n^k$,
if the symbols number $i$, $i+1$ and~$i+2$ of~$C$ are $a_1$, $a_2$ and
$a_3$, then the symbol number~$i$ of the successor configuration~$C'$ of~$C$ is $f(a_1,a_2,a_3)$.

For each input word $u \in \Gamma^*$ we construct a game~$\game_u$
with imperfect information such that the player called Constructor has
a winning strategy for~$\game_u$ if and only if~$M$ accepts~$u$.
The idea for the game corresponding to~$u$ is the following. Player
Constructor selects symbols from~$A$ such that 
the sequence constructed in this way forms an accepting run of~$M$ on~$u$. In order to check the
correctness of the construction, player Verifier may, at any point during the play, but only
\emph{once}, memorize some $i \in \{1, \ldots, n^k\}$, and $a_i$,
$a_{i+1}$ and $a_{i+2}$ chosen by Constructor within the recent configuration. In the next configuration, Verifier checks the $i$th symbol chosen
by Constructor to be correct according to $a_{i-1}$, $a_i$ and $a_{i+1}$, and the function~$f$. 
If the $i$th symbol proves incorrect, Constructor loses, otherwise,
Verifier loses. If Verifier never checks a transition, Constructor
wins if and only if he reaches an accepting configuration. 
Constructor must not notice when Verifier memorizes the recent position, which defines the imperfect information
in the game. Then Constructor has a winning strategy in the game if and only if~$M$
accepts~$u$. To justify the bounds on the graph complexity measures that we have claimed, 
we define the game more formally.

The set of positions is 
$ \{v_0\} \cup \{C, V\} \times A \times \{0, \ldots, n^k\} \times Q \times \{-,1, \ldots, n^k\} \times A^3$,
so a position has the form $(\sigma, a, i, q, j, a_1,a_2,a_3)$
where~$\sigma$ is the player to move,~$a$ is the recent symbol chosen by Constructor
and~$i$ is the number of~$a$ in the recent configuration.
Furthermore,~$q$ is the last state in~$Q$ chosen by Constructor,
and~$j$ and $a_1, a_2, a_3$ represent the information memorized by
Verifier:~$j$ is the number of the symbol to be verified in the next
configuration, and $a_1$, $a_2$ and~$a_3$ are symbols number $j-1$,
$j$ and $j+1$, respectively . All actions are indistinguishable for Constructor and we
omit them in the description. The sign~$-$ in the four last components
of a position means that Verifier did not memorize the corresponding element.

A play begins in position~$v_0$, which belongs to Verifier. He moves to
a position $(C,\#,0,q_0,-,-,-,-)$ or to position
$(C,\#,0,q_0,j,a_1,a_2,a_3)$ where $1 \le j \le n^k$ and
$a_1, a_2, a_3$ are symbols number $j-1$, $j$ and $j+1$ of the initial
configuration of~$M$ on~$u$. 

As long as Verifier does not memorize any symbol, Constructor moves 
from  position $(C,a,i,q,-,-,-,-)$ with $0\le i<n^k$ to some
position $(V,a',i+1,q',-,-,-,-)$ choosing the next symbol and giving
Verifier the possibility to memorize it. Hereby either $a' = (q',a'') \in Q\times \Sigma$ (for some
$a''\in \Sigma$), or $a' \in\Sigma$ and $q' = q$. As Verifier does not
memorize anything yet, he chooses $(C,a',i+1,q,-,-,-,-)$ as the next
position (the other possible move is to memorize~$a'$).
If $i=n^k$, then the next position is $(C,\#,0,\perp,-,-,-,-)$, \ie
Constructor chooses $\#$ and Verifier does not memorize
it. Hereby~$\perp$ is some fixed state in~$Q$, \ie
the state~$q$ is forgotten in this move. We need this to reduce the
structural complexity of the game graph. A move of
Constructor and an answer of Verifier constitute a round.

Now assume that Verifier decides to memorize the tuple $(a_1,a_2,a_3)$
where~$a_1$ is the current symbol number $i<n^k$, and $a_2$ and $a_3$
are the (yet not determined) symbols that will be chosen in the next
two rounds. Then from a position $(V,a_1,i,q,-,-,-,-)$ Verifier
moves to $(C,a_1,i,q,-,a_1,-,-)$. Then Constructor moves to some
$(V,a_2,i+1,q',-,a_1,-,-)$ (where the update of~$q$ to~$q'$ is as before),
then Verifier moves to $(C,a_2,i+1,q',i+1,a_1,a_2,-)$
and Constructor moves to some $(V,a_3,i',q'',i+1,a_1,a_2,-)$ where
$q'$ is again updated as before and~$i'$ depends on~$i$. 
If $i+2 \le n^k$, then $i' = i+2$.
Otherwise $i+2 = n^k+1$, then $i' = 0$ (and $a_3 = \#$). Verifier moves to
$(C,a_3,i',q'',i+1,a_1,a_2,a_3)$. From this position, the players,
first, finish the current configuration and, second, play in the next
configuration until the position with index~$i+1$ is reached, both in the same way as they played
without any memorized information. Formally, we just substitute in the above
positions the four last elements $(\ldots,-,-,-,-)$ by
$(\ldots,i+1,a_1,a_2,a_3)$.
When a position $(C,b, i+1,s,i+1,a_1,a_2,a_3 )$ is reached, the play
stops and Verifier wins if
and only if $f(a_1,a_2,a_3) = b$. At any \emph{other} position
$(\sigma,a,i,q,j,a_1,a_2,a_3)$ (where $j$ and all $a_k$ can be~$-$),
if~$q$ is accepting, Constructor wins and if~$q$ is rejecting,
Verifier wins. In the remaining case of an infinite play (Verifier
never memorizes anything and no final state is reached), Verifier wins.

Imperfect information is defined by making all positions
$(\sigma,a,i,q,j,a_1,a_2,a_3)$ and
$(\sigma',a',i',q',j',a_1',a_2',a_3')$ indistinguishable for
Constructor if $\sigma = \sigma'$, $a = a'$, $i = i'$, and $q = q'$,
\ie Constructor does not know whether Verifier memorized anything.

It is clear that~$u$ is accepted by~$M$ if and only Constructor has a
winning strategy in the game~$\game_u$. If~$u$ is accepted, then
Constructor just constructs the accepting run of~$M$. If not,~$M$
rejects (as~$M$ recognizes an \exptime language, it always stops). In
order not to lose by reaching a rejecting state, Constructor has to cheat.
However, cheating is not a winning strategy for Constructor
because Verifier can memorize the place in the previous configuration
that does not match the same place in the current configuration and win.

We now analyze the structural complexity of the game graph, see Figure~\ref{fig_G_u}. 
The \emph{main subgame}~$\calS$ consists of positions of the form $(\sigma, a,
i, q, - , - , - , -)$ without memorization that build a DAG with a
unique root $(C,\#,0,\perp,-,-,-,-)$ and $2\cdot n^k +1$ layers. A
layer number~$i$ with an even~$i$ has the form
$(C,a,i,q,-,-,-,-)$. From every such position there is an edge to
every position of the form $(V,a',i+1,q',-,-,-,-)$ of the next
layer. Analogously, from every position of layer number~$i+1$ there is
an edge to every position of layer number~$i+2$. Finally, from every
position of the last layer, there is an edge back to the root
$(C,\#,0,\perp,-,-,-,-)$. This constitutes the only cyclicity in the
graph. Additionally, there are edges from~$v_0$ to every position $(V,a,1,\perp,-,-,-,-)$.

From every of $n^k$ Verifier positions~$P$ in the main subgame and from~$v_0$, Verifier can
start memorizing information. Then the play continues in a \emph{checking
subgame} $\calC_P$ and never returns to the main subgame, so we can consider their
complexities independently. Every checking subgame is again a DAG,
which consists of two sub-DAGs. The first one is a copy of the remaining
part of the main subgame (with changed four last components); the other one is a copy 
of the part of the main subgame which has been played until Verifier
intended to memorize information (again with changed four last
components). There are no outgoing edges from the last level of a
checking subgame.

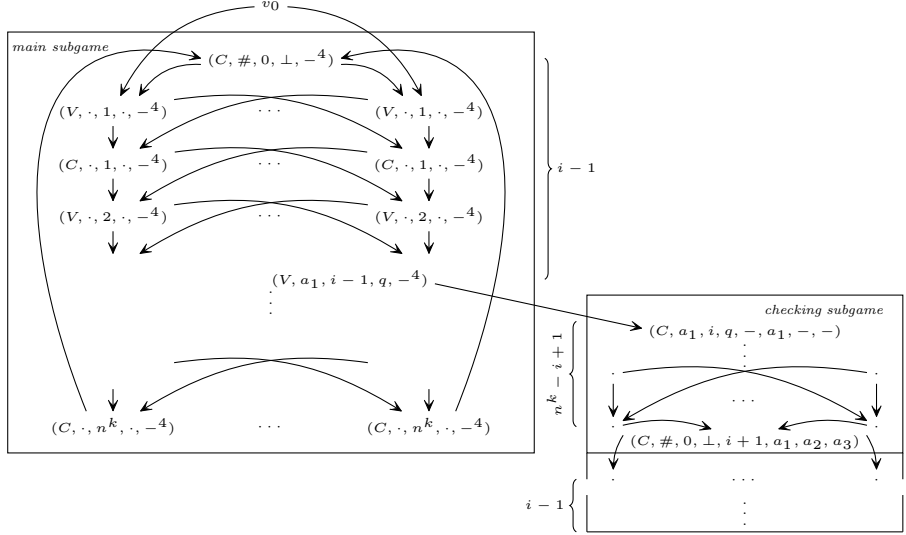
\begin{figure}
\begin{tikzpicture}[scale=0.7]
\tiny
\node (v_0) at (0,0){$v_0$};


\node (root) at (0,-1){$(C,\#,0,\perp,-^4)$};

\node (11) at (-3,-2){$(V,\cdot,1,\cdot,-^4)$};    \node at (0,-2){$\cdots$}; \node  (1s) at (3,-2){$(V,\cdot,1,\cdot,-^4)$};    
\node (21) at (-3,-3){$(C,\cdot,1,\cdot,-^4)$};    \node at (0,-3){$\cdots$}; \node  (2s) at(3,-3){$(C,\cdot,1,\cdot,-^4)$};    
\node (31) at (-3,-4){$(V,\cdot,2,\cdot,-^4)$};    \node at (0,-4){$\cdots$}; \node  (3s) at(3,-4){$(V,\cdot,2,\cdot,-^4)$};    
\node (41) at (-3,-5){\phantom{$(V,\cdot,2,\cdot,-^4)$}};   \node  (4s) at(3,-5){\phantom{$(V,\cdot,2,\cdot,-^4)$}};    
\node (n^k_11) at (-3,-7){\phantom{$(V,\cdot,2,\cdot,-^4)$}};   \node  (n^k_1s) at(3,-7){\phantom{$(V,\cdot,2,\cdot,-^4)$}};    
                                        \node at
                                        (0,-5.5){$\vdots$};

                                              \node (strt_chk) at
                                              (1.5,-5.2){$(V,a_1,i-1,q,-^4$)};

\node (n^k1) at (-3,-8){$(C,\cdot,n^k,\cdot,-^4)$};    \node at (0,-8){$\cdots$}; \node  (n^ks) at (3,-8){$(C,\cdot,n^k,\cdot,-^4)$};

\draw[bend right,-slim] (v_0) to (11); \draw[bend left,-slim] (v_0) to (1s);
\draw[bend right,-slim] (root.185) to (11.30); \draw[bend left,-slim] (root.-5) to (1s.150);

\draw[-slim] (11) to (21); \draw[bend left=20,-slim] (11) to (2s); 
\draw[bend right=20,-slim] (1s) to (21); \draw[-slim] (1s) to (2s);

\draw[-slim] (21) to (31); \draw[bend left=20,-slim] (21) to (3s);
\draw[bend right=20,-slim] (2s) to (31); \draw[-slim] (2s) to (3s);

\draw[-slim] (31) to (41); \draw[bend left=20,-slim] (31) to (4s);
\draw[bend right=20,-slim] (3s) to (41); \draw[-slim] (3s) to (4s); 

\draw[-slim] (n^k_11) to (n^k1); \draw[bend left=20,-slim] (n^k_11) to (n^ks);
\draw[bend right=20,-slim] (n^k_1s) to (n^k1); \draw[-slim] (n^k_1s) to (n^ks); 

\draw[-slim] (n^k1.150)  .. controls (-5,-4) and (-5,0) .. (root.west);
\draw[-slim] (n^ks.30)  .. controls (5,-4) and (5,0) .. (root.east);

\draw (-5,-0.5) rectangle (5,-8.5);
\node at (-4,-0.8){\textit{main subgame}};
\draw[brace] (5.2,-1) -- (5.2,-5.2);
\node at (5.8,-3.1){\tiny $i-1$};


\node (first) at (9,-6.2){$(C,a_1,i,q,-,a_1,-,-)$};
\draw[-slim] (strt_chk.-3) to (first.178);

\node (dots) at (9,-6.5){$\vdots$};

\node (n^k_11) at (6.5,-7){$\cdot$};   \node  (n^k_1s) at(11.5,-7){$\cdot$};    
\node at (9,-7.5) {$\cdots$};
\node (n^k1) at (6.5,-8){$\cdot$};   \node(n^ks) at(11.5,-8){$\cdot$};

\draw[-slim] (n^k_11) to (n^k1); \draw[bend left=20,-slim] (n^k_11) to (n^ks);
\draw[bend right=20,-slim] (n^k_1s) to (n^k1); \draw[-slim] (n^k_1s) to (n^ks);

\draw (6,-5.5) rectangle (12,-8.5);
\draw (6,-8.5) -- (6,-9); \draw (6,-9.3) -- (6,-10);
\draw (12,-8.5) -- (12,-9); \draw (12,-9.3) -- (12,-10);
\draw (6,-10) -- (12,-10);

\node (root) at (9,-8.3){$(C,\#,0,\perp,i+1,a_1,a_2,a_3)$};
\draw[bend left=15,-slim] (n^k1) to (root);
\draw[bend right=15,-slim] (n^ks) to (root);

\node (11) at (6.5,-9){$\cdot$}; \node at (9,-9) {$\cdots$}; \node (1s) at (11.5,-9){$\cdot$};
\draw[bend right=15,-slim] (root) to (11.90);
\draw[bend left=15,-slim] (root) to (1s.90);

\node at (9,-9.5){$\vdots$};

\draw[brace] (5.8,-8) -- (5.8,-6);
\node[rotate=90] at (5.4,-7){$n^k-i+1$};

\draw[brace] (5.8,-10) -- (5.8,-9);
\node at (5.2,-9.5){$i-1$};

\node at (10.5,-5.8){\textit{checking subgame}};

\end{tikzpicture}
\caption{The game graph of $\game_u$ (with only one checking
  subgame). ``$-^4$'' is short for ``-,-,-,-''.}
\label{fig_G_u}
\end{figure}

It is clear that $\ent(\game_u) \leq 1$ (place the cop on the root
$(C,\#,0,\perp,-,-,-,-)$ and wait until the robber reaches a leaf of
the resulting DAG) and $\dpw(\game) \leq 1$ (place one cop on
the root and capture the robber with the other cop on
the resulting DAG). Notice that we are still considering the special case where~$M$ is deterministic.
Obviously, $\game_u$ can be constructed from a given input $u \in
\Gamma^*$ in polynomial time.

Now consider the general case, where~$M$ is not necessarily
deterministic. We let Constructor play the role of the existential
player and Verifier the role of the universal player. As before,
Constructor writes symbols of the current configuration (now including
existential choices) and Verifier checks that the current
configuration can follow the previous one.  However, if we let Constructor
check universal choices of Verifier in the same way (by privately
remembering a place in the previous configuration), the reduction to
the games does not work. Indeed, it can happen that~$M$ accepts~$u$,
but Constructor has no winning strategy: he does not know which place
in which configuration he should remember. For this reason, we
explicitly remember the last choice of Verifier in the position of the game.

Without loss of generality we can assume that each non-terminal
configuration of~$M$ has exactly two successor configurations. 
If there is a configuration~$C$ with just a single successor configuration, then we add a 
default successor to~$C$ which leads to acceptance if~$C$ is universal
and which leads to rejectance if~$C$ is existential. If there is a configuration 
with $b > 2$ successors, then we replace this $b$-branching by a 
binary branching configuration tree of depth~$b$ by modifying 
the transition function of~$M$ in an appropriate way. Obviously, this construction can be 
done in such a way that it merely increases the state space of~$M$ and the time bound by a 
constant factor, but not the space bound. 

Now, instead of one function~$f$, we have two functions $f_1, f_2 : A^3 \to A$, such that the 
following holds. If~$C$ is a configuration of~$M$, $s \in \{1, 2\}$ and 
$1 \leq i \leq n^k$, and the symbols number~$i$, $i+1$ and 
$i+2$ of~$C$ are $a_1$, $a_2$, $a_3$, then the symbol number~$i+1$ of the successor 
configuration number~$s$ (there are two successor configurations) of~$C$ is $f_s(a_1,a_2,a_3)$. Thus, the
main subgame~$\calS$ and every checking subgame $\calC_P$ are replaced by two
copies $\calS^s$ and $\calC_P^s$ for $s\in\{0,1\}$. Thus every position
except~$v_0$ has an additional component~$0$ or~$1$, which we make the 
first one, so a position has the form
$(s,\sigma,a,i,q,j,a_1,a_2,a_3)$. Intuitively, the previous non-{}deterministic (existential
or universal) choice is memorized in the first component of a position. 

Edges from~$v_0$ to~$\calS$ go now to both copies. Edges from the leaves of~$\calS$ to
its root go now from leaves of both subgames to the roots of both subgames (thus
introducing new cycles). If the state of the current
configuration is universal, the leaf positions now belong to Verifier,
\ie we have positions $(s,\sigma,a,n^k,q,j,a_1,a_2,a_3)$ where $\sigma = C$ if~$q$
is existential and $\sigma = V$ if~$q$ is universal. The edges are thus
$(s,P,a,n^k,j,a_1,a_2,a_3)$ to $(s',C,\#,0,j,a_1,a_2,a_3)$ for
$s,s'\in \{0,1\}$. The edges in the checking subgames are changed
analogously (without introducing new cycles, because there are no
edges from the leaves to the roots).

Imperfect information is defined as before with the additional condition that 
Constructor observes the copy of the subgame in which the
play currently takes place. 

Clearly these modifications merely increase the 
entanglement of the graph from at most~$1$ to at most~$2$ (place two
cops on both roots of~$\calS^0$ and of~$\calS^1$). The \dpathw is now at
most~$2$ (place two cops on the roots and use the third cop to capture
the robber on the resulting DAG).
\end{proof}

\begin{remark}
The (undirected) \pathw and the \treew of the game graph are also
bounded. Both~$\calS^s$ have edges only from one layer to the next one
and from the leaves to both roots. Each layer has $|\{0,1\} \times
\{C,V\} \times A \times Q| = 4\cdot|A|\cdot |Q|$ elements, so
$8\cdot|A|\cdot |Q| + 2$ cops capture the robber in $\calS^s$ by blocking both roots
and occupying one layer after another successively. In $\calC_P$ the
layers are larger and have size at most $4\cdot|A|\cdot|Q|\cdot|A|^3$
(note the last but four component~$j$ is fixed and depends only on~$P$). Hence,
$8\cdot|A|^4\cdot|Q|$ cops capture the robber there. If the robber is
visible, $k_1 = 8\cdot|A|^4\cdot|Q|$ suffice, because if the robber goes to
some~$\calC_P$, then a cop occupies~$P$ and there is no way back for
the robber from~$\calC_P$. If the robber is invisible, the cops search
every $\calC_P$ immediately after occupying~$P$. In the meanwhile, one
layer in~$\calS^s$ must remain blocked, so the cops can get along with 
$k_2 = 8\cdot|A^4|\cdot|Q| + 4\cdot|A|\cdot|Q|$ cops.
Assuming that~$M$ recognizes an \exptime{}-complete problem,we obtain that the strategy problem for
reachability games with imperfect information on graphs of \treew at
most~$k_1$ and \pathw at most~$k_2$ is $\exptime$-hard.
\end{remark}


The cases of entanglement and directed path-width at most~$1$ remain open for reachability
games, but we can solve them for \emph{sequence-{}forcing} games. 
A sequence{}-forcing condition can be described by a pair $(S, \col)$
where $\col : V \to C\subset \omega$ is a coloring of game positions by natural numbers
and $S \sq \{1, \ldots, r\}^k$ is a set of sequences of length~$k$ for
some $k \in \omega$. Player~$0$ wins an infinite play~$\pi$ of a
sequence-\mbox{}forcing game if for some $i \in\omega$ we have $\col(\pi(i))\col(\pi(i+1)) \ldots \col(\pi(i+k)) \in S$. 
Clearly if~$k$ is fixed, sequence-\mbox{}forcing games can be polynomially reduced to reachability 
games by using a memory which stores the last $k$ colors that have occurred. (Notice that this reduction may, however, 
increase the complexity of the game graph.) In particular, the strategy problem for sequence-\mbox{}forcing games 
with fixed~$k$ is in $\ptime$. On the other side, the strategy problem for sequence-forcing games with imperfect information is 
$\exptime$-hard on graphs of entanglement and directed path-width at most~$1$, already for $k = 3$. 

\begin{theorem}\label{theorem_exponentialblowup}
Sequence{}-forcing games with imperfect information on graphs of
entanglement and \dpathw 
at most~$1$ are \exptime-complete.
\end{theorem}
\begin{proof}
We modify the proof of Theorem \ref{theorem_exptimehard} as follows. From the nodes on level $2 \cdot n^k + 1$ of
$\calS^s$ for $s\in\{0,1\}$ we do not allow moves directly back to the
roots, but we redirect all edges to a single (new) position $0$, which
is common for both~$\calS^s$, belongs to 
Verifier and has color~$0$. From this position, Verifier may move to
position~$2$, which belongs to Constructor and has color~$2$, or to
position~$1$, which belongs to Verifier and has color~$1$. From~$0$
Constructor chooses whether to proceed in $\calS^0$ or in 
$\calS^1$ and from~$1$ Verifier makes this choice. So as
Constructor does not notice where the play proceeds in the main
subgame or in some checking subgame, the same
construction is performed in the checking subgames at places where the
configurations of~$M$ change and imperfect information is defined accordingly. 
All old positions obtain color~$0$ except for positions
$(s,\sigma,a,i,q,-,-,-,-)$ on the last levels of~$\calS^s$ where~$q$
is universal: they are colored with~$1$.

Now, $S = \{(0, 0, 1)\}$, that means, the unique sequence that Constructor wants to 
enforce is $(0, 0, 1)$. This forces Verifier into giving control back
to Constructor if the state in the recent configuration is
existential. Then the proof of Theorem \ref{theorem_exptimehard}
carries over. Note that a player still wins if his opponent has move,
but is unable to do it, in particular, the players win at their old
winning positions.

If a cop occupies position~$0$ in the modified game, the game graph
becomes acyclic, so the entanglement of the whole graph is~$1$ and its
\dpathw is~$2$.
\end{proof}

Finally, if the we consider acyclic game graphs, the strategy problem for imperfect information reachability games 
is $\pspace$-complete. Notice that acyclic graphs are precisely those having DAG-width $1$.

\begin{theorem}\label{theorem_pspacecomplete}
The strategy problem for reachability games with imperfect information on acyclic graphs is $\pspace$-complete.
\end{theorem}
\begin{proof}
First we prove the membership in $\pspace$. Let~$\game$ be a game on
an acyclic graph with imperfect information and let~$v_0$ be the
initial position. The idea is that carrying out the powerset construction 
on an acyclic graph~$G$ we again obtain an acyclic graph~$\olG$ where
by Lemma \ref{lemma_powerset}, the paths in~$\olG$ are not longer than the paths in~$G$, so we can solve
the reachability game on~$\olG$ by an~$\aptime$ algorithm. Starting from~$\{v_0\}$, we 
proceed as follows. Given a position $\olv \in \olV$ in the corresponding game~$\olG$ with perfect 
information, if $\olv \in \olV_0$, then the existential player guesses
a successor of~$\olv$ and if $\olv \in \olV_1$, then the existential player chooses a successor position of~$\olv$. 
If the computation reaches a leaf node in~$\olV_1$, the algorithm
accepts and if the computation reaches a leaf node 
in~$\olV_0$, the algorithm rejects. The construction of a successor position of some
position~$\olv$ can obviously be done in polynomial time. Moreover, if 
$\olpi = \olv_0, \olv_1, \ldots,  \olv_k$ is any path in~$\olG$, then according to Lemma \ref{lemma_powerset}, 
there is a path $\pi = v_0, v_1, \ldots, v_k$ with 
$v_i \in \olv_i$ for $i \in \{ 0, \ldots, k\}$. Since~$G$ is acyclic,
$k \leq n$. So, the computation stops after at most~$n$ steps.

Conversely, let $L \in \pspace$ be some decision problem. Then, according to Lemma~\ref{lemma_complexity}, 
there is an alternating Turing machine~$M$ with only one tape and time bound~$n^k$ for some $k \in \omega$ that recognizes~$L$. 
We use the same construction as in the proof of Theorem~\ref{theorem_exptimehard}. Since~$M$ has time bound~$n^k$ and only a 
single tape,~$M$ has also space bound~$n^k$. So we can describe configurations of~$M$ in the very same way as in the
proof of Theorem~\ref{theorem_exptimehard} and we can construct a game with positions as before. However, the essential difference 
here is that at a leaf position of~$\calS^s$, the next move does not lead back to the 
top of $\calS^s$ or~$\calS^{1-s}$ (for $s\in\{0, 1\}$), but it leads
to the roots of new copies of~$\calS^0$ and~$\calS^1$. This chain of
copies of~$\calS^s$ stops after~$n^k$ steps. 

If some input~$u$ is accepted by~$M$, then Constructor can prove this by constructing at most~$|u|^k$ configurations,
so winning strategies carry over between the game constructed in the proof of Theorem~\ref{theorem_exptimehard} and 
the game constructed here. Moreover, since the graph we have constructed is acyclic by definition, the proof
is finished.
\end{proof}


\section{Bounded imperfect information}\label{sec_bpi}

We turn to the case where the size of the equivalence classes of
positions is bounded. We show that \treew and entanglement become
unbounded after the application of the powerset construction, but
non-{}monotone \dagw and \dpathw do not. The more difficult case of
\dagw is treated in Sections~\ref{sec_simulated_pg} and~\ref{sec-parity}.

\subsection{Negative results}\label{subsec_neg_results}

The first observation is that bounded 
tree-width may become unbounded when applying the powerset construction. 
Afterwards we will see, that the same result holds for entanglement.

\begin{proposition}\label{prop_unbounded_tw}
For every~$n > 3$, there are games~$\game^n$, with bounded
imperfect information and \treew and \dpathw~$1$, and \dagw, \kellyw and entanglement~$2$
such that the corresponding powerset games~$\olgame^n$ have unbounded \treew.

\end{proposition}
\begin{proof}
The game graph of~$\game^n$ is a disjoint union of~$n$ undirected
paths of length~$n$ together with another vertex~$v_0^n$ and directed
edges from~$v^n_0$ to every other vertex. Imperfect information
connects vertices from neighbor paths.
The graph~$\graphG^4$ (without~$v^4_0$) is shown in Figure~\ref{fig_half_grid} (on the
left). Formally for any even natural number~$n >3$, let~$\calG^n = (V^n,V^n_0,E^n,v^n_0,
\sim^{V,n},\sim^{A,n},\Omega)$ be the following game:
\begin{itemize}
\item~$V^n = V^n_1 = \{v^n_0\} \cup \{ (i, j) \,|\, 1 \le i, j \le n\}$\,,
      \ie~$V^n_0 = \0$\,;

\item actions play no role and we do not consider them;

\item~$E^n = \{\big(v^n_0,(i,j)\big) \mid 1 \le i, j \le n\}\cup \{\big((i,j),
(i+1,j)\big), \big((i+1,j),(i,j)\big) \mid 1 \le i, j \le n\}$\,;

\item $\sim^{A,n} = V^n \times V^n$ (Player~$0$ does not distinguish
  any actions), and for~$i<n$, 
  \begin{itemize}
   \item if~$i$ is odd and~$j$ is even, then~$(i,j)\sim^{V,n} (i+1,j)$\,,
   \item if~$i$ is even and~$j$ is odd, then~$(i,j)\sim^{V,n} (i+1,j)$\,;
  \end{itemize}
\item $\Omega = \0$ (the winning condition does not play any role here). 
\end{itemize}
The values of directed measures for~$\calG^n$ are clear. For
entanglement the strategy is to chase the robber with one cop until he
goes to the right. Then the play proceeds in rounds. In a round one
cop (at the beginning the first cop) is a left bound for the robber
movements. The other cop chases the robber until he goes to the
right. Continuing in this way two cops capture the robber.

The powerset graph~$\olG^n$ has a structure similar to the Gaifman
graph of~$\game^n$. It has the same paths whose vertices have the form~$\{(i,j)\}$,
for~$1
\le i, j \le n$, and are now connected by a gadget consisting of a new vertex
$\big\{(i,j),(i\pm 1,j)\big\}$ (hereby, $\pm 1$ depends on parities of~$i$
and~$j$) and
directed edges going from that vertex to the row above and to the row below. A
connection is in an odd column if the lower row is odd and in an even column if
the lower row is even (starting with the odd row~$1$), see
Figure~\ref{fig_half_grid} (the graph on the right).

Formally,~$\olcalG^n = (\olV^n,\olV^n_0,\olE^n,\olv^n_0)$ (we omit
actions, the absent imperfect information and the winning condition
that play no role) where 
\begin{itemize}
 \item the positions are defined by \[\olV^n =
\big\{ \{v^n_0\}\big \} \cup \big\{ \{(i,j)\} \mid (i,j)\in V^n\big\} \cup 
\big\{ \{(i,j),(i+1,j)\} \mid i+j = 1 \mod 2 \big\}\,,\]

\item no positions belong Player~0:~$\olV^n_0 = \0$\,,

\item the moves are 
  \begin{align*}
  \olE^n = 
&\Big\{
  \big(v^n_0,\{(i,j)\}\big) \mid 1 \le i, j \le n\} \\
  \cup &\big\{
          \big(v^n_0,\{(i,j),(i+1,j)\}\big) 
	 \mid i+j = 0 \mod 2
        \big\} \\
  \cup &\big\{ 
	  \big( \{(i,j), (i+1,j)\}, \{(i+b, j+c)\} \big) \mid b\in\{0,1\},
c\in\{1,-1\}\text{, and } \\
	 &\qquad \{(i+b,j+c)\}\in\olV^n,(i,j)\sim^{V,n}(i+1,j)
	\big\}
\Big\}\,,
  \end{align*}
	and

\item the starting position is~$\olv^n_0 = \{v^n_0\}$\,.
\end{itemize}

We show that~$\olG^n$ has an~$(m\times m)$-grid as minor where~$m = n$ if
$n$ is even and~$m = n-1$ if~$n$ is odd. We cut off~$\olv^n_0$, and if~$n$ is
odd, we cut off the~$n$th column. Further, we delete edges 
$\big( \{(i,j),(i+1,j)\},  \{(i+b,j+1)\} \big)$ if~$i$ is odd and 
$\big( \{(i,j),(i-1,j)\},  \{(i+b,j+1)\} \big)$ if~$i$ is even. The result is
shown in Figure~\ref{fig_half_grid} (the graph on the right). Now, the
directions of edges are forgotten, \ie instead of edges 
$\big\{ \{(i,j),(i+1,j)\},  \{(i+b,j+c)\} \big\}$ we have edges 
$\big\{ \{(i,j),(i+1,j)\},  \{(i+b,j+1)\} \big\}$. We obtain a
\emph{wall-\mbox{}graph} defined in~\cite{Kreutzer09} where it is shown that
such graphs have high \treew. Indeed, we contract edges~$\big( (i,j),(i,j+1) \big)$
for all~$i$ and add~$j$.
The result is an~$(m/2 \times n-1)$-grid, from which it is easy to obtain an
$(m\times m)$-grid by further edge contractions. It is well known that
the \treew of an~$(m\times m)$-grid is~$m$.
\end{proof}

Note that if we consider the whole game structure, \ie the Gaifman graph of the
given game, it is almost of the same shape as the powerset graph and its \treew is
unbounded as well. In fact, we will see in
Corollary~\ref{cor_bounded_Gaifman} that the \treew of a
powerset graph is bounded in the \treew of the Gaifman graph of the
given game.

\begin{figure}

\begin{tikzpicture}[bend angle=45,auto, scale=0.4,thick]
\tiny
\node (offset) at (-9,0) {};

\node (11) at (-6, -6) {(1, 1)};
\node (12) at (-3, -6) {(1, 2)};
\node (13) at (0, -6) {(1, 3)};
\node (14) at (3, -6) {(1, 4)};

\draw [-] (11.east) to (12.west);
\draw [-] (12.east) to (13.west);
\draw [-] (13.east) to (14.west);

\node (21) at (-6, -4) {(2, 1)};
\node (22) at (-3, -4) {(2, 2)};
\node (23) at (0, -4) {(2, 3)};
\node (24) at (3, -4) {(2, 4)};


\draw [-] (21.east) to (22.west);
\draw [-] (22.east) to (23.west);
\draw [-] (23.east) to (24.west);

\draw [ind] (12.north) to (22.south);
\draw [ind] (14.north) to (24.south);

\node (31) at (-6, -2) {(3, 1)};
\node (32) at (-3, -2) {(3, 2)};
\node (33) at (0, -2) {(3, 3)};
\node (34) at (3, -2) {(3, 4)};


\draw [-] (31.east) to (32.west);
\draw [-] (32.east) to (33.west);
\draw [-] (33.east) to (34.west);

\draw [ind] (21.north) to (31.south);
\draw [ind] (23.north) to (33.south);

\node (41) at (-6, 0) {(4, 1)};
\node (42) at (-3, 0) {(4, 2)};
\node (43) at (0, 0) {(4, 3)};
\node (44) at (3, 0) {(4, 4)};


\draw [-] (41.east) to (42.west);
\draw [-] (42.east) to (43.west);
\draw [-] (43.east) to (44.west);

\draw [ind] (32.north) to (42.south);
\draw [ind] (34.north) to (44.south);


\node (11) at (-6+13, -6) {(1, 1)};
\node (12) at (-3+13, -6) {(1, 2)};
\node (13) at (0+13, -6) {(1, 3)};
\node (14) at (3+13, -6) {(1, 4)};

\draw [-] (11.east) to (12.west);
\draw [-] (12.east) to (13.west);
\draw [-] (13.east) to (14.west);

\node (21) at (-6+13, -4) {(2, 1)};
\node (22) at (-3+13, -4) {(2, 2)};
\node (23) at (0+13, -4) {(2, 3)};
\node (24) at (3+13, -4) {(2, 4)};

\draw [-] (21.east) to (22.west);
\draw [-] (22.east) to (23.west);
\draw [-] (23.east) to (24.west);

\node (1222) [sep] at (-6+13, -5) {$\cdot$}; 
\draw [-slim] (1222.north) to (21.south);
\draw [-slim] (1222.south) to (11.north);

\node (1424) [sep] at (0+13, -5) {$\cdot$}; 
\draw [-slim] (1424.north) to (23.south);
\draw [-slim] (1424.south) to (13.north);

\node (31) at (-6+13, -2) {(3, 1)};
\node (32) at (-3+13, -2) {(3, 2)};
\node (33) at (0+13, -2) {(3, 3)};
\node (34) at (3+13, -2) {(3, 4)};

\draw [-] (31.east) to (32.west);
\draw [-] (32.east) to (33.west);
\draw [-] (33.east) to (34.west);

\node (2131) [sep] at (-3+13, -3) {$\cdot$}; 
\draw [-slim] (2131.north) to (32.south);
\draw [-slim] (2131.south) to (22.north);

\node (2333) [sep] at (3+13, -3) {$\cdot$}; 
\draw [-slim] (2333.north) to (34.south);
\draw [-slim] (2333.south) to (24.north);

\node (41) at (-6+13, 0) {(4, 1)};
\node (42) at (-3+13, 0) {(4, 2)};
\node (43) at (0+13, 0) {(4, 3)};
\node (44) at (3+13, 0) {(4, 4)};

\draw [-] (41.east) to (42.west);
\draw [-] (42.east) to (43.west);
\draw [-] (43.east) to (44.west);

\node (3242) [sep] at (-6+13, -1) {$\cdot$}; 
\draw [-slim] (3242.north) to (41.south);
\draw [-slim] (3242.south) to (31.north);

\node (3444) [sep] at (0+13, -1) {$\cdot$}; 
\draw [-slim] (3444.north) to (43.south);
\draw [-slim] (3444.south) to (33.north);

\end{tikzpicture}
\caption[Game graph~$\graphG^4$ and a subgraph of its
powerset graph~$\olcalG^4$\,.]{Game graph~$\graphG^4$ (without~$v_0$) and a
subgraph of its
powerset graph~$\olG^4$\,.}
\label{fig_half_grid}
\end{figure}
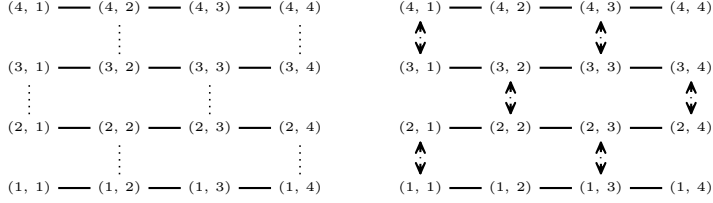

\begin{proposition}\label{prop_unbounded_entanglement}
For every~$n > 3$, there are games~$\game^n$, with bounded
imperfect information such that~$\ent(\game^n) = 2$ and 
the corresponding powerset games~$\olcalG^n$ have unbounded entanglement.
\end{proposition}
\begin{proof}
The game graph of~$\game^n$ (see Figure~\ref{fig_unbounded_ent}) consists of two
disjoint copies~$T_1$ and~$T_2$ of the full undirected binary tree of
depth~$n$. From a vertex in~$T_1$, a path of length two
leads to the corresponding vertex in~$T_2$ and there are
no paths from~$T_2$ to~$T_1$. The paths from~$T_1$ to~$T_2$ are
supplied with imperfect information in such a way that in the powerset
graph there appear connections also from~$T_2$ to~$T_1$. Thus,
in~$\olcalG_n$, corresponding vertices are now connected in both
directions. 

Let~$n \in \omega$ be even. We define the
game~$\game^n = (V^n,V^n_0,E^n,v^n_0,\sim^{V,n},\sim^{A,n},\0)$
where $\sim^{A,n}$ plays no role, so we do not define it. Let~$\alpha$ be the mapping~$\{0,
1\} \to \{a, b\}$ with~$0 \mapsto a$,~$1 \mapsto b$ and let~$\beta$ be
the mapping~$\{a, b\} \to \{\overline{0}, \overline{1}\}$,~$a \mapsto
\overline{0}$,~$b \mapsto \overline{1}$. We generalize~$\alpha$ to
words:~$\alpha(u_1\dots u_n) = \alpha(u_1)\dots \alpha(u_n)$, and
analogously for~$\beta$.
The components of the game can now be defined as follows.
\begin{itemize}
\item~$V^n = V^n_1 = \{v^n_0\}\cup T_1 \cup T_2  \cup a\{a, b\}^{< n}$
  where $T_1 = 0\{0, 1\}^{< n}$ and $T_2 =  \overline{0}\{\overline{0}, \overline{1}\}^{< n} $ (so~$V^n_0 = \0$)\,;

\item~$E^n$ has edges
  \begin{itemize}
  \item~$v^n_0 \to 0$\,,
  \item~$u \to u0$ and~$u \to u1$ for any~$u \in 0\{0, 1\}^{< n-1}$\,,
  \item~$u \to u$ for any~$u \in 0\{0, 1\}^{< n}$\,,
  \item~$u \overline{0} \to u\overline{0}$ and~$u \overline{1} \to
	u\overline{1}$ for any~$u \in \overline{0}\{\overline{0},
	\overline{1}\}^{< n-1}$\,,
  \item~$u \to \alpha(u)$ and $\alpha(u) \to u$ for any~$u \in 0\{0, 1\}^{< n}$\,,
  \item~$u \to \beta(u)$ for any~$u \in a\{a, b\}^{< n}$\,;
  \end{itemize}
\item~$u \sim^n \beta(\alpha(u))$, for any~$u \in 0\{0, 1\}^{< n}$\,.
\end{itemize}

In the informal description above,~$T_1$ is induced by vertices in
$0\{0,1\}^{n-1}$ and~$T_2$ by vertices in
$\overline{0}\{\overline{0},\overline{1}\}^{n-1}$. Intermediate vertices are
those from~$a\{a,b\}^{n-1}$.

Clearly, tree cops can capture the robber on~$G^n$, so $\tw(G^n) =
2$. Let us convince ourselves that~$\ent{\game^n} = 2$. First, two
cops are needed already on the subgraph induced by~$0$ and~$00$. On the other hand,~$2$ cops
suffice to capture the robber. The strategy is to play on~$T_1$ in a
top-\mbox{}down manner. The robber choses a branch of~$T_1$ and the
cops play on that branch as on the path in the proof of
Proposition~\ref{prop_unbounded_tw}. Finally the robber is forced to visit an intermediate
vertex~$\alpha(u)$ for some~$u$. Note that the cops are placed on the robber vertex in every
move, hence when the robber is on~$\alpha(u)$,~$u$ is occupied by a cop, which
forces the robber to proceed to~$T_2$ in the next move. On~$T_2$, he is captured in the
same way as on~$T_1$.

The powerset graph~$\olG^n$ (see Figure~\ref{fig_unbounded_ent}) has~$\{0\}$
as a position and therefore also~$\{a\}$ and~$\{0, \overline{0}\}$. From
$\{0\}$, one possibility is to remain in~$\{0\}$, another is to go to~$\{a\}$.
In~$\game^n$, from~$a$, there are edges to~$0$ and to~$\overline{0}$, which are
indistinguishable, so, in~$\olG^n$, there is an edge~$\{a\} \to
\{0,\overline{0}\}$. From~$0$, the pebble can return to~$a$ and both from~$0$
and from~$\overline{0}$ it can move to~$0$, so in~$\olG^n$ we have edges
$\{0,\overline{0}\} \to \{a\}$ and~$\{0,\overline{0}\} \to \{0\}$. The
described structure is repeated in the lower levels, because from~$\{u\}$, for
$u\in 0\{0,1\}^{<n-1}$, there is an edge to~$\{u0\}$ and to~$\{u1\}$, and
analogously for~$\{\olu\}$.

Essentially,~$\olG^n$ has the same vertices as~$\game^n$. We can identify
$u\in 0\{0,1\}^{<n}$ with~$\{u\in 0\{0,1\}^{<n}\}$,~$\alpha(u)$ with
$\{\alpha(u)\}$, and~$\beta(\alpha(u))$ with~$\{\beta(\alpha(u))\}$.

It remains to prove that the entanglement of the powerset graphs is unbounded.
We adapt the proof from~\cite{BerwangerGraKaiRab12} for similar graphs and show that
$\ent(\olcalG^n) \ge n/2-1$. In the following, we identify vertices~$u$
and
$\alpha(u)$ for simplicity of explanation, which, obviously, does not change the
entanglement.

We show by induction on~$n$ that for every even~$n$, the robber can starting from vertex~$0$ or from vertex~$\ol{0}$
\begin{itemize}
    \item escape~$n/2-2$ cops and
    \item after the~$(n/2-1)$th cop enters~$\olcalG^n$,
	  \begin{itemize}
	  \item if started in~$0$, reach~$\ol{0}$, and
	  \item if started in~$\ol{0}$, reach~$0$\,.
	  \end{itemize}
\end{itemize}
This suffices to prove unboundedness, as the robber has a winning 
strategy on~$\olcalG^{n+1}$ in this case: he switches between the two subtrees of the
root.

For~$n=2$, it is trivial. Assume that the statement is true for some even~$n$
and consider the situation for~$n+2$. We need two strategies: one for~$0$ as
the starting position and one for~$\ol{0}$. By symmetry, it suffices to
describe only a strategy for~$0$. For a word~$u\in\{0,1\}^{\le n+1} \cup
\{\ol{0},\ol{1}\}^{\le n+1}$, let~$T^{u}$ be the subgraph induced by
the subtree of~$T_1$ rooted at~$u$ and by the corresponding subtree of~$T_2$. The robber can play
in a way such that the following invariant is true.
\begin{itemize}
\item[] If the robber is in~$T^{0xy}$, for~$x,y\in\{0,1\}$, and starts
from~$0y$, there are no cops on~$\{\ol{0},\ol{0x}\}$\,.
\end{itemize}
By induction, it follows from the invariant that~$\ol{0}$ and~$\ol{0y}$ are
reachable for the robber.

At the beginning, the robber goes to the (cop-\mbox{}free) subtree~$T^{000}$
via
the path~$(0,00,000)$ and plays there from~$000$ according to the strategy given
by the induction hypothesis for~$T^{000}$. Also by induction,~$\ol{000}$
remains reachable and thus so is~$\ol{0}$ via~$\ol{00}$. Either that play lasts
for ever (and we are done), or the~$(n/2 - 2)$nd cop comes to~$T^{000}$ and
the robber can reach~$\ol{000}$. While he is doing that, no cops can be placed
outside of~$T^{000}$ as the robber does not leave~$T^{000}$. 

Assume that the robber enters a tree~$T^{0xy}$, for
$x,y\in\{0,1\}$ which is free of cops (which is, in particular, the case at
the beginning). By symmetry, we can assume that~$x=y=0$. Further,
assume without loss of generality
that the robber enters~$T^{000}$ at~$000$. Either the play remains in
$T^{000}$ infinitely long (and we are done), or the~$(n/2-1)$-st cop enters
$T^{000}$ and the robber reaches~$\ol{000}$. Note that while the robber is
moving towards~$\ol{000}$, no cops can be placed outside of~$T^{000}$ as the
robber does not leave~$T^{000}$.

If the last cop is already placed, the robber goes to~$\ol{00}$ and then to
$\ol{0}$, which are not occupied by cops by the invariant, and we are done.
If the last cop is not placed yet, all cops are in~$T^{000}$, so the robber
runs along the path~$\ol{000},\ol{00},00,0,01,010$ to~$T^{010}$. Note that
the vertices~$\ol{0}$ and~$\ol{01}$ are not occupied by cops, so the invariant
is still true. The robber plays as in~$T^{000}$ and so on.
\end{proof}

\begin{figure}

\begin{tikzpicture}[bend angle=45,auto, scale=0.35]

\node (offset) at (-9,0){};

\tiny
\node (v_0) at (-6,0){$v^2_0$};

\node (0) at (-2, 0)  {$\bullet$}
	edge [loop left] (0);
\node (a) at (0, 0)  {$\bullet$};
\node (-0) at (2, 0) {$\bullet$};

\node (00) at (-6, -3) {$\bullet$}
	edge [loop left] (00);
\node (aa) at (-4, -3) {$\bullet$};
\node (-00) at (-2, -3) {$\bullet$};

\node (01) at (2, -3) {$\bullet$}
	edge [loop left] (01);
\node (ab) at (4, -3) {$\bullet$};
\node (-01) at (6, -3)  {$\bullet$};

\draw[-slim,bend left=20] (v_0) to (0.120);

\draw [-] (0.east) to (a.west);
\draw [->] (a.east) to (-0.west);
\draw [ind] (0.north east) to [bend left = 20] (-0.north west);

\draw [-] (0.south) to (00.north);
\draw [-] (0.south) to (01.north);
\draw [-] (-0.south) to (-00.north);
\draw [-] (-0.south) to (-01.north);

\draw [-] (00.east) to (aa.west);
\draw [->] (aa.east) to (-00.west);
\draw [ind] (00.north east) to [bend left = 20] (-00.north west);

\draw [-] (01.east) to (ab.west);
\draw [->] (ab.east) to (-01.west);
\draw [ind] (01.north east) to [bend left = 20] (-01.north west);


\node (v_0) at (-6+17,0){$v^2_0$};

\node (0) at (-2+17, 0)  {$\bullet$}
	edge [loop left] (0);
\node (a) at (0+17, 0)  {$\bullet$};
\node (-0) at (2+17, 0) {$\bullet$};

\node (00) at (-6+17, -3) {$\bullet$}
	edge [loop left] (00);
\node (aa) at (-4+17, -3) {$\bullet$};
\node (-00) at (-2+17, -3) {$\bullet$};

\node (01) at (2+17, -3) {$\bullet$}
	edge [loop left] (01);
\node (ab) at (4+17, -3) {$\bullet$};
\node (-01) at (6+17, -3)  {$\bullet$};

\draw[-slim,bend left=20] (v_0) to (0.120);

\draw [->] (0.east) to (a.west);
\draw      (a.east) to (-0.west);
\draw [<-] (0.north east) to [bend left = 20] (-0.north west);

\draw [-] (0.south) to (00.north);
\draw [-] (0.south) to (01.north);
\draw [-] (-0.south) to (-00.north);
\draw [-] (-0.south) to (-01.north);

\draw [->] (00.east) to (aa.west);
\draw      (aa.east) to (-00.west);
\draw [<-] (00.north east) to [bend left = 20] (-00.north west);

\draw [->] (01.east) to (ab.west);
\draw      (ab.east) to (-01.west);
\draw [<-] (01.north east) to [bend left = 20] (-01.north west);

\end{tikzpicture}
\caption{Game graph $\mathcal{G}_2$ and its powerset graph $\overline{G}_{v_0}^2$.}
\label{fig_unbounded_ent}

\end{figure}
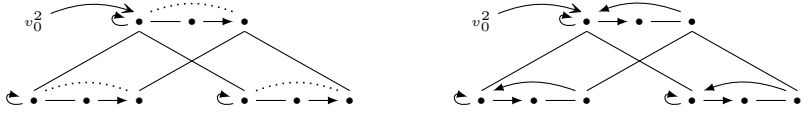

\subsection{Some positive results}\label{subsec_pos_results}

Now we prove that in contrast to tree-width and entanglement, \emph{non-\mbox{}monotone} 
DAG-width is preserved by the powerset construction. 

\begin{proposition}\label{proposition_nonmonotonedagwidth}
Let~$\game = (V,V_0,E,v_0,\sim,\col)$ be a parity game with imperfect
information such that the size of the~$\sim$-classes is bounded by some~$r$.
If~$\nmdw(G) \le k$, then~$\nmdw(\olG) \le k \cdot r \cdot 2^{r-1}$.
\end{proposition}
\begin{proof}
First, we describe our proof idea informally. We follow a play on~$G$
that corresponds to a set of at most~$r$ plays on~$\olG$ which Player~0
considers possible in the parity game with imperfect information. We translate
robber moves from~$\olG$ to the plays on~$G$, look for the answers of
the cops prescribed by their winning strategy for~$G$ and translate them
back to~$\olG$ combining them into one single move.

A position in the parity game on~$\olG$ corresponds to at most~$r$ positions
in the parity game on~$G$, so if the robber occupies a vertex~$[v]$
in~$\olG$, we consider, for any~$w\in[v]$, the possibility that the robber
occupies~$w$ in~$G$. Some plays we considered until some position may
prove to be impossible when the play evolves, some
plays may split in multiple plays. For any robber move to~$w\in[v]$, the
strategy for the cops in the game on~$G$ supplies an answer, moving the cops
from~$C_w$ to~$C'_w$. All these moves are translated into a move in which the
cops
occupy precisely the vertices of~$\olG$ that include a vertex from
some~$C'_w$.
These moves of the cop player on~$\olG$ can be realized with~$k \cdot r \cdot
2^{r-1}$ cops. 

The rough argument why robber moves can indeed be translated from~$\olG$ to~$G$
is that, by Lemma~\ref{lemma_powerset}, for any path~$\olu^0,  \olu^1, \dots,
\olu^t$ in $\olG$ and for any~$u^t \in \olu^t$, there is a path~$u^0, u^1,
\dots, u^t$ in~$G$ such that~$u^i \in \olu^i$ for any~$i \in \{0, \ldots,
t\}$. It also follows that if a play is infinite on~$\olG$, then at least one corresponding play
on~$G$ is infinite as well. Hence, if we start from a winning strategy for~$k$ cops for
the game on~$G$, no strategy for the robber can be winning against~$k \cdot r \cdot
2^{r-1}$ cops on~$\olG$.

Now we give a more formal proof.
Let~$f$ be a winning strategy for~$k$ cops for the \dagw game
on~$G$ (positional strategies suffice) and let~$\olg$ be any strategy for the robber for the \dagw game
on~$\olG$. We construct a play~$\olpi_{f\olg}$
on~$\olG$ that is consistent with~$\olg$ (and depends on~$f$), but is won by the cops.
The proof is by induction on the length of the finite prefixes (\ie histories) of~$\olpi_{f\olg}$.
While constructing~$\olpi_{f\olg}$ we simultaneously construct, for every 
history~$\olpi$ of~$\olpi_{f\olg}$ of length~$i$ a finite
tree~$\zeta(\olpi)$ whose branches are histories of length at most~$i$ in
the \dagw game on~$G$, such that the following conditions hold.  Let 
\[\olpi = \perpcdot(\olC_1,\olv_1)(\olC_1,\olC_2,\olv_1)(\olC_2,\olv_2) \dots
(\olC_{i},
\olv_{i})\]
(if it ends in a cop position), or
\[\olpi = \perpcdot(\olC_1,\olv_1)(\olC_1,\olC_2,\olv_1)(\olC_2,\olv_2) \dots
(\olC_i,\olC_{i+1}, \olv_i)\]
(if it ends in a robber position).
\begin{itemize}
\item[(1)] Each history in~$\zeta(\olpi)$ is consistent with~$f$. 

\item[(2)] $\olv_j = \{v\in V \mid \text{ at level } j
\text{ of } \zeta(\olpi)\text{, there is } (C,v)\text{ or } (C,C',v)
\}$ for all $j\le i$.
Moreover, for each~$v \in V$, on each level there is at most one position of the
form~$(C, v)$ or~$(C, C', v)$. 

\item[(3)]\label{item_how_to_place} For all~$j \le i+1$,
$\olC_j = \{ \olw\in \olV \mid  \text{at level }j\text{, there is }
(C,C',v)\text{ or } (C',v) \text{ with }\\ \olw\cap C' \neq \0\}$.

\item[(4)] Let~$\olpi'$ be a prefix of~$\olpi$. If~$\zeta(\olpi')$ has depth~$r$, 
then~$\zeta(\olpi)$ has depth at least~$r$ and up to level~$r$,~$\zeta(\olpi')$ and~$\zeta(\olpi)$
coincide.
\end{itemize}

To begin the induction, consider any play prefix~$\olpi$ of length~$1$, \ie
any possible initial move~$\perp \to (\0, \olw)$ of the robber player.
With~$\olpi$ we
associate the tree~$\zeta(\olpi)$ consisting of the root~$\perp$ with
successors
$(\0, v)$ for~$v \in \olw$. Clearly, conditions (1)--(4) hold. 

For the translation of the robber moves in the induction step, consider a
play prefix \[\olpi = \perpcdot (\olC_1,\olv_1)(\olC_1,
\olC_2,\olv_1)(\olC_2,\olv_2)
\ldots (\olC_i, \olC_{i+1}, \olv_i)\] with~$i \ge 1$ and let, by induction
hypothesis,~$\zeta(\olpi)$ be constructed up to level~$2i$. Consider a robber
move from~$\olv_i$ to~$\olv_{i+1}$, so~$\olv_{i+1} \notin \olC_i$
and~$\olv_{i+1}$ is reachable from~$\olv_i$ in the graph~$\olG -
(\olC_i \cap \olC_{i-1})$. Let~$\olv^0, \olv^1, \ldots, \olv^t$ be
a path from~$\olv_i = \olv^0$ to~$\olv_{i+1} = \olv^t$ in~$\olG - (\olC_i
\cap \olC_{i-1})$. Then by Lemma~\ref{lemma_powerset}, there are~$v  \in
\olv_{i+1}$ and~$u \in \olv_i$, and a path~$u^0, u^1, \ldots, u^t$ from~$u =
u^0$ to~$v= u^t$ in~$G$ with~$u^l \in \olv^l$, for~$l = 0,\ldots, t$.
Let~$W$ the set of all such~$v$. By Conditions~(2) and~(4) for~$\zeta(\olpi)$,
there is some history~$\pi \in \zeta(\olpi)$ which ends in a position
$(C^v_i, C^v_{i+1}, u)$. So,~$C^v_i$ corresponds to~$\olC_i$ and~$C^v_{i+1}$
corresponds to~$\olC_{i+1}$ in the sense of Condition~(3). We now extend~$\pi$
to the history~$\pi\cdot (C^v_{i+1},v)$. The set of all such
histories extended in this way by~$(C^w_{i+1},w)$ for all~$w\in W$ forms the
tree~$\zeta(\olpi)$.

We have to show that each such move to~$(C^v_{i+1}, v)$ is possible, 
\ie that~$v \notin C^v_{i+1}$ and~$v$
is reachable from~$u$ in~$G - (C^v_i \cap C^v_{i+1})$. As~$\olv_{i+1} \notin
\olC_i$, by Condition~(3), we have~$\olv_{i+1} \cap C^v_{i+1} = \0$, which
implies~$v \notin C^v_{i+1}$. Now assume towards a contradiction that~$v$ is not
reachable from~$u$ in~$G - (C^v_i \cap C^v_{i+1})$. Then there is
some~$l \in \{1, \ldots, t\}$ such that~$u^l \in C^v_i \cap C^v_{i+1}$ 
(notice that~$u^0 = u \notin C^v_i \cap C^v_{i+1}$, otherwise the position
with~$u$
would not be legal as is given by induction). Then since~$u^l \in \olv^l$, we
have
$\olv^l \in \olC^v_i \cap \olC^v_{i+1}$, by~(3), which contradicts the fact that
$\olv^0,
\olv^1, \ldots, \olv^t$ is a path in~$\olG - (\olC^v_i \cap \olC^v_{i+1})$.

We check that Conditions~(1)--(4) hold after the construction.
For Conditions~(2), (3) and~(4) this is obvious.
For~(1), since all play prefixes in~$\zeta(\olpi)$ up to level~$2i$ are
consistent
with~$f$ by induction and all extensions of the play prefixes are robber
moves,
all play prefixes in~$\zeta(\olpi)$ are still consistent with~$f$.

To translate the answer of the cops, assume that we have already
constructed~$\zeta(\olpi)$
up to level~$2i+1$, for some~$i\ge 0$. Note that there are at most~$r$ branches
of 
length~$2i+1$. Let~$W$ be the set of robber vertices in the last positions of
those
branches. For any maximal branch~$\pi^v$ of~$\zeta(\olpi)$ ending with a
position 
with robber vertex~$v$ where
\[ \pi^v = \perpcdot (C^v_1,v_1)(C^v_1,C^v_2,v_1)\dots (C^v_i,v_i)\,,\]
$i\ge 1$ and~$v=v_i$, consider the set~$C^v_{i+1} = f((C^v_i,v_i))$ of
positions 
chosen to be occupied by the cops in the next move according to~$f$. 
We define~$\olC_{i+1}$ by 
\[\olC_{i+1} = \{ [u]\in \olV \mid [u] \cap \bigcup_{v\in W}C^v_i \neq \0
\}\,,\]
\ie the cops occupy those~$[w]$ that contain a vertex from some~$C^v_i$.

This yields the play prefix~$\olpi' = \olpi\cdot(\olC_{i+1}, \olC,\olv_{i+1})$
and we associate an extension~$\zeta(\olpi')$ of~$\zeta(\olpi)$ with it.
The extension is obtained by appending position~$(C^v_i,C^v_{i+1},v)$ 
to each branch of length~$2i+1$ ending with a position with robber vertex~$v$.
It is trivial that all Conditions~(1)--(4) hold.

Assume that~$\olpi_{f\olg}$ is infinite, \ie won by the robber. 
Then~$\zeta(\olpi_{f\olg})$ is infinite as well. Since~$\zeta$ is finitely
branching, by K\"onig's Lemma, there is some infinite path~$\pi$
through~$\zeta$.
By Condition~(1),~$\pi$ is a play in the \dagw game on~$G$ which is
consistent with~$f$. Since~$\pi$ is infinite, this contradicts the fact
that~$f$ is a winning strategy for the cop player. 

It remains to count the number of cops used by the cop player 
in~$\olpi_{f\olg}$. Consider any
position~$(\olC_i, \olC_{i+1}, \olv_i)$ occurring in~$\olpi_{f\olg}$.
By Condition~(2), at level~$2i$ of~$\zeta(\olpi)$, there occur at most
$|\olv_i| \le r$ many play prefixes. Each such play prefix is consistent
with~$f$, so at most~$k$ vertices are occupied by the cops. Hence, by
Condition~(3),~$|\olC_{i+1}| \le k \cdot r \cdot 2^{r-1}$ (note that there are
$2^{r-1}$ subsets of set with~$r$ elements which contain a fixed vertex).
\end{proof}

We stress that this strategy translation does not necessarily preserve monotonicity, 
as the following example shows.

\begin{example}\label{example_DAG_monotonicity_not_preserved}
We give an example where the strategy translation from Proposition~\ref{proposition_nonmonotonedagwidth}
does not preserve monotonicity of the cop strategy. Consider the graph~$G$ depicted in 
Figure~\ref{fig_monDAGwidth} and the following monotone (partial) strategy for the cops.
First, put a cop on~$v_0$. If the robber goes to $1$, put a cop on~$1$ and then move the cop from~$1$ to~$3$. 
If the robber goes to~$2$, put a cop on~$5$ and if the robber goes to~$4$, put a new cop on~$4$.
In the game on the powerset graph, consider the following play, which is consistent with
the translated cop strategy. First, the cops occupy $\{v_0\}$. Let the robber go to $\{1,2\}$ 
in which case the cops occupy $\{1,2\}$ and $\{5\}$. Now the robber goes to $\{3,4\}$, so the cop 
from $\{1,2\}$ is removed. At this moment, the vertex $\{1,2\}$ becomes 
available for the robber again, so the translated strategy is non-monotone. Notice that, nevertheless,
$\dw(\overline{G}) = 2$.
\end{example}

\begin{figure}
\begin{tikzpicture}[bend angle=45,auto]
\node (offset) at (-3,0) {};

\node (v_0) at (0,0)	{$v_0$};
\node (1)   at (-1,-1)	{$1$};
\node (2)   at (1,-1)	{$2$};
\node (3)   at (-1,-2)	{$3$};
\node (4)   at (1,-2)	{$4$};
\node (5)   at (2,-1)	{$5$};

\draw [->]  (v_0) to (1);
\draw [->]  (v_0) to (2);
\draw [->]  (1)   to (3);
\draw       (2)   to (4);
\draw       (2)   to (5);
\draw [->]  (4)   to (1);
\draw [ind] (1)   to (2);
\draw [ind] (3)   to (4);


\node (v_0) at (4,0)	{$\{v_0\}$};
\node (12)  at (4,-1)	{$\{1,2\}$};
\node (34)  at (4,-2)	{$\{3,4\}$};
\node (5)   at (5.5,-1)	{$\{5\}$};
\node (2)   at (7,-1)	{$\{2\}$};
\node (4)   at (5.5,-2)	{$\{4\}$};

\draw [->]  (v_0)   to (12);
\draw       (12)    to (34);
\draw [->]  (12)    to (5);
\draw       (5)     to (2);
\draw [->]  (2)     to (4);
\draw [->]  (4)     to (12);
 
\end{tikzpicture}
\caption{Monotone strategy is translated to a non-monotone one.}
\label{fig_monDAGwidth}
\end{figure}
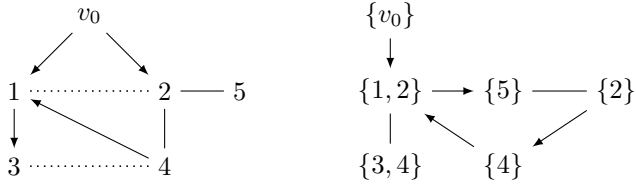

Thus our construction does not guarantee that the \dagw of
the powerset graph is bounded in the \dagw of the original graph and we cannot
conclude that a bound on \dagw allows us to solve parity games in
polynomial time. Although not actually our goal, we can 
consider even stronger conditions on the
structural complexity of given graphs. In the remaining of the section
we show that \emph{\dpathw} is bounded by a construction similar to
that from the proof of Proposition~\ref{proposition_nonmonotonedagwidth} (so the 
conclusion is also stronger). 

For \dagw we give two solutions, each leading to a result that is also
of independent interest. In Section~\ref{sec_simulated_pg} we describe
how to use a technique by Fearnley and Schewe
from~\cite{FearnleySch12} for solving parity games on graphs where the
\dagw is not necessarily bounded, but the \emph{non-{}monotone} \dagw
is. Thus we obtain a stronger result: parity games with imperfect
information can be solved in \ptime on classes of graphs of bounded
\nmon \dagw. In particular, this holds for graphs of bounded
\dagw.

In the remaining sections we go still another way to prove the
latter result. Although it is more cumbersome than the solution
following Fearnley and Schewe, we also present it because it
enlightens the connection between bounded imperfect information and
graph searching. It also contains some results on
graph searching that are independent of solving parity games.

\begin{proposition}\label{proposition_directedpathwidth}
Let~$\game = (V,V_0,E,v_0,\sim,\col)$ be a parity game with imperfect
information in which the size of the~$\sim$-classes is bounded by some~$r$.
If~$\dpw(G) \le k$, then~$\dpw(\olG) \le k \cdot 2^{r-1}$.
\end{proposition}
\begin{proof}
Let~$f$ be a monotone winning strategy for~$k$ cops in the \dpathw game
on~$\game$ and let 
\[\pi = \perpcdot (C_0, C_0, R_0) (C_0, C_1, R_1)\ldots (C_{n-1}, C_n, R_n)\] 
be the unique play which is consistent with~$f$. It is finite, as~$f$
is winning. Recall that the \dpathw game is,
essentially, a one player game and there is a bijection between strategies and
plays, so it suffices to construct a (not necessarily monotone) play 
\[\olpi = \{\perp \} \cdot (\olC_0, \olC_0, \olR_0) (\olC_0, \olC_1, \olR_1)
(\olC_1,
\olC_2,
\olR_2) \ldots\] 
of the game on~$\olG$ that is won by the cops where, for  all
$i$, we have~$|\olC_i| \le k2^{r-1}$. We construct~$\olpi$ inductively by the
length of its finite prefixes such that the following invariant holds.
\begin{enumerate}
\item~$\olC_0 = \0$ and~$\olR_0 = \olV$ (at the beginning, there are no cops in
the graph and the robber occupies the whole graph),
\item~$\olR_n = \0$ (at the end, the robber is captured),
\item~$\olR_{i+1} = \Reach_{\olcalG - (\olC_{i} \cap \olC_{i+1})}(\olR_i)
\setminus \olC_{i+1}$, for all~$i$ (every move is legal),

\item \label{prelast}$\olR_{i+1} \subseteq \olR_i$, \ie the play is monotone,

\item\label{last}~$\bigcup \olR_i \coloneqq \{v\in V \mid \text{ there is some }
\olw \in
\olR_i \text{ with } v\in \olw\} \subseteq R_i$, \ie if the robber occupies a
vertex~$\olw$ in~$\olpi$ and~$v\in \olw$, then the robber occupies~$v$. Note
that vertices (positions) in~$\olpi$ are sets of vertices in~$\pi$.
\end{enumerate}

The last two properties of~$\olpi$ imply the statement of the proposition.
Indeed, by Property~\ref{prelast}, the play is monotone. Furthermore, the robber is
finally captured if and only if the play is finite and ends in a position
$(\olC_{n-1}, \olC_n,\olR_n)$ where~$\olR_n = \0$. Assume that~$\olpi$ is
infinite, then all~$\olR_i \neq \0$, but then there is some~$v$ and~$\olw$ with
$v\in\olw\in\olR_i$ such that~$v$ is not occupied in the~$i$th position
of~$\pi$ (by Property~\ref{last} of the invariant), so~$\pi$ is not winning, but
that contradicts the assumption.

The construction just follows the invariant. Let~$\olpi_0 =
\{\perp\}(\0,\0,\olV)$ and,
for~$i>0$, let~$\olpi_i = \olpi_{i-1} \cdot
(\olC_{i-1}, \olC_i, \olR_i)$ such that, for all~$\olv \in \olV$ we have
$\olv \in \olC_i$ if and only if~$\olv \cap C_i \neq \0$. In other words, we
place a cop on a vertex~$\olv$ in a position of~$\olpi$ if, in the corresponding
position of~$\pi$, we place a cop on some vertex in~$\olv$. As there are
$2^{r-1}$ subsets of~$C_i$ that contain a fixed vertex~$v$, the size of all
$\olC_i$ is at most~$k\cdot 2^{r-1}$.

It remains to show Properties~\ref{prelast} and~\ref{last} of the invariant. 
Assume that the play is not monotone, then there is some~$i$, 
some~$\olw \in \olR_i$ and some~$\olv \in \olC_i \setminus \olC_{i+1}$ such
that~$(\olw,\olv)\in\olE$, \ie a cop was removed from~$\olv$ and the robber
occupies~$\olv$ following one single edge. By Reif's construction, for all~$y\in
\olv$, there is some~$x\in \olw$ with~$(x,y)\in E$. As~$\olw \in \olR_i$, by
induction, we have~$\olw \subseteq R_i$. On the other hand, as~$\olv \in \olC_i
\setminus \olC_{i+1}$, by the construction of~$\olpi$, we have~$\olv\cap C_i\neq
\0$
and~$\olv\cap C_{i+1} = \0$. In other words, in~$\pi$ some vertex~$y$ of~$\olv$
is left by a cop and all vertices of~$\olw$ are occupied by the robber in the
$i$th move. However, the robber can move from~some~$x\in \olw$ to~$y$, which
causes non-\mbox{}monotonicity in~$\pi$,
but we assumed that~$\pi$ is monotone. Thus~$\olR_{i+1} \subseteq \olR_i$.

It remains to prove Property~\ref{last} of the invariant. Assume that it does not
hold and suppose,~$i$ is the least index with~$\bigcup \olR_i \not \subseteq
R_i$. Then there exist some~$\olw \in \olR_i$ and some~$v\in\olw$ such that
$v\notin R_i$. By induction hypothesis, the move from position~$i-1$ to
position~$i$ is monotone, so~$\olw\subseteq \olR_{i-1}$. Then by the choice
of~$i$,~$v\in R_{i-1}$. So we have~$v\in R_{i-1}\setminus R_i$ and thus~$v\in
C_i \setminus C_{i-1}$. By the construction of~$\olpi$, we have~$\olw \in
\olC_i$,
a contradiction to~$\olw \in \olR_i$.
\end{proof}

\begin{corollary}
Parity games with bounded imperfect information can be solved in polynomial time on graphs of bounded
directed path-width.
\end{corollary}

Finally, we remark that our direct translation of the robber moves
back to the game on $\graphG$ cannot be immediately applied to the
games which define \kellyw and directed \treew. In the
Kelly-width game, the robber can only move if a cop is about to occupy
his vertex. It can happen that the cops occupy a vertex
$\{v_1,\dots,v_l\}$ in $\olG$ but not all vertices $v_1$, \dots, $v_l$
in $G$. In the directed tree-width game, the robber is not permitted
to leave the strongly connected component in which he currently is,
which again obstructs a direct translation of the robber moves from
$\olG$ back to $G$. Furthermore, it is not known whether parity games
with perfect information can be sovled in polynomial time if directed
\treew
is bounded.

\section{Simulated parity games }\label{sec_simulated_pg}

Simulated parity games were introduced by Fearnley and Schewe
in~\cite{FearnleySch11,FearnleySch12}.  The idea of a simulated parity
game is to decompose the original game into smaller games such that
one can control cycles that appear when the game
pebble revisits a vertex more efficiently. The simulated game (in that
both players have perfect information) starts on a small
subgame~$S$. If a cycle is reached within that subgame, the game stops
and the winner is determined as in the usual parity game. Otherwise
consider the first visited vertex $v\notin S$. One of the players (it
does not matter which one, say, Player~0) gives some promise: he
claims, for every vertex in~$S$, that he can guarantee a certain
value, the best color, in any play from~$v$ to that vertex. More
precisely, for every vertex $w\in S$, he announces a color~$c$ and
asserts that no worse color will be seen in a play from~$v$ to~$w$
if~$w$ will be the first vertex in~$S$ visited from now on. Player~$1$
either accepts for some vertex~$w$ in~$S$, then the game continues
from~$w$ and the minimum color seen since~$v$ is set to~$c$, or he
rejects. In the latter case the game continues in a next small
subgame~$S'$ containing~$v$. This continues in the same way, except that now when the
play leaves~$S'$, the assertions of Player~$0$ for~$S'$ are added to
those for~$S$. A play ends either if it reaches a vertex for that
Player~$0$ promised a color (and Player~$0$ wins if and only if he
could keep his promise), or a cycle is closed (then the parity
condition applies). This idea of closing cycles is similar to the idea
from~\cite{BerwangerGra05} of alternating cycle detection while
playing the entanglement game.

The game is parametrized by two functions: function $\Next$ determines the next subgame $S'$ and
function $\Hist$ forgets some of the promises of Player~$0$ and is used for optimization.
If every play of the simulated game is finite (intuitively, $\Hist$ does not forget too much), then
Player~$0$ wins the original game if and only if he wins the simulated game
(from the same vertex). Examples of $\Next$ and $\Hist$ are given
in~\cite{FearnleySch12}: $\Hist$ forgets every promise except those from the
last subgame and $\Next$ follows a tree decomposition or a DAG decomposition.

If we construct $\Next$ and $\Hist$ such that they fulfill some
conditions on certain classes of graphs, then we can solve parity
games in \ptime on those classes. The conditions are:
\begin{itemize}
\item Every play of the simulated game is finite.
\item There is a data structure to store the promises of Player~$0$ that uses only 
a logarithmic amount of space in the size of the parity game.
\end{itemize}
An alternating Turing machine just plays the simulated game and
determines the winner. In the rest of the section we describe the
simulated game formally and prove that for graph classes where a
bounded number of cops can capture a robber (not necessarily in a
monotone way) we can indeed find appropriate functions $\Next$ and
$\Hist$.

Let $\game = (V,V_0,E,v_0,\Omega)$ be a parity game with perfect
information.
The \emph{significance order} on the set of colors is defined by
$a \prec b$ if~$a$ is better for Player~$0$ than~$b$, \ie~$a$ is even and~$b$ is odd,
or both are even and~$a<b$, or both are odd and~$a>b$. 
For a positional 
strategy~$f$ of Player~$\sigma$, a set~$F$ of vertices and two vertices~$s$ and~$t$, let
$\Paths_f^F(s,t)$ be the set of paths from~$s$ to~$t$ avoiding~$F$
(except of~$t$ if $t\in F$) and consistent with~$f$, \ie if~$f$ is a 
strategy for Player~$\sigma$, then for all consecutive~$v$ and~$w$ on
a path in $\Paths_f^F(s,t)$, if~$v$ is 
a vetrex of Player~$\sigma$, then~$f(v)=w$. For a sequence~$P$ of vertices let $\mincol(P)$ be the 
minimal with respect to ${}<{}$ color appearing in~$P$.
We denote the best possible color that a strategy~$f$ 
guarantees on $\Paths_f^F(s,t)$ by $\best_f^F(s,t)$, \ie 
\begin{align*}
\best_f^F(s,t) = \opt_{P\in\Paths_f^F(s,t)} \mincol(P)
\end{align*}
where $\opt = \max$ if~$f$ is a Player~$0$ strategy and $\opt=\min$
if~$f$ is a Player~$1$ strategy, both with respect to the significance
order.
Let~$C$ be the set of used colors. A \emph{strategy profile} for a set of vertices~$F$ is a function
$\Profile_{f,s}^F\colon F \to C\cup \{\nix\}$ defined~by:
\[ \Profile_{f,s}^F(t) =
\begin{cases}
  \nix 	&	\text{if }\Paths_{f,s}^F(t) = \0,\\
  \best_f^F(s,t) & \text{otherwise.}
\end{cases}
\]
An \emph{abstract profile} is a function $P_s^F\colon F\to C \cup
\{\nix\}$. A profile is what a strategy~$f$ can actually guarantee, an
abstract profile is what Player~$0$ promises when the play leaves the
current subgame. In particular $P_s^F(t) = \nix$ means that~$t$ should
not be reached at all. Of course, Player~$0$ is free to promise something
that cannot be guaranteed by any of his strategies.

When a play of the simulated game returns to a subgame it left in the
past, we must have stored enough information to check whether
Player~$0$ could keep his promise. The data structure for this
is a history, which is a set of records.\footnote{Note that we
  redefined the notion of a history for this Section.} A \emph{record}
is a triple~$(F,c,P)$ where $F\subseteq V$,~$P$ is an
abstract profile for~$F$ and~$c$ is a color in~$C$. 
Hereby~$P$ stores a promise of Player~$0$ and~$c$ is the minimal color
seen since the promise was made.
For a record~$(F,c,P)$ and a color~$c'$ define $\Update((F,c,P),c')$ by
$(F,\min(c,c'),P)$ where the minimum is with respect to ${}<{}$ and $\min(\nix,c') = c'$, for all~$c'\in C$. A
\emph{history} is a set~$\calH$ of records. For a history~$\calH$ and
a color~$c$, we define $\Update(\calH,c')$ by $\{\Update(F,c,P),c')
\mid (F,c,P) \in \calH\}$.

The function $\Next$ maps a tuple $(S,v,\calH)$ where~$S\subseteq
V$, $v\in V\setminus S$ and~$\calH$ is a history to a set~$S'\subseteq V$. $\Hist$ is a history 
updating function: it deletes some elements
from a given history, 
\ie $\Hist(\calH) \subseteq \calH$. The game
$\Simulate_\game(S,\calH,v,\Next,\Hist)$ is played on~$G$ as
follows.

The positions of the simulated game are of the form $(S,v,\calH,\Pi)$
where~$v\in V$ and~$\Pi$ is a sequence of triples~$(u,c,w)$
with~$u,w\in V$ and~$c\in C$.  Hereby~$v$ is the current vertex
and~$\Pi$ stores the simulated play prefix played so far. By abuse of notation
we apply $\mincol$ also to sequences of colors with the obvious meaning and extend it
to sequences of triples $(u,c,w)$:
$\mincol\big((u_1,c_1,w_1),\ldots,(u_n,c_n,w_n)\big) =
\mincol(c_1,\ldots,c_n)$.  Furthermore, if~$\Pi$ has the form
\[(u_1,c_1,w_1)\ldots(u_m,c_m,w_m)(u_{m+1},c_{m+1},w_{m+1})\ldots(u_n,c_n,
,w_n)(u_m,c_m,w_m)\,,\] \ie it ends in a cycle, then define
\[\Winner(\Pi) = \mincol\big((u_m,c_m,w_m),\ldots,(u_n,c_n,w_n)\big)
\mod 2\,.\] 
The game is played in rounds.  Let $(S,v,\calH,\Pi)$ be the current position. A round consists of
the following steps.
\begin{enumerate}
\item\label{step_move} If $v\in V_\sigma$, then Player~$\sigma$ chooses some
  $v'\in vE$.
\item If $v'\in S$ or $v'\in F$ for some $(F,c,P)\in \calH$, then $\Pi'
  = \Pi \cdot (v,\col(v'),v')$, the new position is $(S,v',\calH,\Pi')$ and
  the play continues from Step~\ref{step_check_end}.
\item\label{step_announce_profile} Player~$0$ announces an abstract
  profile~$P'$ for~$v'$ and~$S$.
\item\label{step_believe_or_not} Player~$1$ can play \accept and
  choose some $w\in S$ with $P'(w)\neq \nix$, or play \reject.
  \begin{itemize}
  \item If Player~$1$ chooses \accept, then the next position is
    $(S,w,\calH,\Pi'')$ where $\Pi'' = \Pi' \cdot (v',\min(\col(v'),w),w)$.
  \item If Player~$1$ chooses \reject, then the history and~$S$ are
    updated as follows and the play continues from
    Step~\ref{step_check_end}:
    \begin{itemize}
    \item $\calH' = \Update(\calH,\col(v'))$;
    \item $\calH'' = \calH' \cup \{(F,\col(v'),P')$\};
    \item $\calH''' = \Hist(\calH'')$;
    \item $S' = \Next(S,v',\calH''')$;
    \item The winner is the winner of
      $\Simulate_\game(S',\calH''',v',\Next,\Hist)$.
    \end{itemize}
  \end{itemize}
\item\label{step_check_end} If $v'\in F$ for some $(F,c,P)\in \calH$,
  then the game stops. Let the current~position be $(S,v',\calH,\Pi^*)$.
  Player~$0$ wins the play if $\min(\mincol(\Pi^*),c) \le P(v')$\,;
  Player~$1$ wins the play if either $\min(\mincol(\Pi^*),c) > P(v')$
  or $P(v') = \nix$\,.
\item\label{step_cycle} If $\Pi^*$ ends with a cycle, then the winner
  of the play is $\Winner(\Pi)$.
\end{enumerate}
The initialization round is played as follows. If~$v_0\in S$, then the
play starts in the first regular round. If $v_0\notin S$, Player~$0$
announces an abstract strategy profile~$P$ for~$v_0$ and~$S$ and the
history is initialized with $\calH = \{(v_0,\col(v_0),v_0)\}$. 

\begin{theorem}[\cite{FearnleySch12}]\label{thm_fearnleySchewe}
  Let $\game = (V,V_0,E,v_0, \col)$ be a parity game. Then, for all~$S$, $\calH$, $\Next$, and $\Hist$, if all plays of
  $\Simulate_\game(S,\calH,v_0,\Next,\Hist)$ are finite, then Player~$0$
  has a winning strategy for~$\game$ if and only if Player~$0$ has a
  winning strategy for $\Simulate_\game(S,\calH,v_0,\Next,\Hist)$.
\end{theorem}

Now we prove the main result of this section. The proof is an
adaptation of the proof of Theorem~$6.4$ from~\cite{FearnleySch11}. We
show that a winning strategy of cops in the cops and robber game
induces functions $\Next$ and $\Hist$ that satisfy the conditions of
Theorem~\ref{thm_fearnleySchewe}. In addition, the resulting play can
be solved in deterministic polynomial time, which leads to an
efficient solution of parity games with perfect information on graphs
where non-{}monotone \dagw is bounded and thus also of parity games
with bounded imperfect information on those graphs.

\begin{theorem}\label{thm_simulated}
Let $\game = (V,V_0,E,v_0,\Omega)$ be a parity game. Let~$k$ cops 
have a  (not necessarily monotone) strategy in the cops and robber game on~$G$ 
that guarantees a capture of the robber. Then there is are functions $\Next$ 
and $\Hist$ ans some $S\subseteq V$ such that for all~$\calH$ and $v_0$ the game $\Simulate_G(S,\calH,v,\Next,\Hist)$ 
has no infinite plays. Furthermore, given~$\game$, in deterministic
polynomial time in the size of~$\game$, we can construct a
representation of the simulated game and solve it.
\end{theorem}
\begin{proof}
It is clear that in the cops and robber game, positional strategies
for both players suffice. Let~$f$ be a positional strategy for~$k$
cops that guarantees a capture of the robber. For every position of
the simulated game $\Simulate_G(S,\calH,v,\Next,\Hist)$ with $v\notin
S$ we define $\Next(S,v,\Hist) = f(S,v)$. We define~$S$ to be the
answer of the cops to the first move of the robber to~$v_0$, \ie $S =
f(\0,v_0)$.  The function $\Hist$ forgets all records from the history
up to the last one. If $|\calH| = 1$, then $\Hist(\calH) =
\calH$. Otherwise $|\calH| = 2$, and, by the definition of $\Next$,
there are some $F_1$, $F_2$, $c_1$, $c_2$, $P_1$ and $P_2$ such that
$\calH = \{(F_1,c_1,P_1),(F_2,c_2,P_2)\}$ and $f(F_1,v) = F_2$,
where~$v$ i the current vertex. We set $\Hist(\calH) =
\{(F_2,c_2,P_2)\}$.

We show that all plays of $\Simulate_G(S,\calH,v_0,\Next,\Hist)$ are
finite. Note that since~$f$ can be non-{}monotone, it is
possible that the pebble in the simulated game returns to a vertex it
left in the past, but the play does not stop because we forgot the
promise of Player~$0$ for that vertex.

Assume that there is an infinite play~$\pi$  of
$\Simulate_G(S,\calH,v_0,\Next,\Hist)$. We describe an infinite
play~$\pi'$ of the cops and robber game that is consistent with~$f$. 
Let $(v_1,w_1),(v_2,w_2)\ldots$ be the infinite sequence of all pebble
moves in the simulated game with $v_i\in S_i$ and $w_i\notin S_i$
where~$(S_i)_{i \ge 1}$ is the sequence of the subgames appearing
in~$\pi$. Then~$\pi'$ is the play in which the robber
chooses~$v_0,w_1,w_2,w_3,\ldots$ and the cops play according
to~$f$. Then~$\pi'$ is infinite, which is a contradiction because~$f$
is winning, but we still have to show that~$\pi'$ is well-{}defined, \ie
that all robber moves  are possible.

The first robber move to~$v_0$ is trivially possible. The cops answer
occupying~$S_1 = f(\0,v_0)$. As $w_1\notin S_1$, there is no cop
on~$w_1$. Because the simulated game proceeded from~$v_1$ to~$w_1$
(not necessarily in one move), there is a path from~$v_1$ to~$w_1$, so
robber move to~$w_1$ is possible.  The same argument applies for
all~$w_i$ with~$i>1$.

To solve $\Simulate_\game(S,\calH,v_0,\Next,\Hist)$ we construct an alternating Turing machine that 
just plays the game. This can obviously be done in polynomial time. We have to prove that the Turing machine uses only a
logarithmic amount of space in the size of~$\game$. As the cops and robber game
admits positional winning strategies for both players, we can assume that if a position of the game repeats in a play, then the
robber wins the play. There are $|V|^k$ possible cop placements and at most~$|V|$
possible robber placements, \ie at most $|V|^{k+1}$ positions. 

The data structures are variables $S$, $\calH$, $\Pi$ and $v$.
\begin{itemize}
\item By construction $S$ is a cop placement, so we need $k\cdot \log|V|$ bits to store it.

\item In any position~$\calH$ always contains only one record $(F,c,P)$ with
$|F|\le |V|$ and can be stored using $\log|V|$ bits for $F$, $\log|C|$ bits for
$c$ and $k\cdot \log|C|$ bits for $P$ because $P$ contains at most $k$ pairs
$(w,c_w)$ with $w\in F$ and $c_w\in C\cup\{\nix\}$. One can represent $P$ as a
list of length at most $k$ of colors from $C\cup\{\nix\}$. Thus $\calH$ can be
stored using $\log|V| + k\cdot(\log|C|+1)$ bits.

\item We need $\log|V|$ bits to store $v$.

\item The variable~$\Pi$ is a sequence of at most~$|S|$ tuples $(v,c,w)$. If
$(v_1,c_1,w_1)$ and $(v_2.c_2,w_2)$ are two consecutive tuples in~$\Pi$, then
$w_1 = v_2$, so we only have to remember pairs $(c,w)$ where~$c$ is a color and~$w$
is one of at most~$|S|$ vertices. As $|S|\le k$, we need at most $k\cdot
(\log|C| + \log k)$ bits for~$\Pi$. Note that although in the initialization
round $v\notin S$ is possible, we do not need to memorize~$v$ because in this
case there can be no closed cycle in~$S$ containing~$v$.
\end{itemize}

Summing up, the alternating Turing machine needs at most
$(k+2)\cdot\log|V| + 2k\cdot\log|C|+k(\log k+1)$ bits. This leads to a
deterministic algorithm running in time $\calO(|V|^{k+2}\cdot
|C^{2k}|\cdot k^k)$.
\end{proof}

The powerset construction produces a graph~$\olG$ that is only polynomially
larger than the original graph~$G$. By Proposition~\ref{proposition_nonmonotonedagwidth},
the \nmdagw of~$\olG$ is bounded in the \dagw of~$G$, so as a corollary from 
Theorem~\ref{thm_simulated} we obtain the following result. 
\begin{corollary}
\label{cor_main}
 Parity games with bounded imperfect information can be solved in polynomial time on graphs
of bounded \dagw.
\end{corollary}

\section{Bounded imperfect information and multiple robbers}\label{sec-parity}


In this section we follow another approach to prove Corollary~\ref{cor_main}.
We translate imperfect information bounded by some constant~$r$ 
into a new graph searching game by introducing~$r$ robbers instead of one. 
This game is a generalization of a similar helicopter cops and robber game with multiple robbers
(that we refer to as a helicopter game, for short) defined by Richerby and Thilikos
in~\cite{RicherbyThi09}. The helicopter game is played on an undirected graph by a team 
of~$k$ cops and a gang of~$r$ robbers where~$k$ and~$r$ are parameters of the game. 
The cops move as in the cops and (single) robber game (up to a non-{}essential new kind
of \emph{sliding} cop moves), and each robber moves independently
of the others also as in the game with a single robber. If a robber is captured,
he is taken away from the graph. When all robbers are captured, the cops win.
Infinite plays, in which at least one robber survives for ever, are won by the robbers.
The cops also lose if they perform a non{}-monotone move, \ie if a robber can reach
a vertex that was previously unavailable for the robbers. Richerby and Thilikos show
that the number of additional cops needed to capture a team of robbers with one additional
robber grows at most logarithmically in~$r$.

Our game differs from the helicopter game in three aspects.  First, we
do not allow sliding moves, but this can introduce a difference in the
cop number by at most one.  Second, we play on directed graphs, and we
will see that this permits the robbers to coordinate their efforts in
a new way to escape from the cops. Third, in our game the robbers can
jump to each other, \ie a robber can leave his vertex and play from
the vertex occupied by another robber.  This rule may seem somewhat
unnatural, but we introduce it for several reasons.  First, we will
see that this rule supplies the robbers with more power. In
particular, the logarithmic upper bound from~\cite{RicherbyThi09} does
not hold any more. We however show that the number of additional cops
is bounded in~$r$ and grows at most linearly in~$r$, which is our main
result about graph searching games. Hence the additional power of the
robber gang makes the boundedness result stronger (for the cost of a
worse bound). The second reason to allow the robber jumping is that
this fits our purpose to solve parity games with bounded imperfect
information in polynomial time. Finally, our graph searching game may
be used to model parallel processes that must be served in some
way. Some processes may terminate or may be ultimately served and thus
finished, some can produce new processes if the total number does not
exceed some bound.  Every process corresponds to a robber and
resources used to serve them correspond to the cops. Captured robbers
describe terminated processes and creating new processes is modeled by
multiple robbers running from the vertex of one robber in different
directions. The cop number describes the minimal amount of resources
needed to serve all processes. In our case, processes are possible
plays of a parity game.

The rest of the paper is structured as follows. In
Section~\ref{sec_main_tech} we prove Theorem~\ref{theorem_main_tech},
which states that if~$k$ cops capture one robber on a graph,
then~$k\cdot r$ cops capture~$r$ robbers on that graph. In particular
the number of new cops needed to capture~$r$ robbers is bounded only
in~$r$ and~$k$ by a linear function.  
We show in Theorem~\ref{thm_impinf_strict_hierarchy},
Section~\ref{subsec_hierarchy}, that a linear bound is unavoidable in
our setting. As our example graphs are undirected, this is not due to
directed edges in the graphs, but is caused by the ability of the
robbers to jump.

Before we turn to the analysis of games with 
multiple robbers, let us show how we can use Theorem~\ref{theorem_main_tech} to obtain
Corollary~\ref{cor_main}. Given a parity game with imperfect information bounded by~$r$ 
on a graph~$G$ of \dagw~$k$, we find a winning strategy for~$k\cdot r$ cops against~$r$ 
robbers on~$G$. This strategy can be used to construct a winning strategy for $k\cdot r\cdot 2^{r-1}$
cops against one robber on the powerset graph~$\olG$ (so $\dw(\olG) \le k\cdot r\cdot 2^{r-1}$). 
We show how to do this in Lemma~\ref{lemma_rG_to_G}. As the size of~$\olG$ 
is polynomially bounded in the size of~$G$, we can solve the parity game with perfect information
on~$\olG$ in polynomial time in the size of~$G$.

\subsection{Boundedness of \dagw and parity games}

Going to the powerset graph, we associate every play we 
consider to be possible on the original graph (there are at most~$r$ such 
plays) with one robber. Tracking at most~$r$ plays corresponds to playing
against at most~$r$ robbers simultaneously.

\begin{lemma}\label{lemma_rG_to_G}
If~$\dw_r(\graphG) \le k$, then~$\dw{\olgraphG} \le k \cdot 2^{r-1}$. 
\end{lemma}
\begin{proof}
Let~$f$ be a winning strategy for the cops in the game against~$r$ robbers
on~$\graphG$. We follow a play~$\pi$ consistent with~$f$ and a play~$\olpi$
of the game 
against one robber on~$\olgraphG$ simultaneously. Cop moves are translated
from~$\pi$ to~$\olpi$
and robber moves are translated in the opposite direction.
We maintain two invariants. 
\begin{itemize}
\item[] \textbf{(Robbers)} If the robber occupies a vertex~$\olv =
  \{v_1\cdots,v_s\}\in \olV$ in a position of~$\olpi$, then in the
  corresponding position of~$\pi$ (after the same number of moves),
  the robbers occupy the set~$\olv\subseteq V$.

 \item[] \textbf{(Cops)} If the cops occupy a set~$U$ in a postion of~$\pi$,
then, for every~$u\in U$, the cops occupy every~$\olu$ in the
corresponding position of~$\olpi$.
\end{itemize}

Consider any strategy~$\olg$ for the robber player for
the game with one robber on~$\olgraphG$.  We construct a play~$\olpisr$ 
of this game that is consistent with~$\olg$ (and depends on~$f$),
but is winning for the cops. As~$\olg$ is arbitrary, it follows that the cops
have a winning strategy.

We construct~$\olpisr$ by induction in the length of its finite prefixes.
For every finite prefix~$\olpi_i$ of~$\olpisr$ of length~$i$ we define a 
history~$\pi_i$ of a play on~$\graphG$ that is
consistent with~$f$ and has length~$i$. Hereby, for all even~$i\ge 2$,
if~$(\olU_j,\olU_{j+1},\olv_j)$ is the~$j$th position of~$\pi_i$,
then~$(U_j,U_{j+1},\olv_j)$
is the~$j$th position of~$\pi_i$ such that 
$\olU_j = \{\olu\in\olV \mid \olu \cap U_j \neq \0\}$ and
$\olU_{j+1} = \{\olu\in\olV \mid \olu \cap U_{j+1} \neq \0\}$, for all~$j\le i$.

For~$i=0$, let~$\olpi_i = \pi_i = \perp$.
For the translation of a robber move, let~$\olpi_i$ and~$\pi_i$ be constructed
and let the robber move from~$(\olU_i,\olU_{i+1},\olv_i)$
to~$(\olU_{i+1},\olv_{i+1})$
in the game on~$\olgraphG$.
We define~$\olpi_{i+1} = \olpi_i \cdot (\olU_{i+1}, \olv_{i+1})$ and
$\pi_{i+1} = \pi_i \cdot (U_{i+1}, \olv_{i+1})$ and show that going 
from~$\olv_i$ to~$\olv_{i+1}$ is a legal robber move in the game on~$\graphG$.

As the move from~$\olv_i$ to~$\olv_{i+1}$ is legal on~$\olgraphG$, we have
$\olv_{i+1} \notin \olU_{i+1}$ and~$\olv_{i+1}\in \Reach_{\olgraphG - (\olU_i \cap
\olU_{i+1})}(\olv_i)$.
Let~$P = \olv_i, \olv^1, \ldots, \olv^t, \olv_{i+1}$ be a path from~$\olv_i$
to~$\olv_{i+1}$ in 
$\olgraphG - (\olU_i \cap \olU_{i+1})$.
Let~$v \in \olv_{i+1}$. Then by Lemma~\ref{lemma_powerset}, there is some~$u \in
\olv_i$ and a path 
$u = u^0, u^1, \ldots, u^t, v$ in~$\graphG$ with~$u^l \in \olv^l$, for~$l = 0, 
\ldots, t$. 
We have to show that~$v \notin U_{i+1}$ and that~$v$ is reachable from~$u$ 
in~$\graphG - (U_i \cap U_{i+1})$.

First,~$\olv_{i+1} \notin \olU_{i+1}$ and therefore
$\olv_{i+1} \cap U_{i+1} = \0$, which implies~$v \notin U_{i+1}$. Now assume
towards a contradiction that~$v$ is not reachable from~$u$ in~$\graphG - (U_i 
\cap U_{i+1})$. Then there is some
$l \in \{1, \ldots, t\}$ such that~$u^l \in U_i \cap U_{i+1}$. However, 
since~$u^l \in \olv^l$, 
by the induction hypothesis, we have~$\olv^l \in \olU_i \cap \olU_{i+1}$, but
$\olv^1, \ldots, \olv^t$ is a path in~$\olG - (\olU_i \cap \olU_{i+1})$.

To translate the answer of the cops, consider the set~$U = f(U_i,\olv_i)$, which~$f$
prescribes to occupy in the next move, so~$\pi_{i+1} = \olpi_i \cdot (U_i, 
U_{i+1}, \olv_i)$.
Let the next move in~$\olpi$ be defined by~$\olU_{i+1} = \{\olu \in \olV\ \mid 
\olu \cap C \neq \0\}$, and
hence,~$\olpi_{i+1} = \olpi_i \cdot (\olU_i, \olU_{i+1}, \olv_i)$.

Finally, play~$\olpisr$ is the limit of all~$\olpi_i$, \ie the~$0$th position 
of~$\olpisr$ is~$\perp$,
and the~$i$th position is~$(\olU_i,\olU_{i+1},\olv)$, if~$i$ is a positive even 
number,
and~$(\olU_i,\olv_i)$ if~$i$ is odd.

We have to show that~$\olpisr$ is won by the cops, \ie that it is monotone and 
the robber is captured. 
To prove the monotonicity, assume towards a contradiction that the 
play~$\olpisr$ is not monotone, \ie
there is some position~$(\olU_i,\olU_{i+1},\olv_i)$ of~$\olpisr$ such that 
there is some~$\olu \in \olU_i \setminus \olU_{i+1}$ reachable from~$\olv_i$ 
in~$\olgraphG - (\olU_i \cap \olU_{i+1})$. 
Let~$\olv_i, \olv^1, \ldots, \olv^t, \olu$ be a path from~$\olv_i$ to~$\olu$ 
in~$\olgraphG$ with~$\olv^l \notin \olU_i$, for all~$l \in\{1, \ldots, t\}$.
Since~$\olu \in \olU_i\setminus \olU_{i+1}$, by the construction of~$\pi$, 
there is some~$u \in \olu$ with~$u \in U_i\setminus U_{i+1}$.  
Moreover, by Lemma \ref{lemma_powerset}, there is some 
$v_i \in \olv_i$ and a path ~$v_i, v^1, \ldots, v^t, u$ in~$\graphG$ 
with~$v^l \in \olv^l$, for all~$l \in\{1, \ldots, t\}$. By the construction
of~$\pi_i$ 
all~$v^l \notin U_i$, thus~$u$ is reachable from~$v_i$ in~$\graphG - U_i$, 
which 
contradicts the monotonicity of~$f$. Hence,~$\olpisr$ is monotone.

Consider the play~$\pisr$ obtained as a limit of all~$\pi_i$. If~$\olpisr$ is
infinite, then~$\pisr$ is infinite as well, which is impossible, as~$\pisr$
is consistent with~$f$.

Finally, we count the number of cops used by the cop player in~$\olpisr$.
Consider any position~$(\olU_i, \olU_{i+1}, \olv_i)$ occurring in~$\olpisr$. 
Since~$\pisr$ is consistent with~$f$, for the corresponding position~$(U_i, 
U_{i+1}, \olv_i)$ in
$\pisr$, we have~$|U_{i+1}| \le k$. From the construction of~$\olpisr$, it 
follows that
$|\olU_{i+1}| \le k \cdot 2^{r-1}$. Therefore, the robber does not have a 
winning strategy against 
$k \cdot 2^{r-1}$ cops in the game on~$\olgraphG$. By determinacy,~$k \cdot
2^{r-1}$ cops
have a winning strategy.
\end{proof}

\subsection{The multiple robbers game}
Let $G = (V, E)$ be a graph and $k,r\in
\omega$. The~\emph{$k$ cops and~$r$ robbers game} $\game_k^r(G)$ is defined as
follows. A position has the form $(U, R)$ or $(U, U', R)$ where $U, U', R
\subseteq V$ with  $|U|, |U'| \leq k$ and $|R| \leq r$.  Hereby~$U$
and~$U'$ are as in the game with one robber and~$R$ are the vertices
occupied by the robbers. From a cop position $(U, R)$, the cops can move to any
position $(U, U', R)$ as in the game with one robber. From a robber position $(U, U', R)$, the
robbers can move to any position  $(U', R')$ such that $R' \cap U' = \0$ and
each $r' \in R'$ is reachable from some $r \in R$ in $G - (U \cap U')$. In the
first move, the robbers can go from the initial position $\perp$ to any position $(\0, R)$ with $|R|
\leq r$.

Notice that this definition blurs the role of single robbers: first, a robber
can leave the graph and, second, one robber can induce multiple robbers in the
next position. Indeed, there may be distinct $v_1,v_2\in R'$ reachable (in $G-(U\cap U')$) only from
one vertex $v\in R$. In this case, we say informally that
robber $v_1$ \emph{runs} and robber $v_2$ \emph{jumps} if we assume that the
robber on~$v_1$ was on~$v$ before the move and the robber on~$v_2$ was on a
vertex~$w$ with $v_2\notin\Reach_{G-(U\cap U')}(w)$. However, this distinction
is not formalized (we could also swap the roles of~$v_1$ and~$v_2$) and used
only to develop better intuition. 

A play of a cops and multiple robbers game is \emph{(robber-)monotone} if
the play contains no position $(U, U', R)$ such that some
$u \in U \setminus U'$ is reachable from some $r \in R$ in $G - (U\cap U')$. 
Monotone finite plays are won by the cops,  non{}-monotone plays and infinite plays are won by
the robbers.

A \emph{memory strategy} for the cop player in a cops and multiple robbers game
is a memory structure $\calM=(M,\init,\upd)$ together with a strategy function
$f:M\times 2^V\times V \to 2^V$ (for the cop strategy), respectively $f:M\times 2^V\times 2^V \to 2^V$
(for the robber strategy).
Hereby~$M$ is a set of memory states, $\init:V\to M$, respectively $\init:2^V\to M$ 
is the memory initialization function mapping the robbers placement after the first 
move of the robbers to a memory state, and 
$\upd:M\times 2^V\times 2^V\times V\to M$, respectively $\upd:M\times 2^V\times 2^V\times 2^V\to M$
is the memory update function, which maps a memory state and a cop respectively a robber position to a new state.
A memory strategy is positional if $|M|=1$, in which case~$M$ can be omitted.
Winning strategies, plays, histories and consistency are defined in the usual
way, analogously to the case of a single robber.
A cop strategy is monotone, if every play consistent with it is monotone. As the cops have a reachability winning
condition, the cops and multiple robbers games are positionally determined. We will use 
memory strategies because they allow us more intuitive descriptions.

The least~$k$ such that the cops have a winning strategy for the~$k$ cops
and~$r$ robbers game on~$G$ is denoted by $\dw_r(G)$. Note that the \DAGw of a graph $G$ is
$\dw_1(G)$. We define $\tw_r(G)$ analogously to the case of one
robber, \ie $\tw_r(G) = \dw_r(G^\leftrightarrow)-1$ where $\graphG^\leftrightarrow$ is as~$\graphG$, but with
the edge relation replaced by its symmetrical closure. Recall from
Section~\ref{subsec_powerset} that~$\olG$ is the graph obtained from a
graph~$\graphG$ by applying the powerset construction.

As discussed above, to complete our second proof of Corollary~\ref{cor_main}  it
remains to show that Theorem~\ref{theorem_main_tech} holds. We do this in
the next section.

\section{From one robber to \texorpdfstring{$\mathbi{r}$}{\textbf{\textit{r}}} robbers}\label{sec_main_tech}

As the first step we show that we can assume without loss of generality
two restrictions on robbers strategies.  A robber strategy~$f$ is
\emph{isolating} if no two
robbers can reach one another, \ie if in any cop position~$(U,R)$ of any play
consistent with~$f$, for all~$v,w\in R$, we have~$v\notin
\Reach_{\graphG-U}(w)$.
An important special case of this rule is that there can never be two robbers in
the same component. Intuitively, if~$v\in \Reach_{\graphG-U}(w)$, then the robber
on~$v$ 
is redundant: the robbers can place him on~$v$ also in the next move. He can go 
to~$v$ by first jumping to the robber on~$w$ and then running from~$w$ to~$v$.

The second restriction on the robber moves is that each of them leaves his
vertex either if he jumps to another robber (a reason for a jump can be that he
is needed somewhere else) or if otherwise (if he does not jump, but runs) the
destination
of his run would become unreachable for him in the next move. Formally, we
say that a robber strategy~$f$ is \emph{prudent}\label{def_prudent} if, for
each robber move
$(U, U', R)\to (U', R')$ consistent with~$f$, we have~$r'\notin
\Reach_{G - U'}(R)$, for any~$r' \in R' \setminus R$. This is not a proper
restriction to the robber moves either. Indeed, running within the same
component makes no sense, as the set of vertices reachable for the robber does
not change. Running outside of the current component makes even less sense, as
that set becomes smaller.

\begin{lemma}\label{lemma_isolating_prudent}
    If~$r$ robbers have a winning strategy against~$k$ cops, then~$r$
    robbers have an isolating prudent  winning strategy 
    against~$k$ cops.
\end{lemma}
\begin{proof}
  Given a set of vertices~$U$, we say that~$R$ and~$\hat R$ are
  equivalent, $R\equiv_U \hat R$, if for all $r\in R$ there is some
  $\hat r \in\hat R$ and vice versa, for all $\hat r\in\hat R$ there
  is some $r\in R$ such that $r$ and~$\hat r$ are in the same
  component of $G-U$.

  Let~$f$ be a positional winning strategy for~$r$ robbers in the
  monotone multiple robbers game on~$G$ against~$k$ cops. We construct
  an isolating prudent strategy~$\hat f$ for~$r$ robbers against~$k$
  cops by induction on the play length playing simultaneously a
  play~$\pi$ consistent with~$f$ and a play~$\pi'$ consistent
  with~$\hat f$. We translate a cop move from~$\hat \pi$ to~$\pi$ and
  a robber move from~$\pi$ to~$\hat \pi$ and show the
  following. If~$(U,U',R)\to(U',R')$ is the~$i$th robber move in~$\pi$
  and~$(\hat U,\hat U',\hat R)\to(\hat U',\hat R')$ is the~$i$th
  robber move in~$\hat\pi$, then
\begin{itemize}
\item~$U= \hat U$ and~$U' = \hat U'$, and
\item~$\Reach_{G-U'}(R') \subseteq \Reach_{G-U'}(\hat R')$.
\end{itemize}

Clearly, this implies that~$\hat f$ is a winning strategy.
In the beginning of a play, the first robber move
is translated as in the general case. Cop moves are translated without any
change, so the invariant is not broken.

For the translation of a robber move, consider the topological
order~${}\topordereq{}$ on vertices of~$G-U'$ where $v \topordereq w$
if $w \in \Reach_{\graphG - U'}(v)$. Let~$\alpha\colon 
G\bigl/_{\equiv_{U'}} \to G$ be a choice function
on~$G\bigl/_{\equiv_{U'}}$. For every set of vertices~$R$, let~$\topomin\colon 2^G
\to 2^G$ be defined by $\topomin(R) = \{\alpha(r) \mid r \text{ is }\topordereq
\text{-minimal in }R\}$.

Let~$A = \{v'\in R'\setminus R \mid v' \in \Reach_{G-U'}(R)\}$ and let
$\beta\colon A\to R$ be some function with $v' \in
\Reach_{G-U'}(\beta(v'))$. It is exists by the definition of
robber moves. Let $\gamma\colon R' \to G$ with $\gamma(v') =
\beta(v')$ if $v'\in A$ and $\gamma(v') = v'$ otherwise.
If~$f$ prescribes to move from~$(U,U',R)$ to~$(U',R')$,
then~$\hat f$ prescribes to move from~$(U,U',\hat R)$ to~$(U',\hat
R')$ where $\hat R' = \topomin(\{\gamma(v') \mid v'\in R'\})$.
Then~$\hat f$ is isolating and prudent.
Note that, by the first part of the invariant, strategy~$\hat f$ is well
defined. The invariant follows directly from the construction.

\end{proof}

\subsection{Tree-width and componentwise hunting}\label{sec_tree-width}

Before we prove our main result of this section, 
let us first consider the same problem for the game characterizing \treew. 

\begin{lemma}\label{lemma_Rtreewidth} For all $G$ and $k,r>0$,
if $\tw(G)+1 \leq k$, then $\tw_r(G) +1 \leq r \cdot (k+1)$.
\end{lemma}

It follows from this lemma that if \treew is fixed, then parity games
with bounded imperfect information are solvable in polynomial time
because $\tw_r(\graphG) \le r\cdot k$ implies $\dw_r(\graphG) \le
r\cdot k$ . As the \pathw of a graph is always at least its \treew, we obtain the
same result for (undirected) \pathw.

\begin{proof}
Without loss of generality let~$G$ be undirected.
Let~$f$ be a monotone winning strategy for~$k$ cops
in the game on~$G$ against one robber. 
As~$f$ is monotone, we can assume that cops are not placed on vertices that are 
already unavailable for the robber, \ie for a move $(U,v)\to(U,U',v)$ we always have 
$U'\setminus U \subseteq\Reach_{G-(U\cap U')}(v)$. (Otherwise, instead of~$f$, 
consider a strategy that is as~$f$, but never places cops on such vertices. 
This strategy will be still monotone and winning and will use at most~$k$ 
cops.)
We construct a monotone strategy $\otimes_r f$ for $k\cdot r$ cops
in the game on~$G$ with $r$ robbers that is winning 
against each isolating robber strategy. 

Intuitively, the cop player uses~$r$ teams of cops 
with~$k$ cops in each team.
Every team plays independently of each other chasing its own robber
according to~$f$. We maintain the invariant that in each cop position $(U,R)$
that is consistent with $\otimes_r f$, there is a partition $(U_1,\cdots , U_r)$ of
$U$ and an enumeration of $v_1,\cdots,v_r$ of $R$ such that for each $v_i$, 
$(U\setminus U_i) \cap \Reach_{G-U_i}(v_i)=\emptyset$, \ie cops on $U_i$
block $v_i$ from other cops, and that $(U_i,v_i)$ is consistent
with $f$ in the game with one robber. The next move of the cops is
$\otimes_r f(U,R) = \bigcup_{i=1}^r f(U_i,v_i)$. By a simple induction on the length
of a play it is easy to see that the invariant holds, which implies that the cops
monotonically catch all~$r$ robbers.
\end{proof}

The reason why the proof is so simple is that in an \emph{undirected}
graph the set of vertices reachable from a given position is
precisely the connected component which contains these positions. Thus
the strategy~$f$ does not need to place cops on vertices outside the
robber component. For directed graphs, this is not true and the simple
translation of strategies is not possible without certain refinement
any more. Consider the following possible situation. The cops play
simultaneously against all robbers according to a winning strategy~$f$
in the game against one robber as before. A slightly different
variant of this approach (that will be used in the proof of
Theorem~\ref{theorem_main_tech}) is that they choose one of them
(say, occupying some vertex~$v_1$) to play against him further while the cops
of other teams wait for this robber to be caught. The robbers stay in two distinct components 
on~$v_1$ and~$v_2$. The problem is that~$v_2$, may prevent playing against~$v_1$. If~$f$ says
to place a cop on a vertex~$v$ that is reachable from~$v_2$, it may become
impossible to reuse the cop from~$v$ later playing against~$v_1$, although~$f$ 
prescribes to do so: $v_2$ would induce non{}-monotonicity on~$v$. 

One approach to solve this problem is to change~$f$ such that it does
not prescribe to place cops outside of the robber component. It would
suffice to prove that there is a function $F : \omega \to \omega$ such
that every cop winning strategy~$f$ for~$k$ cops against one robber
can be transformed into a winning strategy~$f'$ for $F(k)$ cops
against one robber that never prescribes to place cops outside of the
robber component.  In other words, strategy~$f'$ should fulfill the
following property: in a position $(U,v)$, if~$C$ is the component of
$G-U$ with $v\in C$, then $f'(U,v) \subseteq C$. However, such a
function~$F$ does not exist, as we will show in
Theorem~\ref{thm_OffhandedCops}. For this proof we need a statement
about cop strategies. In the next lemma we show that any cop
positional winning strategy for the game with one robber can be
modified without using additional cops to obtain a new positional
strategy that obeys the following rules. It does not place a cop on a
vertex that is already unavailable for the robber and always
prescribes to place new cops.  In a graph $G$, for a set $A$ and a
vertex $v\notin A$, let $\front_G(v,A)$ be the inclusion minimal
subset~$B$ of~$A$ such that $\Reach_{G-A}(v) = \Reach_{G-B}(v)$. It is
easy to see that~$B$ is unique and thus well-{}defined.

\begin{lemma}\label{lemma_useless_moves}
  On a graph~$G$, if~$f$ is a positional monotone winning strategy
  for~$k$ cops against one robber, then there is a positional monotone
  winning strategy~$f^*$ for~$k$ cops against one robber, such
  that, for every position $(U,v)$ that appears in a play consistent
  with~$f^*$, we have that $f^*(U,v) \setminus U \neq \0$
  and that any $u \in f^*(U,v) \setminus U$ is reachable
  from~$v$ in $G - U$.
\end{lemma}
\begin{proof}
  We construct~$f^*$ by induction on the length of
  the finite prefixes~$\pi$ of plays consistent with~$f$ together
  with finite prefixes~$\pi^*$ of plays consistent with~$f^*$
  such that the following invariant holds:
\begin{itemize}
\item $|\pi^*|\le |\pi|$;

\item if $\last(\pi) = (U,v)$ and $\last(\pi^*)=(U^*,v^*)$, then 
    \begin{itemize}
    \item $v = v^*$,
    \item $U^*\subseteq U$ and
    \item $\Reach_{G-U}(v) = \Reach_{G-U^*}(v)$;
    \end{itemize}

\item if $\last(\pi) = (U,U',v)$ and $\last(\pi^*)=(U^*,{U^*}',v^*)$, then 
    \begin{itemize}
    \item $v = v^*$,
    \item $U^*\subseteq U$, ${U^*}'\subseteq U'$ and
    \item $\Reach_{G-(U\cap U')}(v) = \Reach_{G-(U^* \cap {U^*}')}(v)$. 
  \end{itemize}
\end{itemize}
Notice that the invariant immediately implies that~$f^*$ is winning
for the cops.

For a cop position~$(U,v)$, let~$(U=U_0,v),(U_1,v),\ldots,(U_m,v)$ be defined by
the following rule: 
\begin{itemize}
\item $U_0 = U$, $U_1 = f(U_0,v)$ (so $m\ge 1$), and
\item if $\front_G(v,U_{i-1}) \neq \front_G(v,U_i)$, then $U_{i+1} =
f(U_i,v)$, otherwise,~$m = i$ and~$U_{i+1}$ does not exist.
\end{itemize}
Then~$f^*(U,v) = U_m$. Intuitively, we skip all cop moves according
to~$f$ in which new cops are only placed on or removed from vertices
behind the front, \ie on vertices that are not reachable from the
robber vertex because of other cops. The next cop move
according to~$f^*$ is the first move according to~$f$ where the cops
are placed between the robber and the front (when the front changes)
under the assumption that the robber does not move while the cops move
behind the front.

At the beginning, we have~$\pi = \pi^* = \perp$ and the invariant trivially
holds. In general, let~$f^*$ be defined for all positions in plays up to a
certain length. Consider finite histories~$\pi$ and~$\pi^*$ as above. 

Let~$\last(\pi) = (U,U',v)$ and $\last(\pi^*)=(U^*,{U^*}',v)$ and let the robber
move from~$\last(\pi^*)$ to a position~$({U^*}',v')$. Then we extend~$\pi$ by
position~$(U',v')$ and the invariant holds again.

Let~$\last(\pi) = (U,v)$ and $\last(\pi^*)=(U^*,v^*)$. Then the next move of the
cops is to~$f^*(U^*,v)$ and the invariant still holds. As~$U^*\subseteq U$
and~${U^*}'\subseteq U'$, $f^*$ uses at most~$k$ cops. Note that $f^*$
is positional.
\end{proof}

Now we prove that the cops have to place themselves outside of the
robber component.

\begin{theorem}\label{thm_OffhandedCops}
For $n>0$ there are graphs $G_n$ such that $\dw(G_n) \leq 3$ for all~$n$,
but any winning cop strategy which is restricted to place cops only inside the robber
component, uses at least $n+1$ cops.
\end{theorem}
\begin{proof}
  The graph~$G_n$ is the disjoint union of an undirected and a
  directed tree, both of the same shape: full trees of branching
  degree and depth~$n+1$, with some additional edges connecting
  the trees, see Figure~\ref{pic_offhanded}.

  Let, for~$i\in\{0,1\}$ and~$m,n>0$, $ A(i,m,n) = (\{1, \dots,
  n\}\times \{i\})^{\le m}$ be the set of all sequences of length at
  most~$m$ over the alphabet~$\{1, \dots, n\}$ labeled with~$i$ (the
  labeling is used to distinguish the trees).  For~$v =
  (v_0,i),\dots,(v_l,i)\in A(i,m,n)$, let~$v'$ be the word
  $(v_0,1-i),\dots,(v_l,1-i) \in A(1-i,m,n)$.

The vertex set of $G_n$ is defined by~$V_n = V^0_n\cup V^1_n$ where 
$V^0_n = A(0,n,n+1)$ and~$V^1_n = A(1,n,n+1)$.

The edges are defined by~$E_n = E^0_n \cup E^1_n\cup E'_n$. Hereby 
\[E^0_n = \{(v, vj),  (vj,v) \mid v \in A(0,n-1,n+1), j \in A(0,1,n+1)\}\,,\] 
\[E^1_n = \{(vj, v) \mid v \in A(1,n-1,n+1), j \in A(1,1,n) \}\,,\text{ and}\]
\[E'_n = \{(v,v') \mid v\in A(0,n,n+1)\} \cup 
  \{(vj, v') \mid v \in A(1,n-1,n+1), j \in A(1,1,n)\}\,.\]

It is easy to see that cops three capture the robber. They occupy both 
roots~$(0,1)$ and~$(1,1)$ in the first move. By symmetry we can assume that 
the robber goes to the left-most subtree. Then the third cop is placed 
on the successor~$(1,2)$ of~$(1,1)$ and then the cop from~$(1,1)$ is moved
to~$(0,2)$. In this manner, the cops work through both trees top-down and the 
robber is captured in some leaf.

For the second statement, define $T^0_n = (V^i_n,E^i_n)$ for
$i=\{0,1\}$ and note that it makes no sense for the cops to leave out
holes, \ie to place cops on subtrees of~$T^0_n$ or~$T^1_n$ rooted at a
vertex~$v\in V^0_n$, respectively, at~$v'$, if~$v$ is reachable from
the robber vertex via some cop free path.  Indeed, due to the high
branching degree, the robber can switch between subtrees of~$v$ going
into those having no cop in them until~$v$ is occupied by a cop. In
that position the cops from other components than that of the robber can be
removed by Lemma~\ref{lemma_useless_moves}. So we can assume that the
cops play top-down, \ie they never leave out holes. Then the robber
strategy is just to stay in the left-most branch. Note that after a
vertex~$v\in V^0_n$ is occupied by a cop, vertex~$v'\in V^1_n$ is not
in the robber component any more. Thus the cops occupy
successively~$(\epsilon,0)$,~$(1,0)$, $(2,0)$, and so on.  In that
way, more and more cops become tied, \ie for every cop on a
vertex~$(j,0)$, there is a cop-{}free path from the robber vertex
to~$(i,0)$.
\end{proof}

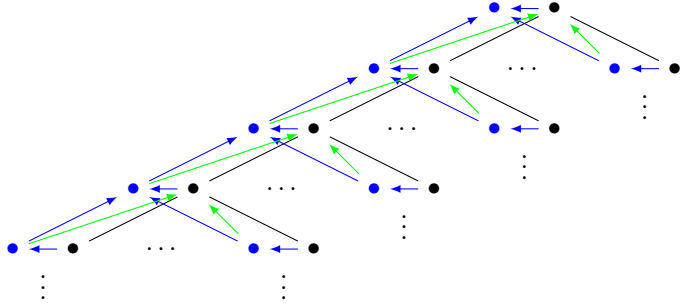
\begin{figure}
\begin{center}
\begin{tikzpicture}[scale=0.8]

\node 		(epsilon)	 at (0,0) {$\bullet$};
\node[blue] 	(epsilon') at (-1,0) {$\bullet$};
\node 		(1) at (-2,-1) {$\bullet$};
\node[blue] 	(1') at (-3,-1) {$\bullet$};
\node 		(11) at (-4,-2) {$\bullet$};
\node[blue] 	(11') at (-5,-2) {$\bullet$};
\node 		(111) at (-6,-3) {$\bullet$};
\node[blue]	(111') at (-7,-3) {$\bullet$};
\node 		(1111) at (-8,-4) {$\bullet$};
\node[blue] 	(1111') at (-9,-4) {$\bullet$};
\node 		(n) at (2,-1) {$\bullet$};
\node[blue] 	(n') at (1,-1) {$\bullet$};
\node 		(1n) at (0,-2) {$\bullet$};
\node[blue] 	(1n') at (-1,-2) {$\bullet$};
\node 		(11n) at (-2,-3) {$\bullet$};
\node[blue] 	(11n') at (-3,-3) {$\bullet$};
\node 		(111n) at (-4,-4) {$\bullet$};
\node[blue] 	(111n') at (-5,-4) {$\bullet$};

\draw (epsilon) to (1);
\draw (epsilon) to (n);
\draw (1) to (11);
\draw (1) to (1n);
\draw (11) to (111);
\draw (11) to (11n);
\draw (111) to (1111);
\draw (111) to (111n);

\draw[<-,blue] (epsilon') to (1');
\draw[<-,blue] (epsilon') to (n');
\draw[<-,blue] (1') to (11');
\draw[<-,blue] (1') to (1n');
\draw[<-,blue] (11') to (111');
\draw[<-,blue] (11') to (11n');
\draw[<-,blue] (111') to (1111');
\draw[<-,blue] (111') to (111n');

\draw[<-,blue] (epsilon') to (epsilon);
\draw[<-,blue] (1') to (1);
\draw[<-,blue] (11') to (11);
\draw[<-,blue] (111') to (111);
\draw[<-,blue] (1111') to (1111);

\draw[<-,blue] (n') to (n);
\draw[<-,blue] (1n') to (1n);
\draw[<-,blue] (11n') to (11n);
\draw[<-,blue] (111n') to (111n);

\draw[->,green] (1') to (epsilon);
\draw[->,green] (n') to (epsilon);
\draw[->,green] (11') to (1);
\draw[->,green] (1n') to (1);
\draw[->,green] (111') to (11);
\draw[->,green] (11n') to (11);
\draw[->,green] (1111') to (111);
\draw[->,green] (111n') to (111);

\node (dots) at (-0.5,-1){\dots};
\node (dots) at (-2.5,-2){\dots};
\node (dots) at (-4.5,-3){\dots};
\node (dots) at (-6.5,-4){\dots};

\node (dots) at (1.5,-1.5){\vdots};
\node (dots) at (-0.5,-2.5){\vdots};
\node (dots) at (-2.5,-3.5){\vdots};
\node (dots) at (-4.5,-4.5){\vdots};
\node (dots) at (-8.5,-4.5){\vdots};

\end{tikzpicture}
\end{center}
\caption{$\dw(G_n)=4$, but 
the robber wins against $n$ cops if they move only into his component.}
\label{pic_offhanded}
\end{figure}

Before we proceed with the case of directed graphs, let us mention
that bounded \treew of the \emph{Gaifman graphs} of given games
already implies that parity games with imperfect information are
solvable in \ptime. The Gaifman graph of a relational structure $S =
(A,R_1,\dots,R_m)$ is the undirected graph with vertices~$A$ and an
edge between vertices~$v$ and~$w$ if $v$ and $w$ appear in the same
tuple in some relation $R_i$ for $i\in\{1,\dots,m\}$. Thus a
$\sim$-equivalence class in a game with imperfect information
induces a clique in the Gaifman graph consisting of all equivalent
vertices. Thus if the \treew of the Gaifman graphs of some games is
bounded, then so is the imperfect information. This implies the
following corollary.
\begin{corollary}\label{cor_bounded_Gaifman}
  Parity games whose Gaifman graphs have bounded \treew can be solved
  in deterministic polynomial time.
\end{corollary}

\subsection{Generalization to the directed case}\label{sec_DAG-width}

We are ready to prove our main result of Section~\ref{sec-parity}.

\begin{theorem}\label{theorem_main_tech}
For $k,r>0$, if $\dw(G) \leq k$, then   $\dw_r(G) \leq k \cdot r$. 
\end{theorem}
The rest of the section is devoted to the proof of this theorem.
Let~$f$ be a positional monotone winning strategy for~$k$ cops against
one robber on a directed graph~$G$.  According to
Lemma~\ref{lemma_useless_moves} we can assume \wLOG that for any
history~$\pi'$ consistent with $f$ such that $\last(\pi') = (U, v)$ we
have $f(\pi') \setminus U \neq \0$ and any $u \in f(\pi') \setminus U$
is reachable from~$v$ in $G - U$.  Moreover, due to
Lemma~\ref{lemma_isolating_prudent} it suffices to construct a
strategy $\otimes_r f$ for $r \cdot k$ cops against~$r$ robbers which
is winning against all isolating prudent robber strategies.  First, we
sketch a description of a memory strategy $\otimes_r f :
\calM\times (2^V \times 2^V) \rightarrow 2^V$ and the corresponding
memory structure.

\WLOG we can assume that~$G$ is strongly connected. Indeed, given a winning 
strategy for~$k\cdot r$ cops against~$r$ robbers on every strongly connected
component
of~$\graphG$, we can traverse the graph by applying the strategy to the
topologically
minimal components, then eliminate them and continue in that way until all 
robbers are captured.

\subsection*{Informal description, some invariants and some elements of
  the memory structure}
The cops play in~$r$ teams of~$k$ cops. Consider a position~$(U,R)$ in
a play with~$r$ robbers. With every vertex~$v\in R$ occupied by a
robber, we associate a team of cops~$U_i\subseteq V$ with~$|U_i|\le
k$. With each~$U_i$ we associate a history~$\rho_i$ of the game
against one robber that is consistent with~$f$ such that~$(U_i,v)$ is
the last position of~$\rho_i$. We formulate this as an invariant in
the game with~$r$ robbers:

\medskip
\textbf{(Cons)} Any history~$\rho_i$ is consistent with~$f$.
\medskip

For any position that appears in a play against~$r$ robbers, we keep~$s\le r$
histories~$\rho_i$ in memory and write ~$\rho = \rho_1\cdot \ldots \cdot
\rho_s$.
This sequence of histories is the main part of the memory. 
The following invariant says that, up to the last robber moves, all~$\rho_i$
are linearly ordered~by~$\propprefix$.

\medskip
\textbf{(Lin)}~$\rho_1 \propprefix \rho_2 \propprefix \ldots \propprefix
\rho_s$.
\medskip

Sequence~$\rho$ is constructed and maintained in the memory in the following
way. At the beginning
of a play, we set~$\rho = \rho_1 = \perp$, \ie~$\rho$ consists of one play
prefix
containing only the initial position. When the play with~$r$ robbers goes on,
but
only
one robber is in the graph,~$\rho_1$ grows together with the play with~$r$
robbers
and the latter gets the form 
$\perpcdot(U^1,R^1)\cdots(U^m,R^m)(U^m,U^{m+1},R^m)$ where all~$R_i$ 
are singletons. While playing this part of the play, 
all teams make the same moves according to~$f$. We store the sequence in
the memory
as~$\rho = \rho_1 = \perpcdot (U^1,v^1)\cdots(U^m,v^m)(U^m,U^{m+1},v^m)$ 
where~$\{v^i\}=R^i$ (see Figure~\ref{fig_impinf_one_to_r}).
When more robbers come into the graph, they go into different components
(because
they play according to an isolating strategy)
and the cops choose one of them, say on a vertex~$b_1$.
We associate~$\rho_2 = \rho_1\cdot (U^m,U^{m+1},b_1)$ with that robber 
and store~$\rho = \rho_1\cdot\rho_2$ in the memory.
Note that~$\rho_1$ ends with a robber position.
Assume for a moment that only the robber in~$\flap_U(b_1)$ moves
where~$U$ is the placement of the cops in the position when new
robbers entered the graph.
Then only this robber is pursued by its team of cops according to~$f$, but
cops are not placed on vertices~$v$ if $v\in \Reach_{\graphG-U^m}(R_1)$
where~$R_1$ is the set of robbers distinct from~$b_1$.
The cops belonging to other teams remain idle.
Cop moves are appended to~$\rho_2$, however, without respecting the omitted
placements. To put it differently, let~$W_i$ be the last cop placement
in~$\rho_i$ and 
let~$b_2$ be the last robber vertex in~$\rho_2$.
Then in a position~$(U,R)$ of the game with~$r$ robbers, 
we have 
\[\rf(U,R) = (U\setminus W_2) \cup \big(f(W_2,b_2)\setminus
\Reach_{\graphG-W_1}(b_2)\big)\,.\]
Hereby,~$U\setminus W_2$ are cops from the team associated with~$\rho_1$.
Note that~$\rf$ depends also on the memory state, but we will not write
this
explicitly.
For the memory state update, in~$\rho_2$, not the actual 
move $f(W_2,b_2)\setminus \Reach_{\graphG-W_1}(b_2)$ is stored,
but the intended one, \ie~$f(W_2,b_2)$.
If later new robbers come and occupy different components of~$\flap_U(b_2)$,
we again choose one of them (say, on~$b_3$), create~$\rho_3$ and set~$\rho_3$,
$W_3$ 
and~$b_3$ analogously to~$\rho_2$,~$W_2$ and~$b_2$, and store~$\rho =
\rho_1 \cdot \rho_2,\rho_3$.
Analogously, the cops play according 
to~\[\rf(U,R) = (U\setminus W_2) \cup 
\Big(f(W_3,b_3)\setminus \big(\Reach_{\graphG-W_2}(b_2)\cup
\Reach_{\graphG-W_1}(b_1)\big)\Big)\,.\]

Histories in~$\rho$ are subject to change, so at different points of 
time,~$\rho$ and~$\rho_i$ are different objects, but we will not reflect that
in our notation to avoid unnecessary indexes. It will be always clear from the
context 
what~$\rho$ is.
Note that cops from teams~$\propprefix$-smaller than~$3$ (in general,~$s$)
cannot be removed from their vertices, as, according to~$f$, omitted
placements must be performed first. Hence, taking the cops may infer
non-\mbox{}monotonicity. For example, both cops from~$\rho_1$ in
Figure~\ref{fig_impinf_one_to_r}
cannot be removed before the omitted placement in~$\rho_1$ is performed.
Note also that there may be more than one robber in~$R_i$
associated to a play~$\rho_i$ if~$i<s$ and at most one robber is associated
with~$\rho_s$. 

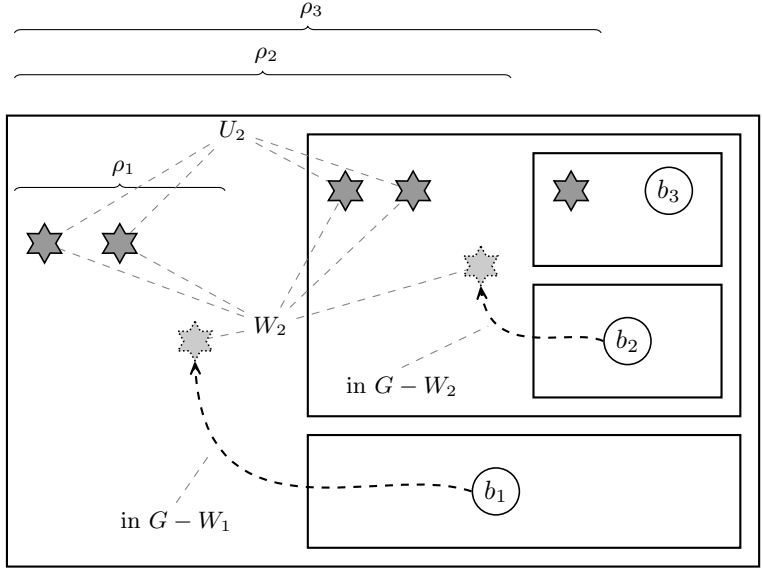
\begin{figure}
\begin{center}
\begin{tikzpicture}[thick]

\draw (-5,-3) rectangle (5,3); 
\draw (-1,-1) rectangle (4.75,2.75); 
\draw (-1,-2.75) rectangle (4.75,-1.25); 
\draw (2,1) rectangle (4.5,2.5); 
\draw (2,-0.75) rectangle (4.5,0.75); 

\node[cop] (cop1) at (-4.5,1.3){};
\node[cop] (cop2) at (-3.5,1.3){};
\node[ocop] (ocop1) at (-2.5,0){};

\node[cop] (cop3) at (-0.5,2){};
\node[cop] (cop4) at (0.4,2){};
\node[ocop] (ocop2) at (1.3,1){};

\node[cop] (cop5) at (2.5,2){};

\node[robber] (robber3) at (3.8,2){$b_3$};
\node[robber] (robber2) at (3.25,0){$b_2$};
\node[robber] (robber1) at (1.5,-2){$b_1$};

\draw[-slim,dashed] (robber1.west) ..controls (0,-1.6) and (-2.4,-2.8) ..
(ocop1.south);
\draw[-slim,dashed] (robber2.west) ..controls (2.5,0.2) and (1.4,-0.3) ..
(ocop2.south);

\small
\node (reach1) at (-2.75,-2.4) {in~$\graphG-W_1$};
\draw[ie] (reach1.north) to (-2.3,-1.5);
\node (reach2) at (0.25,-0.6) {in~$\graphG-W_2$};
\draw[ie] (reach2.north) to (1.4,0.2);

\draw[brace] (-4.9,2) -- (-2.1,2);
\node (rho1) at (-3.45,2.3) {$\rho_1$};
\draw[brace] (-4.9,3.5) -- (1.7,3.5);
\node (rho2) at (-1.55,3.8) {$\rho_2$};
\draw[brace] (-4.9,4.1) -- (2.9,4.1);
\node (rho3) at (-0.95,4.4) {$\rho_3$};

\node (U_2) at (-2,2.8) {$U_2$};
\draw[ie] (U_2) to (cop1);
\draw[ie] (U_2) to (cop2);
\draw[ie] (U_2) to (cop3);
\draw[ie] (U_2) to (cop4);

\node (W_2) at (-1.5,0.2){$W_2$};
\draw[ie] (W_2) to (cop1);
\draw[ie] (W_2) to (cop2);
\draw[ie] (W_2) to (ocop1);
\draw[ie] (W_2) to (cop3);
\draw[ie] (W_2) to (cop4);
\draw[ie] (W_2) to (ocop2);

\end{tikzpicture}
\end{center}
\caption[Memory used by strategy~$\rf$ on the graph~$\graphG$.]{Memory used
by strategy~$\rf$ on the graph~$\graphG$. Squares are 
robber components. Stars denote cop vertices, 
dotted light gray stars denote vertices where cop placements were omitted.}
\label{fig_impinf_one_to_r}
\end{figure}

Now we describe the remaining elements of the memory.
A complete element of the memory structure has the form 
\[\zeta = (\rho_1,R_1,O_1)\cdot \ldots \cdot (\rho_{s-1},R_{s-1},O_{s-1})\cdot
\rho_s\,.\]
Hereby~$\rho_i$ are as before and, for~$i<s$,~$\rho_i$ ends with a robber
position. 
The last robber moves associated with~$\rho_i$ (other robbers may join the
robber from~$\rho_i$) are stored in~$R_i$. Whether~$\rho_s$ ends with a robber
or a cop position depends on the current position in the game with~$r$ robbers:
either both end with a cop position, or both end with a robber position.
Set~$R_i$ represents the vertices occupied by robbers that are associated
with~$\rho_i$.
Elements~$O_i$ are sets of vertices where cops of longer histories are not
placed because, roughly, those vertices are reachable from~$R_i$
in~$\graphG-W_i$. 
However, we will see later that, in fact, sets~$O_i$ are more dynamic.

The strategy we described so far is the strategy from
as constructed in the case of undirected graphs, just this time
with omitted placements of cops. Now we drop the assumption that robbers 
from~$R_i$ stay idle. They may prevent the cops to play against the robber 
from the longest history~$\rho_s$. 
One possibility is that one of them, say the robber from~$b_i\in R_i$, for
some~$i<s$,
jumps to the robber on~$b_s$, in a position~$(U,U',R)$ of the game with~$r$
robbers.
Then both robbers (the one from~$b_s$ and the one 
who jumped to~$b_s$) run to vertices~$b'_i$ and~$b'_s$ in different components
of~$\flap_{U'}^U(b_s)$.\,\footnote{Recall the definition of a~$\flap_{U'}^U(b_s)$ on
Page~\pageref{def_flap}.} 
Now some cops from~$U'$ may be reachable from~$b'_i$ and cannot be removed
as~$f$ may prescribe to play against~$b'_s$ later. Previously,
we used cops from team corresponding to robber~$b'_i$ who remained on~$U_i$
(which is now~$U'$) and cops from team~$b'_s$ pursued~$b'_s$. Thus we have to
reuse 
cops from~$U_i$, but they cannot be just removed before cop placements are
made up 
that were omitted because of the robber on~$b_i$. Instead, we let the cops 
from~$U_i$ play according to~$f$ from~$U_i$ until they occupy the same
vertices as 
cops from~$U_{i+1}$ of the next longer history. While this is done
the cop vertices are stored in~$\rho_i$. Then~$\rho_i$ and~$\rho_{i+1}$ are
merged. 

Note that it does not suffice to catch up all moves between the ends of~$\rho_i$
and ~$\rho_{i+1}$ in one move placing cops as in the last position
of~$\rho_{i+1}$.
The robber may use the absence of the cops in the intermediate positions and run
to 
a vertex such that the resulting placement of that robber and the cops is not
consistent
with~$f$.

There is an other case when the cops have to play in a different way:
the robber corresponding the longest history is captured or jumps away. 
In this case, his component is not reachable for any robber any more, as 
the robbers play according to an isolating strategy. We remove the cops from the
graph placed since the last position in~$\rho_{s-1}$, \ie since the last time
the robbers from~$\rho_{s-1}$ and~$\rho_s$ ran into different components. Then
we choose another robber from $R_{s-1}$ to chase and append a new history
to~$\rho$.

\subsection*{Formal description, the rest invariants and the full memory structure}
Now we present the strategy~$\rf$ and the memory updates formally.
Given a position~$(U,R)$ or~$(U,U',R)$ of the game with~$r$ robbers
and a memory state
\[\zeta = ((\rho_1, R_1,O_1), \ldots,
(\rho_{s-1},R_{s-1},O_{s-1}),\rho_{s})\,,\]
we define the new set~$U' = \rf((U,R),\zeta)$ of vertices
occupied by cops  (if the current position belongs to the cops) and the new
memory state 
\[\zeta' = ((\rho_1', R_1',O_1'), \ldots,
(\rho_{s'-1}',R_{s'-1}',O_{s-1}'),\rho_{s'}')\,.\]

We also maintain some additional invariants. To describe them, we 
define~$W^{-1}_i$,~$W_i$,~$W^i$,~$b_i$,~$U_i$,~$R_i$ and~$O_i$ such that

\begin{itemize}
\item~$\last(\rho_i) = (W^{-1}_i,W_i,b_i)$, for~$i\in\{1,\dots,s-1\}$,
\item~$\last(\rho_s) \in \{ (W_s,b_s),(W_s^{-1},W_s,b_s)\}$,
\item~$U_i = W_i\setminus O^{i-1}$,~$U^i = \bigcup_{j=1}^i U_j$ and 
     ~$W^i = \bigcup_{j=1}^i W_j$, for~$i\in\{1,\dots,s\}$, 
\item~$R^i = \bigcup_{j=1}^i R_j$ and~$O^i = \bigcup_{j=1}^i O_j$,
for~$i\in\{1, \dots, s-1\}$,
\item~$R_s = \{b_s\}$, if~$b_s \in R$ and~$R_s = \0$ otherwise.
\end{itemize}
In other words,~$W_i$ is the placement of the cops in the last position
of the play~$\rho_i$ as it is stored (without respecting that some moves were 
omitted),~$b_i$ is the stored position of the robber in that play (but the
robber may be somewhere else in the play with~$r$ robbers). 
Furthermore,~$U_i\subseteq U\cap W_i$ is the set of cops who are 
indeed placed and belong to~$\rho_i$, and~$O_i$ is the set of vertices
on which we do not place cops from~$\prefix$-greater plays even if~$f$
prescribes to do so.

\paragraph*{Invariants}

\begin{itemize}
 \item[] \textbf{(Robs)} The sets~$R_i$ are pairwise disjoint and~$R =
\bigcup_{i=1}^s R_i$.

 \item[] \textbf{(Cops)}~$U = \bigcup_{i=1}^s U_i$. 

 \item[] \textbf{(Omit)} For all~$i \in \{1, \ldots, s-1\}$,~$R_i \subseteq O_i
= \Reach_{\graphG - W_i}(O_i)$.

 \item[] \textbf{(Ext)} For all~$i \in \{1, \ldots, s-1\}$,~$O_i \subseteq
\Reach_{\graphG - W^{-1}_i}(b_i)$.

\end{itemize}

Conditions~(Omit) and~(Ext) describe what sets~$O_i$ actually are. We
assume that a robber may occupy or reach~$b_i$. From here, he
threatens all vertices that are reachable from~$b_i$
in~$\Reach_{\graphG - W^{-1}_i}(b_i)$, \ie if he is bounded in his
moves only by his own cops~$W^{-1}_i$. Note that~$W^{-1}_i$ are the
cops from the previous position of~$\rho_i$, but the cops~$W_i$ are not placed
yet: the robber can run in~$\graphG-(W^{-1}_i\cap W_i)$, but, as~$f$
is monotone, we can consider $\graphG-W^{-1}_i$ instead
of~$\graphG-(W^{-1}_i\cap W_i)$. In particular, the placement~$R_i$ of
the robbers is reachable from~$b_i$ in $G - W^{-1}_i$. Furthermore,~$O_i$ are closed
under reachability \emph{after} the cops are placed on~$W_i$.

In addition to (Cops), we also assume that, if~$(U,R)$ is a cop position and
$b_s \in R$ 
(the stored vertex of the robber in the longest play is indeed occupied by a robber),
then 
$\last(\rho_s) = (W_s, b_s)$.

The first part of (Omit) together with (Ext) guarantees that the last move of 
each robber who is associated with~$\rho_i$ is consistent with it.

\begin{lemma}\label{lemma_cons}
For all~$i\le s$ and for all~$b \in R_i$,~$\rho_i \cdot (W_i,b)$ is consistent
with~$f$.
\end{lemma}
\begin{proof}
By (Omit) we have~$b \in O_i$ and therefore, using (Ext), we obtain that~$b$ is
reachable from
$b_i$ in~$\graphG - W_i^{-1}$. Moreover, as~$\last(\rho_i) = (W_i^{-1}, W_i, b_i)$
and~$\rho_i$ is consistent 
with~$f$ according to (Cons),~$\rho_i \cdot (W_i, b)$ is consistent
with~$f$ as well.
\end{proof}

The next lemma, which follows from the monotonicity of~$f$, states
that every (stored) robber is bounded by his cops on their last vertices and is
not affected by previous placements.

\begin{lemma}\label{lemma_convenience}~
\begin{itemize}
\item[(1)] For any~$i \in \{1, \ldots, s-1\}$ and any~$b \in R_i$,
$\Reach_{\graphG - W_i}(b) = \Reach_{\graphG - W^i}(b)$.
\item[(2)]~$\Reach_{\graphG - W_s}(b_s) = \Reach_{\graphG - W^s}(b_s)$.
\end{itemize}
\end{lemma}
\begin{proof}
Consider some~$i \in \{1, \ldots, s-1\}$ and some~$b \in R_i$.
As~$W_i \subseteq W^i$, we have~$\Reach_{\graphG - W_i}(b) \supseteq \Reach_{\graphG
- W^i}(b)$, so
assume that the converse inclusion~$\Reach_{\graphG - W_i}(b) \subseteq
\Reach_{\graphG - W^i}(b)$ does
not hold. Then there is some~$u \in W^{i-1} \setminus W_i$ such that~$u \in
\Reach_{\graphG - W_i}(b)$.
Now if~$j \in \{1, \ldots, i-1\}$ such that~$u \in W_j$, then due to (Lin),
$\rho_j \propprefix \rho_i$. Moreover,~$\last(\rho_j) = (W_j^{-1},W_j,b_j)$
and, 
by Lemma~\ref{lemma_cons},~$\rho_i \cdot (W_i,b_i)$ is consistent
with~$f$, 
but as~$\rho_j$ is consistent with~$f$ as well due to (Cons),
$\Reach_{\graphG - W_i}(b) \cap W_j \neq \0$ 
contradicts the monotonicity of~$f$ (which is violated in position
$(W_i^{-1}, W_i, b_i)$). 

For~$b_s$, the argument is the same.
\end{proof}

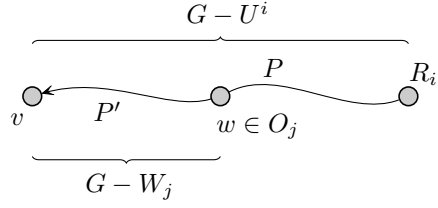
\begin{figure}
\begin{center}
\begin{tikzpicture}[thick]

\node (R_i) [vertex] at (0,0) {};

\node (R_iL) at (0.2,0.3) {$R_i$};

\node (v) [vertex] at (-5,0) {};

\node (vL) at (-5.2,-0.3) {$v$};

\node (P) at (-1.8,0.38) {$P$};
\node (G-U^i) at (-2.4, 1.1) {$\graphG-U^i$};
\draw [brace] (-5,0.7) -- (0,0.7);

\node (P') at (-4,-0.2) {$P'$};
\node (G-W_j) at (-3.7,-1.2) {$\graphG-W_j$};
\draw [brace] (-2.5,-0.8) -- (-5,-0.8);

\node (wL) at (-2,-0.4) {$w \in O_j$};

\draw[path] (R_i) .. controls  (-0.8,-0.3) and (-1.7,0.4) .. (-2.5,0)  ..
controls (-3.2,-0.2) and  (-4,0.3) .. (v) node [vertex,pos=0]{} {};

\end{tikzpicture}
\end{center}
\caption{$v \in \Reach_{\graphG - W_j}(O_j)$ implies~$v \in O_j$ by (Omit)}
\label{figure_lemma_le}
\end{figure}

The following lemma is one of the key arguments for monotonicity of~$\rf$.
It states that the robbers (who are indeed on the graph in the play with~$r$
robbers)
associated with play~$\rho_i$ are bounded by the cops (who are indeed on the
graph
in the play with~$r$ robbers) in a way that they can reach only vertices
in~$O_i$,
which are not occupied by cops from longer plays.
The lemma can be directly derived from~(Omit) without using other invariants.

\begin{lemma}\label{lemma_le}
For~$i\le s-1$,~$\Reach_{\graphG - U^i}(R_i) \subseteq O^i$.
\end{lemma}
\begin{proof}
Let~$v \in \Reach_{\graphG - U^i}(R_i)$ and let~$P$ 
be a path from~$R_i$ to~$v$ in~$\graphG - U^i$ as show in
Figure~\ref{figure_lemma_le}. 
If~$v \in \Reach_{\graphG - W_i}(R_i)$, then by (Omit) we have~$v \in
\Reach_{\graphG - W_i}(O_i) = O_i \subseteq O^i$.
Let therefore~$v \notin \Reach_{\graphG - W_i}(R_i)$. Then~$P \cap W_i \neq \0$
and we
consider the minimal~$l \le i$ such that~$P \cap W_l \neq \0$ and some~$w \in P
\cap W_l$. 
As~$P \cap U^i = \0$ we have~$w \notin U^i$ and thus~$w \notin U_l$,
as~$U_l\subseteq U^i$
by the definition of~$U^i$. As~$U_l = W_l\setminus O_{l-1}$, this yields 
$w \in O^{l-1}$, that means,~$w \in O_j$ for some~$j < l$. 
Now~$v$ is reachable from~$w$ in~$\graphG$ via some path~$P' \subseteq P$ and, due
to the minimal 
choice of~$l$,~$P \cap W_j = \0$. Hence,~$P' \cap W_j = \0$, see
Figure~\ref{figure_lemma_le}. 
This yields~$v \in \Reach_{\graphG - W_j}(w) \subseteq \Reach_{\graphG - W_j}(O_j)$
and as, by (Omit), 
$\Reach_{\graphG - W_j}(O_j) = O_j$ it follows that~$v \in O_j \subseteq O^i$.
\end{proof}

Finally, we formulate the fact that the reachability area of a robber is not
restricted by cops of
longer histories as a direct corollary of Lemma~\ref{lemma_le}.

\begin{corollary}\label{corollary_lemma_le}
For all~$i\in\{1, \dots, s-1\}$ and all~$b \in R_i$ we have 
$\Reach_{\graphG - U}(b) = \Reach_{\graphG - U^i}(b)$.
\end{corollary}

We proceed with a description of~$\rf$ and the memory update.

\paragraph*{Initial Move}
As we assumed that~$\graphG$ is strongly connected, by
Lemma~\ref{lemma_isolating_prudent}, 
the robbers do not split in the first move. So let the initial move be~$\bot
\rightarrow (\0, \{b\})$.
After the move, the memory state is set to~$\rho = \rho_1 = \big( (\0,b) )$. 
All the invariants hold obviously for~$(\0, \{b\})$ and~$\big( (\0,b) )$.

\medskip

Now we consider some cop position~$(U, R)$ and 
some memory state~$\zeta$ such that all invariants are fulfilled.

\medskip\noindent
\textbf{Move of the Cops.}
In the following, we define the new set~$U' = \rf((U,R),\zeta)$ of vertices
occupied by cops and the new memory state 
\[\zeta' = ((\rho_1', R_1',O_1'), \ldots,
(\rho_{s'-1}',R_{s'-1}',O_{s-1}'),\rho_{s'}')\,.\]

\noindent
\textbf{Case I:}~$b_s\notin R$\\
That means, the robber~$b_s$ which is stored in the longest history is not
on the graph
any more. Hence, if~$s = 1$ (the memory contains only one history),
then that robber has been captured and, as there are no other robbers, 
all the robbers are captured and the cops have won. 
Otherwise, we set~$U' := U^{s-1} = \bigcup_{i = 1}^{s-1} U_i$, \ie
we remove the cops corresponding to the longest history from the graph.
For the memory update, consider~$\rho_{s-1}$ and distinguish two cases:

\begin{itemize}
\item~$R_{s-1} = \emptyset$\\
That means, there are no robbers on the graph that are associated with 
the next longest history. 
The new memory state~$\zeta'$ is obtained from~$\zeta$ by deleting~$\rho_s$ and 
replacing~$(\rho_{s-1}, R_{s-1}, O_{s-1})$ by the history~$\rho_{s-1} \cdot
(W_{s-1},b_s)$. 
Note that we could delete all last plays from~$\rho$ that have no associated
robbers on the graph at once, but, for the ease of proving our invariants,
we do it step by step.

\item~$R_{s-1} \neq \emptyset$\\ 
In this case, we have to select one of the robbers from~$R_{s-1}$ that we want
to
pursue next.
Choose some robber~$b \in R_{s-1}$ and define 
$\tilde{O}_{s-1} \coloneqq \Reach_{\graphG - W_{s-1}}(R_{s-1} \setminus \{b\})$.
Then 
the new memory state~$\zeta'$ is obtained from~$\zeta$ by replacing 
$(\rho_{s-1}, R_{s-1}, O_{s-1})$ by~$(\rho_{s-1}, R_{s-1} \setminus \{b\},
\tilde{O}_{s-1})$ 
and replacing~$\rho_s$ by~$\rho_{s-1} \cdot (W_{s-1}, b)$. 

\end{itemize}

\noindent
\textbf{Case II:}~$b_s\in R$.

\noindent
\textbf{Case II.1:} There is some~$i \in \{1, \dots, s-1\}$ such that~$R_i =
\emptyset$.\\
That means, there is no robber associated with history~$\rho_i$. 
First, consider the next robber move in~$\rho_i$ according to~$\rho_{i+1}$ (note
that~$i < s$,
so~$\rho_{i+1}$ exists).
Consider a vertex~$\tilde{b}_i\in V$ and the suffix~$\eta$ of~$\rho_{i+1}$ such
that 
$\rho_{i+1} = \rho_i\cdot (W_i, \tilde{b}_i)\cdot \eta$. We distinguish three more cases.

\begin{itemize}
\item[\textbf{(a)}]\label{reach_max}~$\rho_{i+1} = \rho_i \cdot (W_i,
\tilde{b}_i) =
\rho_s$, \ie~$\eta$ is empty.\\
In this case,~$\rho_i$ already reached the end of~$\rho_{s}$, but is not deleted
yet.
Indeed, all histories~$\rho_i$, for~$i\le s-1$, end with a robber position.
If~$\eta$
is empty, then~$i=s-1$. Set~$U' \coloneqq U$, \ie the cops stay idle, and update
the
memory by deleting 
$(\rho_i, R_i, O_i)$ from~$\zeta$. 
\end{itemize}

\noindent
For the other cases, we set 
\begin{itemize}
 \item~$\tilde{W}_i \coloneqq f(W_i, \tilde{b}_i)$ and
 \item~$U' \coloneqq \bigcup_{j \neq i}U_j \cup (\tilde{W}_i \setminus O^{i-1})$
\end{itemize} 
to define the next cop move and
\begin{itemize}
 \item~$\tilde{O}_i = (O_i \cap \Reach_{\graphG - W_i}(\tilde{b}_i)) \setminus
\tilde{W}_i$ and
 \item~$\tilde{\rho}_i = \rho_i \cdot (W_i, \tilde{b}_i) \cdot  (W_i,
\tilde{W}_i, \tilde{b}_i)$
\end{itemize} 
for the definition of the memory update.

\begin{enumerate}
\item[\textbf{(b)}]\label{no_reach}~$\tilde{\rho}_i \neq \rho_{i+1}$.\\
That means, we have not reached the end of the next history. In this case,
we replace~$(\rho_i,R_i,O_i)$
by~$(\tilde{\rho}_i,R_i,\tilde{O}_i)$. 

 \item[\textbf{(c)}]\label{reach_next}~$\tilde{\rho}_i = \rho_{i+1}$.\\
The memory update is to replace~$(\rho_{i+1}, R_{i+1}, O_{i+1})$ by 
$(\rho_{i+1}, R_{i+1}, O_{i+1}\cup \tilde{O}_i)$ 
and to remove~$(\rho_i, R_i, O_i)$. Note how we conservatively updated~$O_i$.

\end{enumerate}

\noindent
\textbf{Case II.2:} For all~$i \in \{1, \dots, s-1\}$ we have~$R_i \neq
\emptyset$.\\
In this case, the cops play against the robber from~$\rho_s$.
We define 
\begin{itemize}
\item~$\tilde{W}_s = f(W_s,b_s)$ and
\item~$U' \coloneqq \bigcup_{j < s}U_j \cup (\tilde{W}_s \setminus O^{s-1})$
\end{itemize}
and, for the memory update, we replace~$\rho_s$ by~$\rho_s' = \rho_s \cdot (W_s,
\tilde{W}_s, b_s)$.\\

As a next step, we prove that the cop moves from~$U$ to~$U'$ is monotone,
\ie 
that no robber can reach any vertex from~$U \setminus U'$ in~$\graphG - (U \cap
U')$.

\begin{figure}
\begin{center}
\begin{tikzpicture}[thick]

\node[vertex] (R_l) at (0,0){};
\node[vertex] (tilde-b_i) at (-4,0){};
\node[vertex] (v) at (-5,-1){};
\node[vertex] (w) at (-2,-1){};

\draw[path] (tilde-b_i) .. controls (-3.2,-0.3) and (-2.7,-0.3).. (-1.9,0.2) ..
controls (-1,0.5) and (-0.5,0.2).. (R_l);
\draw[path] (R_l.-90) ..controls (-0,-1) and (-1,-0.5) .. (w)
        ..controls (-2.5,-1.4) and (-4.5,-1.4) .. (v);

\node (R_lText) at (0.2,0.35) {$R_l$};
\node (tilde-b_iText) at (-4.4,0){$\tilde{b}_i$};
\node (vText) at (-5.3,-1.3){$v$};
\node (wText) at (-1,-1.5){$w\in O_j, j<i$};
\node (PText) at (0.2,-0.55){$P$};
\draw[brace] (0.5,0)--(0.5,-1);
\node (G-CcapU') at (2,-0.55){$\graphG-(U\cap U')$};
\node (P') at (-3.4,-1.7){$P'$};
\draw[brace] (-2,-2.1) -- (-5,-2.1);
\node (G-W_j) at (-3.45,-2.7){$\graphG-W_j$};

\end{tikzpicture}
\end{center}
\caption{Robbers from longer histories than~$\rho_i$ cannot cause
non-monotonicity.}
\label{figure_lemma_monotonicity}
\end{figure}
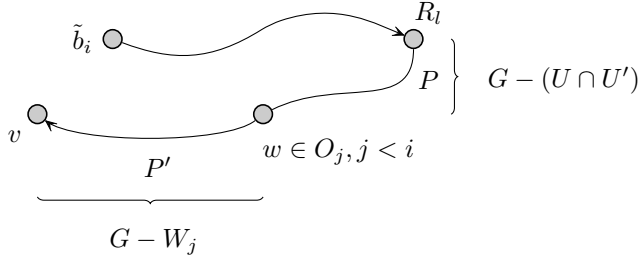

\begin{lemma}\label{lemma_monotonicity_app}
$(U\setminus U')\cap\Reach_{\graphG-(U \cap U')}(R)=\emptyset$.
\end{lemma}
\begin{proof}
We go through the cases defined in the description of the cop move.

\medskip

\emph{Case~I.}
If~$b_s \notin R$ we have~$U' = U^{s-1}$ so, by (Cops),~$U^{s-1} \subseteq U
\cap U'$.
Moreover, (Robs) yields~$R = \bigcup_{i = 1}^{s-1} R_i$ and hence, using
Lemma~\ref{lemma_le}, 
we obtain~$\Reach_{\graphG - (U \cap U')}(R) \subseteq O^{s-1}$. Due to the 
definition of~$U_s$ we have~$O^{s-1} \cap U_s = \0$, which implies
$\Reach_{\graphG - (U \cap U')}(R) \cap U_s = \0$ and thus, by (Cops), the move 
of~$\rf$ is monotone in this case.

\medskip

\emph{Case~II.}
Here we have~$b_s\in R$ and two further cases. 

\medskip

\emph{Case~II.1:} there is some~$i \in \{1, \dots, s-1\}$ such that~$R_i =
\emptyset$.
In Subcase~(a), the cops stay idle, so the move is monotone. Otherwise we have
$U' = \bigcup_{j \neq i} U_j \cup \tilde{U}_i$ with~$\tilde{U}_i = \tilde{W}_i
\setminus O^{i-1}$ where 
$\tilde{W}_i = f(W_i,\tilde{b}_i)$ and~$\rho_{i+1} =
\rho_i\cdot(W_i,\tilde{b}_i)\cdot\eta$ are as above. 
Assume that this move is not monotone, \ie there is some~$v \in U \setminus U'$
with 
$v \in \Reach_{\graphG - (U \cap U')}(R)$. 
Then~$v \in U_i \setminus \tilde{U}_i$, by the definition of~$U'$ and (Cops). 

We distinguish, which robbers can reach~$v$. 
First, consider robbers from 
smaller histories than~$\rho_i$, that means, from the set~$R^{i-1}$.
As~$U^{i-1} \subseteq U \cap U'$, by Lemma~\ref{lemma_le}, we obtain
$\Reach_{\graphG - (U \cap U')}(R^{i-1}) \subseteq O^{i-1}$.
Due to the definition of~$U_i$, we have~$O^{i-1} \cap U_i = \0$ and hence
$v \notin \Reach_{\graphG - (U \cap U')}(R^{i-1})$, \ie no robber from~$R^{i-1}$
can cause non-monotonicity.

As~$R_i = \0$, we have~$v \in \Reach_{\graphG - (U \cap U')}(R^{>i})$ where
$R^{>i} = \bigcup_{l = i+1}^{s-1} R_l \cup \{b_s\}$ is the set of
robbers from longer histories than~$\rho_i$. Consider some path~$P$ from 
$R^{>i}$ to~$v$ in~$\graphG - (U \cap U')$ as shown in
Figure~\ref{figure_lemma_monotonicity}.

First, we show that~$v \notin \Reach_{\graphG - (W_i \cap \tilde{W}_i)}(R^{>i})$,
\ie
that the robber has to visit omitted vertices.
For~$l \in \{i+1, \ldots, s-1\}$ and any~$b \in R_l$, by (Lin),~$\rho_i \cdot
(W_i, \tilde{b}_i)$ is a strict prefix of
$\rho_l \cdot (W_l, b)$ and, by Lemma~\ref{lemma_cons}, both of these histories 
are consistent with~$f$. 
So, by monotonicity of~$f$, any robber~$b \in R_l$ is reachable from
$\tilde{b}_i$ in~$\graphG - W_i$ and hence 
in~$\graphG - (W_i \cap \tilde{W}_i)$. 
Moreover, as we are in Case~II.1 (b) or (c), the same arguments show
that~$b_s$ 
is also reachable from~$\tilde{b}_i$ in~$\graphG - W_i$ and hence in~$\graphG - (W_i
\cap \tilde{W}_i)$. 
Therefore, if~$v \in \Reach_{\graphG - (W_i \cap \tilde{W}_i)}(R^{>i})$, then 
$v \in \Reach_{\graphG - (W_i \cap \tilde{W}_i)}(\tilde{b}_i)$.
But as~$v \in U_i \subseteq W_i$ this contradicts monotonicity of~$f$
since 
$\rho_i \cdot (W_i, \tilde{b}_i) \cdot (W_i, \tilde{W}_i, \tilde{b}_i)$ is
consistent with~$f$. 
Hence,~$v \notin \Reach_{\graphG - (W_i \cap \tilde{W}_i)}(R^{>i})$. 

As the robber visits omitted vertices,~$P \cap (W_i \cap \tilde{W}_i) \neq \0$. 
We consider the minimal~$l \le i$ such that 
$P \cap \widehat{W}_l \neq \0$ where~$\widehat{W}_j = W_j$ for~$j < i$ and 
$\widehat{W}_i = W_i \cap \tilde{W}_i$. We define~$\widehat{U}_j$ analogously.
The meaning of~$\widehat{W}_j$ is that it contains precisely the vertices
occupied 
by cops according to~$\rho_j$ which remained idle in the last move. Let~$w$ be
some 
vertex in~$P \cap \widehat{W}_l$. 
First, as~$w \in P$,~$w \notin U \cap U'$, so (Cops) and the definition of~$U'$
yield~$w \notin \widehat{U}_l$.
Therefore,~$w \in \widehat{W}_l \setminus \widehat{U}_l$ and hence, using the
definitions of~$U_l$ 
and~$\tilde{U}_i$, if~$l = i$, we obtain~$w \in O^{l-1}$, \ie~$w \in O_j$ for
some~$j < l$.
Moreover,~$v$ is reachable from~$w$ in~$\graphG$ via some path~$P' \subseteq P$
and, due to the 
minimal choice of~$l$,~$P' \cap \widehat{W}_j = P' \cap W_j = \0$, so~$v \in
\Reach_{\graphG - W_j}(w)
\subseteq \Reach_{\graphG - W_j}(O_j) = O_j \subseteq O^{i-1}$. The last equality
is due to (Omit).
But as~$O^{i-1} \cap U_i = \0$,~$v \in O^{i-1}$ is a contradiction to~$v \in
U_i$.

Finally, consider Case~II.2, \ie 
for all~$i \in \{1,\dots,s-1\}$ we have~$R_i \neq \emptyset$.
First notice that, due to definition of~$U'$ and (Cops),~$U \setminus U'
\subseteq U_s$.
For robbers other than~$b_s$ the same arguments as in Case~I and Case~II.1,
using (Robs) and Lemma~\ref{lemma_le},
show that they cannot cause non-monotonicity.
The argument for~$b_s$ is the same as in Case~II.1:
assume that~$b_s$ causes non-\mbox{}monotonicity at some vertex~$v$. As~$\rho_s$
is
consistent 
with~$f$ due to (Cons) and~$f$ is monotone,~$b_s$ can reach~$v$ only
via~$O^{s-1}$ (using (Cops)). 
However,~$O^{s-1}$ is closed under reachability in~$\graphG - U$ and~$v$ cannot be
in~$O^{s-1}$, so this is impossible.
\end{proof}

For the cop move, it remains to prove that all invariants still hold after the
move.
We first give a separate lemma for (Robs), (Lin), (Cons) and (Ext) and prove
them quite briefly 
as they can be obtained easily from the induction hypothesis, using the
definition of the cop move. 

\begin{lemma}\label{lemma_cop_invariants}
(Robs), (Lin), (Cons) and (Ext) are preserved by the cop move.
\end{lemma}
\begin{proof}
(Robs) follows immediately from the induction hypothesis.
Linearity of~$\propprefix$ is obviously preserved in Case~I, Case~II.1~(a)
and~(b) and in Case~II.2.
In Case~II.1~(b), we have to show that~$\tilde{\rho}_i \propprefix \rho_{i+1}$.
First
notice that~$\rho_i \cdot (W_i, \tilde{b}_i) \propprefix \rho_{i+1}$ as
$\rho_{i+1} = \rho_i \cdot (W_i, \tilde{b}_i) \eta$ and~$\eta \neq \0$.
Furthermore, the first
position in~$\eta$ is~$(W_i, \tilde{W}_i, \tilde{b}_i)$ as~$\rho_{i+1}$ is
consistent with
$f$ by (Cons) and~$\tilde{W}_i = f(W_i, \tilde{b}_i)$.
As~$\tilde{\rho}_i \neq \rho_{i+1}$ it follows that~$\tilde{\rho}_i \propprefix
\rho_{i+1}$.

For (Cons), consider first Case~I. If~$R_{s-1} = \0$, then~$\rho_{s'}' =
\rho_{s-1}\cdot (W_{s-1},b_s)$.
As~$\last(\rho_s) \in \{(W_s^{-1},W_s,b_s), (W_s, b_s)\}$ and~$\rho_{s-1}
\propprefix \rho_s$ 
and due to (Cons) both of these histories
are consistent with~$f$, which is monotone,~$b_s$ is reachable from
$b_{s-1}$ in
$\graphG - (W_{s-1}^{-1} \cap W_{s-1})$, so~$\rho_{s-1} \cdot (W_{s-1},b_s)$ is
consistent
with~$f$. If~$R_{s-1} \neq \0$, then~$\rho_{s-1} \cdot (W_{s-1},b)$ is
consistent
with~$f$ for any~$b \in R_{s-1}$ due to Lemma~\ref{lemma_cons}.
In Case~II.1 (a) and (b), (Cons) follows immediately from the induction
hypothesis.
In Case~II.1~(b), (Cons) follows from (Lin) as~$\rho_{s'}' = \rho_s$
is consistent with~$f$ and~$\tilde{\rho}_i \propprefix \rho_s$. 
Finally, in Case~II.2,~$\rho_s$ is consistent with~$f$ due to (Cons) and 
$W_s' = f(W_s,b_s)$, so~$\rho_{s'}' = \rho_s'$ is consistent with~$f$
as well.

To prove (Ext) first notice that in Case~I, if~$R_{s-1} = \0$, then (Ext)
follows immediately from the induction hypothesis. Moreover, if~$R_{s-1} \neq
\0$, 
then~$s' = s$ and we have to show that~$O_{s-1}' = \tilde{O}_{s-1}
\subseteq \Reach_{\graphG - W_{s-1}^{-1}}(b_{s-1})$. As, by
Lemma~\ref{lemma_cons}, for any~$b' \in R_{s-1}$ the history
$\rho_{s-1}(W_{s-1},b')$ is consistent with~$f$, which is monotone, the
reachability area of any~$b' \in R_{s-1}$ in~$\graphG - W_{s-1}$ is a subset of
the reachability area of~$b_{s-1}$ in~$\graphG - W_{s-1}^{-1}$. Hence, by
definition of~$\tilde{O}_{s-1}$, the statement follows. In Case~II, (Ext)
follows easily from the induction hypothesis, using the definition 
of~$\tilde{O}_i$ in Case~II.1 (b) and (c).
\end{proof}

For the remaining two invariants (Omit) and (Cops), we have two separate lemmas
which we prove 
in greater detail. 
The most interesting cases in the proofs of these two invariants are Cases~II.1
(b) and (c).
The crucial point here is the new set~$\tilde{O}_i$. See
Figure~\ref{figure_lemma_cop_invariants_omit} 
for an illustration.

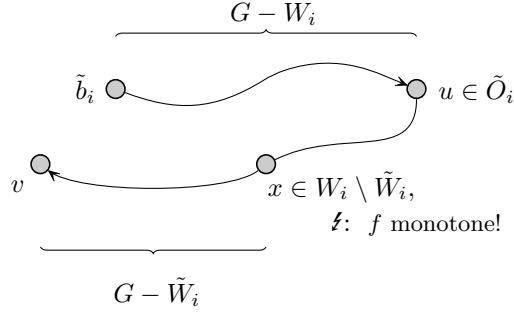
\begin{figure}
\begin{center}
 \begin{tikzpicture}[thick]
\node[vertex] (u) at (0,0){};
\node[vertex] (tilde-b_i) at (-4,0){};
\node[vertex] (v) at (-5,-1){};
\node[vertex] (x) at (-2,-1){};

\draw[path] (tilde-b_i) .. controls (-3.2,-0.3) and (-2.7,-0.3).. (-1.9,0.2) ..
controls (-1,0.5) and (-0.5,0.2).. (u);
\draw[path] (u.-90) ..controls (-0,-1) and (-1,-0.5) .. (x)
        ..controls (-2.5,-1.4) and (-4.5,-1.4) .. (v);

\node (R_lText) at (0.8,0) {$u\in \tilde{O}_i$};
\node (tilde-b_iText) at (-4.4,0){$\tilde{b}_i$};
\draw[brace] (-4,0.7)--(0,0.7);
\node (G-W_i) at (-1.9,1){$\graphG-W_i$};
\node (vText) at (-5.3,-1.3){$v$};
\node (xText) at (-1,-1.3){$x\in W_i\setminus \tilde{W}_i,$};
\draw[brace] (-2,-2.1) -- (-5,-2.1);
\node (G-tildeW_i) at (-3.45,-2.7){$\graphG-\tilde{W}_i$};
\node (contradiction) at (0,-1.8){\Lightning: \small~$f$ monotone!};

 \end{tikzpicture}
\end{center}
\caption{$\tilde{O}_i$ is closed under reachability in~$\graphG - \tilde{W}_i$}
\label{figure_lemma_cop_invariants_omit}
\end{figure}

\begin{lemma}\label{lemma_cop_invariants_omit}
(Omit) is preserved by the cop move.
\end{lemma}
\begin{proof}
In Case~I, if~$R_{s-1} = \0$, (Omit) follows immediately from the induction
hypothesis, 
so consider Case~II where~$R_{s-1} \neq \0$. We have~$s' = s$,~$W_{s-1}' =
W_{s-1}$ and 
$O_{s-1}' = \tilde{O}_{s-1} = \Reach_{\graphG - W_{s-1}}(R_{s-1} \setminus
\{b\})$. Clearly,
this yields that~$O_{s-1}'$ is closed under reachability in~$\graphG - W_{s-1}$.
Moreover, by (Omit),~$R_{s-1} \subseteq \Reach_{\graphG - W_{s-1}}(O_{s-1})$, so
we have
$R_{s-1} \cap W_{s-1} = \0$ an hence~$R_{s-1} \setminus \{b\} \subseteq 
\Reach_{\graphG - W_{s-1}}(R_{s-1} \setminus \{b\}) = O'_{s-1}$.

Consider Case~II.1. In Case~(a), (Omit) follows immediately from the induction
hypothesis. 
In Case~(b),~$R_i' \subseteq O_i'$ is trivial as~$R_i' = R_i = \0$, so we have
to show that 
$O_i' = \Reach_{\graphG - W_i'}(O_i')$.
We have~$O_i' = \tilde{O}_i = (O_i \cap \Reach_{\graphG - W_i}(\tilde{b}_i))
\setminus \tilde{W}_i$ 
and~$W_i' = \tilde{W}_i = f(W_i, \tilde{b}_i)$. 
Moreover, by the definition of~$\tilde{O}_i$ in this case, we have~$\tilde{O}_i
\cap
\tilde{W}_i = \0$, so 
$\tilde{O}_i \subseteq \Reach_{\graphG - \tilde{W}_i}(\tilde{O}_i)$.

As a next step, we show that~$\tilde{O}_i$ is closed under reachability in
$\graphG - \tilde{W}_i$. Let~$v \in \Reach_{\graphG - \tilde{W}_i}(\tilde{O}_i)$.
Clearly,~$v \notin \tilde{W}_i$. Let~$u \in \tilde{O}_i$ such that~$v$ is
reachable from~$u$ in~$\graphG - \tilde{W}_i$. As~$\tilde{O}_i = (O_i \cap
\Reach_{\graphG - W_i}(\tilde{b}_i)) \setminus W_i$, we have~$u \in \Reach_{\graphG
- W_i}(\tilde{b}_i)$ and~$v \in \Reach_{\graphG - \tilde{W}_i}(u)$. Therefore,
there is a cop-\mbox{}free path from~$\tilde{b}_i$ to~$v$ via~$u$ in~$\graphG -
(W_i \cap
\tilde{W}_i)$. By (Cons),~$\rho_i$ is consistent with~$f$ and~$\tilde{W}_i
= f(W_i, \tilde{b}_i)$, so, as~$f$ is monotone, this path must be
cop-\mbox{}free in~$\graphG - W_i$, see Figure
\ref{figure_lemma_cop_invariants_omit}.
Thus,~$v \in \Reach_{\graphG - W_i}(u)$ and as~$u \in O_i$ (by the definition of
$\tilde O_i$) and~$u \in \Reach_{\graphG - W_i}(\tilde{b}_i)$, we have~$v \in
\Reach_{\graphG - W_i}(O_i)$ and~$v \in \Reach_{\graphG - W_i}(\tilde{b}_i)$.
By (Omit), we have~$\Reach_{\graphG - W_i}(O_i) = O_i$, so~$v
\in O_i \cap \Reach_{\graphG - W_i}(\tilde{b}_i)$ and as~$v \notin \tilde{W}_i$
this yields~$v \in \tilde{O}_i$.

In Case~(c), we have to show that~$R_i' \subseteq O_i'$ 
and that~$O_i'$ is closed under reachability in~$\graphG - W_i'$.
We have~$R_i' = R_{i+1}$,~$O_i' = O_{i+1} \cup \tilde{O}_i$ and~$W_i' =
W_{i+1}$.
By (Omit),~$R_{i+1} \subseteq O_{i+1} \subseteq O_{i+1} \cup \tilde{O}_i$.
Moreover, as in Case~(b),~$\tilde{O}_i$ is closed under reachability in
$\graphG - \tilde{W}_i$ and as~$\tilde{\rho}_i = \rho_{i+1}$ we have~$\tilde{W}_i
= W_{i+1}$.
By (Omit),~$O_{i+1}$ is closed
under reachability in~$\graphG - W_{i+1}$, so the union 
$O_{i+1} \cup \tilde{O}_i$ is closed under reachability in~$\graphG - W_{i+1}$ as
well.
Finally, in Case~II.2, (Omit) follows again from the induction hypothesis.
\end{proof}

\begin{lemma}\label{lemma_cop_invariants_cops}
(Cops) is preserved by the cop move.
\end{lemma}
\begin{proof}
We have to show that
$U' = \bigcup_{j = 1}^{s'} U_j'$ where~$s' \in \{s-1,s\}$ is the length of
$\zeta'$.
Note that, by the definition,~$U_j' = W_j' \setminus (O^{j-1})'$ for~$j = 1,
\ldots,
s'$.

In Case~I, Case~II.1~(a) and Case~II.2, this can easily be obtained using the
induction
hypothesis and the definition of~$U'$. Consider Case~II~(b). 
We have~$s' = s$ and~$O_j' = O_j$, for~$j \neq i$, and~$O_i' = \tilde{O}_i
\subseteq O_i$.
As, moreover,~$W_j' = W_j$ for~$j < i$, we have~$U_j' = U_j$, for~$j < i$. 
Furthermore,~$U_i' = W_i' \setminus (O^{i-1})' = \tilde{W}_i \setminus O^{i-1}$
and, as~$(O^{j-1})' \subseteq O^{j-1}$, for~$j = 1, \ldots, s$, we have 
$U_j \subseteq U_j'$, for~$j > i$. Hence,~$U' \subseteq \bigcup_{j = 1}^s U_j'$
and it 
remains to show~$\bigcup_{j=1}^s U'_j \subseteq U'$.

Towards a contradiction, assume that there is some 
$v \in (\bigcup_{j = 1}^s U_j') \setminus U'$.
Then~$v \in U_j'$, for some~$j > i$, and, as~$v \notin U' \supseteq U_j$, we
have 
$v \in W_j \setminus (O^{j-1})'$, but~$v \notin O^{j-1}$. Since~$O_l' = O_l$ 
for~$l \neq i$, we have~$v \in O_i \setminus O_i'
= O_i \setminus \tilde{O}_i$. So, by the definition of~$\tilde{O}_i$, we have 
$v \in \tilde{W}_i$ or~$v \notin \Reach_{\graphG - W_i}(\tilde{b}_i)$. 
As~$v \notin U'$ we have~$v \notin \tilde{W}_i \setminus O^{i-1}$ and, as 
$v \notin O^{j-1} \supseteq O^{i-1}$, it follows 
that~$v \notin \tilde{W}_i$, so~$v \notin \Reach_{\graphG - W_i}(\tilde{b}_i)$. 
Let~$\rho^* = \widehat{\rho} (W^{-1}, W, b)$ be the shortest prefix of~$\rho_j$ 
such that~$v \in W$. Note that such a prefix exists as~$v \in W_j$.
Due to (Cons),~$\tilde{\rho}_i$ and~$\rho^*$ are consistent with~$f$ and
$f$ is monotone, so since~$v \notin \tilde{W}_i$ we have~$\tilde{\rho}_i
\propprefix \rho^*$
and as~$v \notin \Reach_{\graphG - W_i}(\tilde{b}_i)$, we also have~$v \notin
\Reach_{\graphG - W^{-1}}(b)$. 
However, this is a contradiction to the fact that~$f$ is active.

Finally, in Case~(c), we have~$s' = s-1$, as we delete the~$i$th element
of~$\zeta$. Hence, 
we have a shift of indexes. Accounting for this fact, (Cops) can be proven
analogously to
the Case~(b).
\end{proof}

\medskip\noindent \textbf{Move of the Robbers.} Let~$R'$ be the set of vertices
occupied by robbers after their move. If~$R' = R$, we do not update the memory.
This happens in particular after the cop move in Case~I and in Case~II.1~(a) of
the cop move: in those cases, we do not place new cops on the graph,
so the robbers stay idle because they stick to a prudent strategy. 
We will not consider these cases.

Let~$R' \neq R$ and consider the memory state 
\[\zeta =
\big((\rho_1,R_1,O_1),\dots,(\rho_{s-1},R_{s-1},O_{s-1}),\rho_s\big)\] 
before the robber move from~$R$ to~$R'$. Note that~$b_s \in R$.

We will also need the memory state \[\olzeta =
((\olg_1, \olR_1, \olO_1), \ldots,
(\olg_{\ols-1}, \olR_{\ols-1},
\olO_{s-1}), \olg_{\ols})\] and the set~$U^{-1}$ of
vertices occupied by cops before the last cop moves.

\paragraph*{Assignment of the robbers to histories}
We assign every robber~$b\in R'$ to the shortest history~$\rho_i$
with~$b\in O_i$, which yields the new set~$\tilde{R}_i$ replacing~$R_i$:
\begin{itemize}
 \item[ ] If~$b \in O^{s-1}$, then let 
	 ~$i = \min \{ j \in \{1, \ldots, s-1\} \,|\, b \in O_j\}$ 
	  and assign~$b$ to~$\rho_i$.
	  Otherwise assign~$b$ to~$\rho_s$.

\end{itemize}

The crucial point we have to prove about the memory update after a robber move
is that a robber assigned to a certain history is
consistent with it according to~$f$. For the robbers
assigned to histories~$\rho_i$ with~$i < s$ this follows easily from the
fact that~$\tilde{R}_i \subseteq O_i$, similar as in Lemma~\ref{lemma_cons}. For
the robbers in~$\tilde{R}_s$ this is, however, much more involved. We have to
show that each such robber can be reached from~$\olb_{\ols} =
b_s$ in the graph~$\graphG - \olW_{\ols}$ which then shows that
prolonging the longest history by a move from~$b_s$ to some robber from
$\tilde{R}_s$ yields again an~$f$-history. This property is proved in
the following lemma.

\begin{figure}
\begin{center}
\begin{tikzpicture}[thick]
\node[vertex] (olO_j) at (0,0) {};
\node[vertex] (d) at (-2,0) {};
\node[vertex] (d') at (2,0) {};

\draw[path] (d') ..controls (1.5,-0.5) and (0.5,-0.5) .. (olO_j) .. controls
(-0.5,0.5) and (-1.5,0.5) .. (d);

\node (olO_jText) at (0.2,0.35){$\olO_j$};
\node (dText) at (-2.35,0){$d$};
\node (d'Text) at (2.9,0){$d'\in R_l$};
\node (PText) at (1,-0.2){$P$};
\node (P'Text) at (-1,0.17){$P'$};

\draw[brace] (-0.05,-1) -- (-2,-1);
\draw[brace] (2,-1) -- (0.05,-1);
\node (G-olW_j) at (-0.9,-1.5) {$\graphG-\olW_j$};
\node (G-U^-1) at (1.2,-1.5) {$\graphG-U^{-1}$};

\end{tikzpicture}
\end{center}
\caption{Any~$d \in \tilde{R}_s \setminus \Reach_{\graphG - 
\olW_{\ols}}(\olb_{\ols})$
is in~$\olO^{s-1}$.}
\label{figure_lemma_almost_ext1}
\end{figure}

\begin{lemma}\label{lemma_almost_ext}~$\tilde{R}_s \subseteq \Reach_{\graphG -
\olW_{\ols}}(\olb_{\ols})$. \end{lemma}
\begin{proof} Let~$d \in \tilde{R}_s$. As the robbers have moved from~$R$ to
$R'$ in their move, there is some~$d' \in R$ such that~$d$ is reachable from
$d'$ in~$\graphG - (U^{-1} \cap U)$. As we have already shown in
Lemma~\ref{lemma_monotonicity_app}, the move from~$U^{-1}$ to~$U$ was monotone,
so~$d$ is reachable from~$d'$ in~$\graphG - U^{-1}$. Let~$P$ be a path from~$d'$
to~$d$ in~$\graphG - U^{-1}$ and assume that~$d \notin \Reach_{\graphG -
\olW_{\ols}}(\olb_{\ols})$. We show that then~$d
\in O^{s-1}$ in contradiction to~$d \in \tilde{R}_s$ as by the definition of
$\tilde{R}_s$,~$\tilde{R}_s \cap O^{s-1} = \0$. By (Robs) for
$\olzeta$,~$R = \bigcup_{i=1}^{\ols}(\olR_i)$, so there
is some (unique)~$l \le \ols$ with~$d' \in \olR_l$.

First we show~$d \in \olO^{\ols-1}$, see
Figure~\ref{figure_lemma_almost_ext1}. If~$d' \neq \olb_{\ols}$,
then according to (Omit) for~$\olzeta$ we have~$d' \in \olR_l
\subseteq \olO_l$ and as~$d' \in P$, we have~$P \cap \olO^l \neq
\0$. In the other case we have~$d' = \olb_{\ols}$ so~$d \in
\Reach_{\graphG - U^{-1}}(\olb_{\ols})$ and as, by (Cops) for
$\olzeta$,~$\olU_{\ols} \subseteq U^{-1}$, we have~$d
\in \Reach_{\graphG - \olU_{\ols}}(\olb_{\ols})$.
However, by our assumption,~$d \notin \Reach_{\graphG -
\olW_{\ols}}(\olb_{\ols})$, so by the definition of
$\olU_{\ols}$,~$P \cap \olO^{\bar s-1} \neq \0$. Hence,
in any case we have~$P \cap \olO_j \neq \0$ for some~$j \le
\min\{\ols-1,l\}$ and we consider the minimal such~$j$. Then by (Cops)
for~$\olzeta$,~$\olU_j \subseteq U^{-1}$, so~$d$ is reachable
from~$\olO_j$ in~$\graphG - \olU_j$ via a path~$P' \subseteq P$,
see Figure~\ref{figure_lemma_almost_ext1}. So if~$d \notin \Reach_{\graphG -
\olW_j}(\olO_j)$, then, by the definition of~$\olU_j$, we have
$P \cap \olO^{j-1} \neq \0$, which contradicts the minimality of~$j$. Hence,
$d \in \Reach_{\graphG - \olW_j}(\olO_j) = \olO_j$ by
(Omit) for~$\olzeta$.

Now we show that~$d$ is also in~$O^{s-1}$. We distinguish the moves that the
cops may have made. Case~I and Case~II.1~(a) of the cop move do not have to be
considered here as discussed above. If~$\olO_j = O_j$, which in
particular holds in Case~II.2, then~$d \in O_j \subseteq O^{s-1}$. Now assume
that~$\olO_j \neq O_j$, so we are in Case~II.1 (b) or (c). Let~$i$ be as
in these cases. Then for all~$m < i$, we have~$\olO_m = O_m$, so~$j \ge
i$. Moreover, for all~$m > i$,~$\olO_m = O_m$ (in Case~II.(b)) or
$\olO_m \subseteq O_{m-1}$ (in Case~II.(c)), so either~$d \in O^{s-1}$
or~$j \le i$. The remaining case is~$j = i$. Note that in this case,~$j < l$ as
either~$l = \ols$ and~$j \le \ols-1$ or~$l < \ols$. In
the latter case, the reason is that~$j \le l$ and~$\olR_j =
\olR_i = \0$ and~$d' \in \olR_l \neq \0$. We show that~$d \in
\tilde{O}_j$, then by the definition of the memory update~$d \in O_j$ and
hence~$d
\in O^{s-1}$.

By definition,~$\tilde{O}_j = (\olO_j \cap \Reach_{\graphG -
\olW_j}(\tilde{b}_j)) \setminus \tilde{W}_j$ where~$\tilde{W}_j = W_j' =
W_j$ and~$\tilde{b}_j = b_j$. We have already shown that~$d \in \olO_j$.
In order to see that~$d \notin W_j$ notice that~$d \in R'$, and~$U_j \subseteq
U$ according to (Cops), so~$d \notin U_j$. Hence, if~$d \in W_j$, we have~$d \in
\olO^{j-1} = O^{j-1}$ by the definition of~$U_j$, contradicting~$d \in
\tilde{R}_s$. Thus,~$d \notin W_j$ and it remains to show that~$d \in
\Reach_{\graphG - \olW_j}(b_j)$. First notice that since~$j < l$, we have
$\tilde{\rho}_j \preccurlyeq \olg_{j+1} \preccurlyeq
\olg_l$. So as, according to (Cons), all these histories are
consistent with~$f$, which is monotone,~$\olb_l$ is reachable from
$b_j$ in~$\graphG - \olW_j$, see Figure~\ref{figure_lemma_almost_ext2}.
Now if~$l < \ols$,~$d' \in \olR_l$, so by (Ext),~$d'$ is
reachable from~$\olb_l$ in~$\graphG - \olW_l^{-1}$. Moreover,
using again that~$\tilde{\rho}_j \preccurlyeq \olg_l$ are both
consistent with~$f$ and that~$f$ is monotone, this yields that~$d'$ is
reachable from~$\olb_l$ in~$\graphG - \olW_j$. If, on the other
hand,~$l = \ols$, then~$d' = \olb_{\ols} =
\olb_l$, so clearly,~$d'$ is reachable from~$\olb_l$ in the
graph~$\graphG - \olW_j$. Therefore,~$d'$ is reachable from~$b_j$ in the
graph~$\graphG - \olW_j$ and as, by (Cops),~$\olU_j \subseteq
U^{-1}$,~$d$ is reachable from~$d'$ in~$\graphG - \olU_j$ via~$P$. Hence,
if~$d$ is not reachable from~$b_j$ in~$\graphG - \olW_j$, then due to the
definition of~$\olU_j$ there is some vertex from~$\olO^{j-1}$ on
the path~$P$ which contradicts the minimality of~$j$. Hence,~$d \in
\Reach_{\graphG - \olW_j}(b_j)$. \end{proof}

\begin{figure}
\begin{center}
\begin{tikzpicture}[thick]

\node[vertex] (d') at (1,0) {};
\node[vertex] (d) at (-3,0){};
\node[vertex] (olb_l) at (1,-2){};
\node[vertex] (b_j) at (-3,-2){};
\node[vertex] (x) at (-2,0){};

\draw[path] (d') ..controls (-0.2,0.3) and (-0.8,-0.5) .. (x) .. controls
(-2.5,0.2) .. (d);
\draw[path] (olb_l) .. controls (0.8,-1.6) and (1.5,-0.8) .. (d');
\draw[path] (b_j) .. controls (-1.5,-1.3) and (-0.8,-2.6) .. (olb_l);

\node (dInolO_jText) at (-3.9,0){$d\in \olO_j$};
\node (d'InolR_lText) at (2,0){$d'\in \olR_l$};
\node (olb_lText) at (1.5,-2){$\olb_l$};
\node (b_jText) at (-3.9,-2){$b_j=\tilde{b}_j$};

\draw[brace] (3,0) -- (3,-2);
\node (G-olW_j) at (4,-1){$\graphG - \olW_j$};
\node (PText) at (-0.5,0.3){$P$};
\draw[brace] (-2,0.7) -- (1,0.7);
\node (G-olW_j_above) at (-0.4,1.2){$\graphG - \olW_j$};
\draw[brace] (1,-2.7) -- (-3,-2.7);
\node (G-olW_j_below) at (-0.9,-3.3){$\graphG - \olW_j$};
\node[rotate=-15] (xText) at (-1.3,-0.5) {$x\in \olO^{j-1},$};
\node[rotate=-15] (contradiction) at (-1,-1.15) {\Lightning: \small~$j$
minimal!};

\end{tikzpicture}
\end{center}
\caption{If~$i = j < l$, the robber~$\tilde{b}_j$ can still 
reach~$d$ in the graph~$\graphG - \olW_j$ via~$d'$.}
\label{figure_lemma_almost_ext2}
\end{figure}
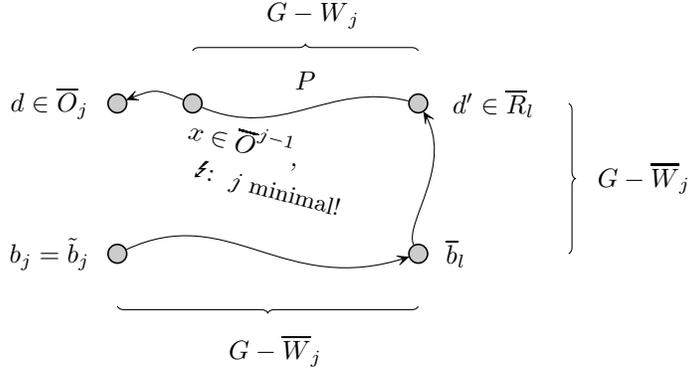

\paragraph*{Memory update}
For the memory update, we distinguish three cases according to the
number of robbers that have been assigned to~$\tilde{R}_s$, and according to 
whether the last position of~$\rho_s$ belongs to the cops or to the robber.
We simplify the case distinction by proving that if we
did not play against the robber in the longest history in the last cop move,
then at most the robber~$b_s = \olb_{\ols}$ can be
consistently associated with~$\rho_s$. 

\begin{lemma}\label{lemma_robbers_position}
If~$\rho_s$ ends with a position of the cop player,
then~$\tilde{R}_s \subseteq \{b_s\}$.
\end{lemma}
\begin{proof}
Assume that~$\rho_s$ ends with a cop position, \ie~$\rho_s = 
\widehat{\rho}_s (W_s, b_s)$. Then the last cop moves was not as in
Case~II.2. As Case~I and Case~II.1~(a) do not need to be considered as
discussed above, we have~$W_s = \olW_{\ols}$ (and~$b_s =
\olb_{\ols}$). So Lemma~\ref{lemma_almost_ext} yields
$\tilde{R}_s \subseteq \Reach_{\graphG - W_s}(b_s)$. By
Lemma~\ref{lemma_convenience} we have~$\Reach_{\graphG - W_s}(b_s) = \Reach_{\graphG
- W^s}(b_s)$, so~$\tilde{R}_s \subseteq \Reach_{\graphG - W^s}(b_s) \subseteq
\Reach_{\graphG - U}(b_s)$. Since~$b_s \in R$, it follows that~$\tilde{R}_s \not
\subseteq \{b_s\}$ contradicts the assumption that the robbers use a prudent
strategy. 
\end{proof}

There remain two other cases. 

\medskip
\noindent
\textbf{Case 1:}~$\rho_s$ ends with a position of the robber player 
and~$|\tilde{R}_s| \ge 1$.\\

Intuitively, this case means that the last cop move was according to~$\rho_s$
and~$|\tilde{R}_s| \ge 1$. In other words, at least one of the robbers 
from~$R'$ can be consistently associated with~$\rho_s$. 
(As we will see in Lemma~\ref{lemma_robber_invariants}, it follows from~(Cons)
that all robbers from~$\tilde{R}_s$ can be associated with~$\rho_s$.)

We choose one of the robbers~$b \in \tilde{R}_s$ which we pursue further (that
means,~$b$ will be the new robber from the longest history), and add a
new history~$\rho_{s'} = \rho_{s+1}$ extending~$\rho_s$ by the
robber move from~$b_s$ to~$b$. The remaining robbers~$\tilde{R}_s \setminus
\{b\}$
are still associated with~$\rho_s$. The new set~$O'_{s'-1} = O'_s$ 
contains exactly the vertices reachable from~$\tilde{R}_s \setminus \{b\}$ 
in~$\graphG - W_s$. 

Formally, we choose some~$b \in \tilde{R}_s$, define~$\tilde{O}_s =
\Reach_{\graphG
- W_s}(\tilde{R}_s \setminus \{b\})$ and set \[\zeta' = \big( (\rho_1,
\tilde{R}_1, O_1), \dots, (\rho_{s-1}, \tilde{R}_{s-1}, O_{s-1}), (\rho_s,
\tilde{R}_s \setminus \{b\}, \tilde{O}_s),\rho_s \cdot (W_s, b) \big)\,.\]

\noindent
\textbf{Case 2:}~$\rho_s$ ends with a position of the cop player
or~$|\tilde{R}_s|
= 0$.

This case means that either we did not play according to~$\rho_s$,
or we did, but~$b_s$ was captured or returned to an shorter~$\rho_i$. 

We define 
\[\zeta' = \big( (\rho_1, \tilde{R}_1, O_1), \dots, 
(\rho_{s-1}, \tilde{R}_{s-1}, O_{s-1}), \rho_s \big)\,.\]

\paragraph*{Invariants after the robber move}
Now we prove that all invariants still hold after the robber move.

\begin{lemma}\label{lemma_robber_invariants} 
All invariants are preserved by the robber move.
\end{lemma}
\begin{proof}
(Robs) holds by the definition of the sets~$\tilde{R}_i = R_i'$ and the
construction
of the memory update. (Lin) and (Cops) are obvious.

To prove (Omit), first notice that by (Omit) for~$\zeta$, each set~$O_i$ for~$i
= 1, \ldots, s-1$ is closed under reachability in~$\graphG - W_i$ and as, for~$i =
1, \ldots, s-1$, we have~$O_i' = O_i$ and~$\rho_i' = \rho_i$, the invariant
holds for all~$i = 1, \ldots, s-1 \ge s'-2$. Moreover,~$R_i' = \tilde{R}_i
\subseteq O_i = O_i'$ holds by the definition of the sets~$\tilde{R}_i$ for~$i =
1,
\ldots, s-1$. In particular, in Case~2, there is nothing to show, so consider
Case~1. We have~$s' = s+1$ and~$O_{s'-1}' = O_s' = \tilde{O}_s = \Reach_{\graphG -
W_s}(\tilde{R}_s \setminus \{b\})$, so~$O_s'$ is obviously closed under
reachability in~$\graphG - W_s$ and as~$W_{s'-1}' = W_s$,~$O_s'$ is closed under
reachability in~$\graphG - W_{s'-1}'$. It remains to show that~$R'_{s'-1}
\subseteq O'_{s'-1}$. First, we have~$W_s \cap (\tilde{R}_s \setminus \{b\}) =
\0$. Assume that, to the contrary, there is some~$v \in W_s \cap (\tilde{R}_s
\setminus \{b\})$. Then~$v \notin U$ (as a cop and a robber cannot be on the
same vertex) and according to (Cops) we have~$U = \bigcup_{i = 1}^s U_i$. So~$v
\notin U_s$ and hence, according to the definition of~$U_s$,~$v \in O^{s-1}$,
which contradicts~$v \in \tilde{R}_s$. So, indeed,~$W_s \cap (\tilde{R}_s
\setminus \{b\}) = \0$. Hence, by the definition of~$O_{s'-1}'$ in Case~1, we
have
$R_{s'-1}' = \tilde{R}_s \setminus \{b\} \subseteq \tilde{O}_s = O_{s'-1}'$ and
thus, (Omit) follows.

Notice that, by (Ext) for~$\zeta$,~$O_i \subseteq
\Reach_{\graphG - W_i^{-1}}(b_i)$ for~$i = 1, \ldots, s-1$ and as~$O_i' = O_i$ and
$\rho_i' = \rho_i$ for~$i = 1, \ldots, s-1$, the invariant holds for all~$i = 1,
\ldots, s-1 \ge s' - 2$. In particular, in Case~2, there is nothing to show
and we consider Case~1. First, notice that~$(W_{s'-1}')^{-1} = W_s^{-1} =
\olW_s$ and~$b_{s'-1}' = b_s = \olb_s$, so according to
Lemma~\ref{lemma_almost_ext}, we have~$R_{s'-1}' \subseteq \tilde{R}_s \subseteq
\Reach_{\graphG - W_s^{-1}}(b_s)$. Moreover, by the definition,~$O_{s'-1}' =
\tilde{O}_s = \Reach_{\graphG - W_s}(\tilde{R}_s \setminus \{b\})$. So if~$v \in
\tilde{O}_s$, then~$v$ is reachable from some~$\widehat{b} \in \tilde{R}_s
\setminus \{b\}$ in~$\graphG - W_s$ and as~$\tilde{R}_s \subseteq \Reach_{\graphG -
W_s^{-1}}(b_s)$,~$\widehat{b}$ is reachable from~$b_s$ in~$\graphG - W_s^{-1}$.
Thus,~$v$ is reachable from~$b_s$ in~$\graphG - (W_s^{-1} \cap W_s)$ and, as
$\rho_s = \widehat{\rho}(W_s^{-1}, W_s, b_s)$ is consistent with~$f$ by
(Cons) for~$\zeta$ and~$f$ is monotone, we have~$v \in \Reach_{\graphG -
W_s^{-1}}(b_s)$.

Finally, for (Cons), Case~2 is trivial. For Case~1, as~$\rho_s$ is consistent
with~$f$ by (Cons), it suffices to show that~$b \in \Reach_{\graphG -
W_s^{-1}}(b_s)$. However, we have shown in Lemma~\ref{lemma_almost_ext} that
$\tilde{R}_s \subseteq \Reach_{\graphG - \olW_s}(\olb_s)$ and as
in Case~1 we have~$\olb_s = b_s$ and~$\olW_s = W_s^{-1}$, this
follows from~$b \in \tilde{R}_s$. 
\end{proof}

It remains to show that, first,~$\rf$ uses at most~$r \cdot k$ cops and,
second, playing according to~$\rf$ the cops capture all robbers.

\paragraph*{Using at most~$k\cdot r$ cops}
By~(Cops), the number of cops is bounded by~$|\bigcup_{i=1}^sU_i|$. 
By definition of~$U_i$, we have~$|\bigcup_{i=1}^sU_i|\le |\bigcup_{i=1}^sW_i|$. 
Due to~(Cons), all~$W_i$ have size at most~$k$. Thus we have to show that there 
are at most~$r$ distinct sets~$W_i$.

\begin{lemma}\label{lemma_impinf_number_of_cops}
For any memory state~$\zeta$ consistent with~$\rf$ we have~$|\zeta| \le r+1$
and, 
if~$|\zeta| = r+1$, then~$W_s = W_{s-1}$.
\end{lemma}
\begin{proof}
In the following, we denote by~$\zeta$ the memory state before and 
by~$\zeta'$ the memory state after the cop move (and before the
robber move) and by~$\zeta''$ the memory state after the robber move.

If~$|\zeta| \le r$, then, by inspecting all cases, we can see that 
$|\zeta''| \le r$, or, in Case~1 of the robber move,~$|\zeta''| \le r+1$ and 
$W_s = W_{s-1}$. Consider the case
$|\zeta| = r+1$ and~$W_s = W_{s-1}$. As~$|R| \le r$, it follows from (Robs)
that~$R_i = \0$, for some~$i \in \{1, \ldots, s-1\}$ or~$b_s \notin R$. 

If~$b_s \notin R$, then, after the cop move, we either have~$|\zeta'| = r$
(if~$R_{s-1} = \0)$, or~$|\zeta'| = r+1$ and~$W_{s'}' = W_s' = W_{s-1} =
W_{s-1}' = W_{s'-1}'$ (if~$R_{s-1} \neq \0$). Moreover, in that case the memory
state after the robber moves (which is empty) is the same as after the
cop moves. 

Now assume that~$b_s \in R$ and let~$i \in \{1, \ldots, s-1\}$ be such that~$R_i
= \0$. Then in the cop move, we are in Case~II.1. If we are in Case~II.1~(a) or
in Case~II.1~(c), then we have~$|\zeta'| = r$ after the cop move, so
after the robber move,~$|\zeta''| \le r+1$ holds. If we are in
Case~II.1~(b), then after the cop move, we have~$s' = s$,~$\rho_{s-1}' =
\rho_{s-1}$ and~$\rho_s' = \rho_s$. Hence,~$W_{s-1}' = W_s'$ and, as~$\rho_s$
ends with a cop position (because after the robber move,~$\rho_s$ always ends in
a cop position and Case~II. (b) does not change~$\rho_s$),~$\zeta''$ is
constructed according to Case~3 of the memory update after the robber move.
Hence,~$|\zeta''| = |\zeta'| = r+1$ and~$W_{s''}'' = W_{s'}' = W_{s-1}'
= W_{s''-1}''$.
\end{proof}

\paragraph*{Capturing all robbers} 
To prove that~$\rf$ is winning, we, first, prove that an additional
invariant holds.
\begin{itemize}
\item[] \textbf{(Progress)} For~$i \in \{2, \ldots, s-1\}$,~$R_i \cap O^{i-1} =
\emptyset$ and~$b_s \notin O^{s-1}$ .
\end{itemize}

The invariant expresses that an~$O_i$ can only be a reason not to place 
\emph{any} cops when playing against robbers from smaller histories. Indeed, 
any winning strategy finally places a cop into the robber component, so after 
some omitted placements, some cop is really placed. This is true, in 
particular, for~$i=s$, which guarantees that the set of vertices available to 
the robbers shrinks.

The reason why we have to maintain that property also for the shorter play
prefixes is that when the robber leaves~$b_s$ one of shorter~$\rho_i$ 
becomes the longest one. 

Basically, (Progress) follows from the assumption that the robbers use an
isolating strategy. However, as the sets~$O_i$ are defined with respect to
reachability in~$\graphG - W_i$ and not in~$\graphG - U$, we have to transfer 
that topological incomparability from~$\graphG - U$ to~$\graphG - W_i$.

\begin{lemma}\label{lemma_progress}
(Progress) is preserved by both cop and robber moves.
\end{lemma}
\begin{proof}
First, consider the situation after the cop move.
In Case~I, we have~$R_j' = R_j$ and~$O_j = O_j'$ for~$j = 1, \ldots, s-2$ 
and hence,~$R_j' \cap (O^{j-1})' = \0$ by (Progress) for~$\zeta$.
Moreover, if~$R_{s-1} = \0$, then~$s' = s-1$, so~$s'-1 = s-2$ and it remains
to show that~$b_{s'} \notin O^{s'-1}$. However, as~$b_{s'} = b_s$ and~$O^{s'-1}
= O^{s-2}$,
this follows immediately from (Progress) for~$\zeta$.

If~$R_{s-1} \neq \0$, then~$s' = s$ and~$R_{s-1}' \subseteq R_s$ so~$R_{s-1}'
\cap (O^{s-2})' = \0$
and~$b_s' = b \notin (O^{s-2})'$ follows again immediately from~$(O^{s-2})' =
O^{s-2}$ and 
(Progress) for~$\zeta$. So it remains to show that 
$b \notin O_{s-1}' = \tilde{O}_{s-1} = \Reach_{\graphG - W_{s-1}}(R_{s-1}\setminus
\{b\})$.
As the robbers play according to an isolating strategy, 
$b \notin \Reach_{\graphG - U}(R_{s-1} \setminus \{b\})$. Assume that
$b \in \Reach_{\graphG - W_{s-1}}(R_{s-1} \setminus \{b\})$. 
Then due to Lemma~\ref{lemma_convenience}, 
$b \in \Reach_{\graphG - W^{s-1}}(R_{s-1} \setminus \{b\})
\subseteq \Reach_{\graphG - U^{s-1}}(R_{s-1} \setminus \{b\})$. Moreover, by
Corollary~\ref{corollary_lemma_le},
$\Reach_{\graphG - U^{s-1}}(R_{s-1} \setminus \{b\})
= \Reach_{\graphG - U}(R_{s-1} \setminus \{b\})$, which is a contradiction.
In Case~II, (Progress) for~$\zeta'$ follows easily from (Progress) for~$\zeta$
using the definition of the memory update.

Now consider the situation after the robber move.
In Case~2, (Progress) holds by the construction of the sets~$\tilde{R}_i = R_i'$
for
$i = 1, \ldots, s$. Moreover, in Case~1,~$R_i' \cap (O^{i-1})' = \0$ holds
for~$i = 1, \ldots, s'-1$ by the construction of the sets~$R_i'$ as well and
$b \notin (O^{s'-2})' = O^{s-1}$ holds by the construction of~$\tilde{R}_s$. 

It remains to show that
$b \notin O_{s'-1}' = \tilde{O}_s = \Reach_{\graphG - W_s}(\tilde{R}_s \setminus
\{b\})$. 
As the robber plays according to an isolating strategy, we have
$b \notin \Reach_{\graphG - U}(\tilde{R}_s \setminus \{b\})$. Assume that
$b \in \Reach_{\graphG - W_s}(\tilde{R}_s \setminus \{b\})$. Then as~$W_s =
W_{s'-1}'$ and 
$\tilde{R}_s \setminus \{b\} = R_{s'-1}'$, Lemma~\ref{lemma_convenience}
for the memory state~$\zeta'$ after the robber move
yields~$b \in \Reach_{\graphG - (W^{s'-1})'}(R_{s'-1}') = \Reach_{\graphG -
W^s}(\tilde{R}_s \setminus \{b\})
\subseteq \Reach_{\graphG - U^s}(\tilde{R}_s \setminus \{b\})$. Moreover,~$U^s =
U$, so 
$b \in \Reach_{\graphG - U}(\tilde{R}_s \setminus \{b\})$, which is a
contradiction.
\end{proof}

We conclude the proof of Theorem~\ref{theorem_main_tech} with the 
following lemma, whose proof uses (Progress) to show that all robbers
are finally captured in any play consistent with~$\rf$.

\begin{lemma}\label{lemma_impimf_winning}
$\rf$ is winning.
\end{lemma}
\begin{proof}
First observe that every cop that is placed on the graph according to the
longest history restricts the set of vertices reachable for the robber
on~$b_s$ because~$f$ is active.

Assume that there is a play
\[\pi = \perpcdot (U_0,R_0)\cdot (U_0,U_1,R_0) \cdot (U_1,R_1) \ldots\]
consistent with~$\rf$ and a position~$(U_j,R_j)$ of~$\pi$ after which 
the set of vertices reachable for the robber in the longest history~$\rho_s$
remains constant. (Due to the 
monotonicity of~$\rf$, it never becomes smaller.) As the robbers play
according to a prudent strategy,~$R_i$ also remains constant.
Let~$b(i)$ be vertex~$b_s$ stored in the memory after move number~$i$.
Then~$\Reach_{\graphG - U_l}(b(j)) = \Reach_{\graphG - U_{l+1}}(b(l+1))$, 
for~$l\ge j$. As the robber strategy is prudent, it follows that~$b(j) = b(l)$,
\ie the robber does not change his vertex after move number~$j$.

It suffices to prove that Case~II.2 appears infinitely often.
If it does, we place new cops on~$f(W_s,b_s)\setminus O^{s-1}$
again and again.
As~$\Reach_{\graphG - U_l}(b(l)) = \Reach_{\graphG - U_{l+1}}(b(l+1))$, for
all~$l\ge j$,
it follows that~$\rf$ never places cops into~$\Reach_{\graphG - U(l)}(b(l))$ 
and thus~$U_l = U_{l+1}$, by the definition of~$U_l$.
Since~$\rf$ places cops according to~$f$, 
it prescribes to place cops only in~$O^{s-1}$.
Therefore,~$b_s$ is never occupied by any cop according to~$f$  
due to the invariant (Progress). Hence,~$f$ is not winning,
which contradicts our assumption.

Assume that after some position, Case~II.2 does not appear.
Then Case~I or Case~II.1 appear infinitely often. In both cases,~$s$ 
does not increase.

In Case I, if~$R_{s-1} = \0$, then
the number~$s$ of histories in~$\zeta$ decreases. 
If~$R_{s-1} \neq \0$, then~$|R_{s-1}|$ decreases.

In Case II.1, histories that are shorter than~$\rho_s$ 
are extended or deleted (which decreases~$s$), if they reach the next play
prefix. 
The length of the longest history in~$\zeta$ is an upper bound for the
growth 
of their lengths. As the robbers do not change their placement,~$|R_{s-1}|$ 
will never increase again.
Together, either~$s$ or~$|R_{s-1}|$ decrease, so Cases~I and~II.1 can appear
only finitely many times. It follows that we have Case~II.2 infinitely many
times,
but that contradicts our assumption.
\end{proof}

This finishes the proof of Theorem~\ref{theorem_main_tech}.

\subsection{Robbers hierarchy, imperfect information and \dpathw}\label{subsec_hierarchy}

In this section we extend the results from~\cite{RicherbyThi09} about the 
dependence of cop number on the number of robbers to our setting.
For the same graph~$\graphG$, increasing the number of robbers induces a 
hierarchy
of cop numbers that are needed to capture the robbers. 
It is clear that less robbers do not demand more cops.
Furthermore, one robber corresponds to the \dagw game and~$|\graphG|$ robbers
to the \dpathw game, hence we have the following scheme:
\[\dw(\graphG) = \dw_1(\graphG) \le \dw_2(\graphG) \le \ldots \le
\dw_{|\graphG|}(\graphG) = \dpw(\graphG)\]
where~$n$ is the number of vertices of~$\graphG$.
In general, \ie on some graphs, this hierarchy does not collapse, because
\pathw is not bounded in \treew. We give explicit lower bounds for the stages.
In a sense, \dagw can be approximated by a refinement of \dpathw,
but there are infinitely many stages of approximation.
This result is analogous to similar results in~\cite{RicherbyThi09}
and in~\cite{FominFraNis09}.

\begin{theorem}\label{thm_impinf_strict_hierarchy}
For every~$k > 0$, there is a class~$\classG^k$ of graphs such that, for
all~$\graphG\in \classG^k$,
we have~$\dw_1(\graphG) = 2\cdot k$ and, for all~$r>0$, there
exists~$\graphG^k_r\in\classG^k$ with 
\begin{enumerate}[(1)]
  \item~$\dpw(\graphG^k_r) = k\cdot(r+1)$, and
  \item\label{cond_dw_i_unbounded} for
all~$i\in\{1,\dots,r\}$,~$\dw_i(\graphG^k_r) \ge \frac{i\cdot (k-1)}{2}$.
\end{enumerate}
\end{theorem}
\begin{proof}

Class~$\classG^k$ consists of graphs~$\graphG^k_r$, for each~$r>0$.
Every~$\graphG^k_r$ 
is the lexicographic product~$\graphT_r\oplus \graphK_k$ of the full undirected 
tree~$\graphT_r$ with branching degree~$\lceil \frac{r}{2}\rceil+2$ and of 
height~$r+1$, with the~$k$-clique~$\graphK_k$. 
In other words,~$\graphG^k_r$ is~$\graphT_r$ where every vertex~$v$ is replaced by
a~$k$-clique~$K(v)$ and if~$(v,w)$ is an edge of~$\graphT_r$, then all
pairs~$(v',w')$
with~$v'\in K(v)$ and~$w'\in K(w)$ are edges of~$\graphG^k_r$.

It is clear that~$\dw_1(\graphG^k_r)$ is~$2\cdot k$:
the cops play as on~$\graphT_r$ occupying~$K(v)$  instead of single tree
vertex~$v$
and leaving~$K(v)$ if~$v$ is left.\footnote{The idea to use the lexicographic
product and 
of the proof is due to~\cite{HunterPhD}.}
We have to show that 
$\dpw(\graphG^k_r) = k(r+1)$ and that~$\dw_i(\graphG^k_r) \ge
\frac{i\cdot(k-1)}{2}$.

We start with \dpathw. A similar proof can be found, for example,
in~\cite{Bodlaender98}.
Note that the branching degree of all~$T_r$ is at least~$3$.
Let us see that the statement follows from~$\dpw(\graphT_r)=r+1$.
First, as for \dagw above, we have~$\dpw(\graphG^k_r) \le k\cdot(r+1)$. 
The statement of the other direction follows from the fact that 
it makes no sense for the cops to occupy only a part of a~$k$-clique.
We formulate that statement as a small lemma.

\begin{lemma}
Every winning strategy~$f$ for~$k(r+1)$ cops can be turned into 
a winning strategy~$f'$ for~$k(r+1)$ cops that always prescribes to occupy 
whole~$k$-cliques.
\end{lemma}
\begin{proof}
Strategy~$f'$ is as follows.
If~$f$ prescribes to occupy only a part of a clique, then~$f'$
does not place any cops in the clique, otherwise~$f$ and~$f'$ are the
same.
Assume that~$f'$ is not winning. Then there is a cop
move~$(U,R)\to(U,U',R)$
such that a path~$P$ from~$R$ to~$U\setminus U'$ exists in~$\graphG - (U \cap U')$.
Consider a path~$P'$ that is as~$P$, but for vertices~$v$ occupied by cops,
it contains a vertex~$w\in K(v)$ that is cop-\mbox{}free. It is clear that such
a vertex~$w$
always exists. Then~$P'$ is an evidence that~$f$ is not monotone, which is
a contradiction
to our assumption.
\end{proof}

We prove~$\dpw(\graphT_r)=r+1$ by induction on~$r$. 
The case~$r=1$ is trivial. If~$r+1$ cops win on~$\graphT_r$,
then~$r+2$ cops win on~$T_{r+1}$ by placing a cop on the root and applying the
strategy for~$r+1$ cops
from the induction hypothesis for every subtree.

The other direction (that~$\dpw(\graphT_r)\ge r+1$) is also proven by induction
on~$r$. The induction base
is clear. Assume that~$\dpw(\graphT_r) \ge r+1$. In~$\graphT_{r+1}$, let the
direct successors of the root be
$v_1,\dots,v_m$ (recall that~$m\ge 3$). All subtrees~$\graphT^i$ rooted
at~$v_i$, for~$i\in \{1,\dots,m\}$,
must be decontaminated (\ie the robber must be expelled from~$\graphT^i$) 
and~$r+1$ cops are needed for that. Assume \wLOG
that~$\graphT^1$ is the first and~$\graphT^2$ is the second decontaminated 
subtree. In some position all $r+1$ cops are in~$T^2$.
However, there is a path from~$T^m$ via the root of the whole tree to~$\graphT^1$.
Thus~$T^1$ 
becomes recontaminated, which contradicts the monotonicity of 
\dpathw~\cite{Hunter06}.

It remains to show that~$k\cdot i$ robbers win against~$\frac{i\cdot(k-1)}{2}$
cops on~$\graphG^k_r$.
We show only that~$i$ robbers win against~$\lfloor \frac{i}{2}\rfloor$ cops
on~$\graphT_r$, the result 
with factor~$k$ follows as above. As in the proof of
Theorem~\ref{thm_OffhandedCops}, we can 
assume that the cops play top-down because the tree has a high branching degree.

The winning strategy for robbers is to tie every cop. A cop is
\emph{tied}
if there is a cop-\mbox{}free path
from a robber to the cop. When a cop is placed on a vertex~$v$, the robbers
occupy two subtrees of~$v$.
As there are at least two robbers for each cop, 
this is always possible. A cop is untied only if two other cops in both 
subtrees chosen by the robbers become tied, so at every tree level at least one 
more cop becomes tied. At the latest when a cop reaches level~$\lfloor
\frac{i\cdot(k-1)}{2}\rfloor$, all cops are tied.
\end{proof}

\section{Discussion and future work}
We analyzed the connection between imperfect information in parity
games and structural complexity of game graphs. If the amount of
imperfect information is unbounded, restricting structural complexity of game
graphs does not lead to lower computational complexity of the strategy
problem. For the case of bounded imperfect information we showed that
some graph complexity measures have unbounded values when performing
the powerset construction, and some are still bounded. As side
effects of our proofs we showed that, first, monotonicity of \dagw is not
necessary for an efficient solution of the strategy problem for
perfect information parity games, and, second, that introducing new
robbers demands only linearly more cops to capture them. We believe
that those results are also of independent relevance.

To complete the picture, it would be interesting to prove that \kellyw and
directed \treew also remain bounded after performing the powerset
construction. For directed \treew it is not known whether perfect
information parity games can be solved in \ptime, so a bound would not
immediately imply an efficient solution of parity games with imperfect
information. It would be also worth attention to analyze which other
variants of the graph searching game with multiple robbers make sense
and what are the differences between them, our version and the games
from~\cite{RicherbyThi09}.

\subsection{Acknowledgments}

We thank {\L}ukasz Kaiser for many inspiring discussions, Tsvetelina
Yonova-{}Karbe and Sebastian Siebertz for the proof reading.

\newcommand{\etalchar}[1]{$^{#1}$}

\end{document}